\documentclass[
	letterpaper,
	USenglish,
	11pt,
	fleqn
]{article}

\usepackage{babel}
\usepackage[T1]{fontenc}
\usepackage{microtype}
\usepackage[svgnames]{xcolor}
\usepackage[
	pdfusetitle,
	unicode=true,
	colorlinks=true,
	linkcolor=NavyBlue,
	citecolor=NavyBlue,
	urlcolor=NavyBlue,
	]{hyperref}
\usepackage[
	margin=1in
]{geometry}
\usepackage{amsmath,amsthm}
\usepackage[
	export
	]{adjustbox}
\usepackage{newtxtext}
\usepackage{newtxmath}
\usepackage{csquotes}
\usepackage{enumitem}
\usepackage[
	capitalize
]{cleveref}
\usepackage{caption}
\usepackage{subcaption}
\usepackage{authblk}
\usepackage{xargs}

\usepackage{mathtools}
\usepackage{thmtools, thm-restate}
\usepackage{bm}
\usepackage[
]{todonotes}

\usepackage{tikz}
\usetikzlibrary{arrows, arrows.meta, calc, graphs, math, shapes.arrows, shapes.geometric}

\MakeOuterQuote{"}

\newlength{\myA}%

\crefname{rule}{rule}{rules}
\Crefname{rule}{Rule}{Rules}
\newenvironment{rules}[1][]{\begin{enumerate}[#1]\crefalias{enumi}{rule}}{\end{enumerate}}

\crefname{condition}{condition}{conditions}
\Crefname{condition}{Condition}{Conditions}
\newenvironment{conditions}[1][]{\begin{enumerate}[#1]\crefalias{enumi}{condition}}{\end{enumerate}}

\crefname{statement}{statement}{statements}
\Crefname{statement}{Statement}{Statements}
\newenvironment{statements}[1][]{\begin{enumerate}[#1]\crefalias{enumi}{statement}}{\end{enumerate}}

\crefname{claim}{Claim}{Claims}
\Crefname{claim}{Claim}{Claims}

\declaretheorem[
	name=Definition,
	numberwithin=section,
	style=definition,
	qed={\ensuremath{\lrcorner}}
]{definition}
\let\endDefinition\qedhere
	
\declaretheorem[
	name=Example, 
	sibling=definition,
	style=definition,
	qed={\ensuremath{\lrcorner}}
]{example}
\let\endExample\qedhere
		
\declaretheorem[
	name=Remark, 
	sibling=definition,
	style=definition
]{remark}

\declaretheorem[
	name=Lemma, 
	sibling=definition
]{lemma}

\declaretheorem[
	name=Fact, 
	sibling=definition
]{fact}

\declaretheorem[
	name=Claim, 
	sibling=definition
]{claim}

\declaretheorem[
	name=Claim, 
	numbered=no
]{claim*}

\declaretheorem[
	name=Proposition, 
	sibling=definition
]{proposition}

\declaretheorem[
	name=Theorem,
	sibling=definition,
]{theorem}

\declaretheorem[
	name=Theorem
]{alphtheorem}
\definecolor{myred}{HTML}{882255}
\definecolor{myblue}{HTML}{332288}
\definecolor{mygreen}{HTML}{117733}

\newcommand*{\RelSymE}{\textcolor{myred}{\textbf{E}}}
\newcommand*{\RelSymF}{\textcolor{mygreen}{\textbf{F}}}
\newcommand*{\RelSymR}{\textcolor{myblue}{\textbf{R}}}

\tikzset{
    colorE/.style={draw=myred},
    colorF/.style={draw=mygreen},
    colorR/.style={draw=myblue},
    slice/.style={},
	RelE/.style={thick, -Stealth, colorE,fill=myred},
	RelF/.style={thick, densely dashed, -Stealth, colorF,fill=mygreen},
	RelR/.style={thick, dashed, colorR},
	mynode/.style={inner sep=2pt},
	myvertex/.style={rectangle, draw, very thick},
    tpFor/.style={-latex},
    tpBack/.style={latex-},
    edgeAnnotation/.style={draw=none, font={\tiny}},
    myLeft/.style={},
    myRight/.style={bend right=10},
    cr-node/.style={circle, draw=black, inner sep=1pt, minimum size=10pt},
    dreieck/.style={regular polygon, regular polygon sides=3},
    viereck/.style={regular polygon, regular polygon sides=4},
    fuenfeck/.style={regular polygon, regular polygon sides=5},
    sechseck/.style={regular polygon, regular polygon sides=6},
    siebeneck/.style={regular polygon, regular polygon sides=7},
    achteck/.style={regular polygon, regular polygon sides=8},
    neuneck/.style={regular polygon, regular polygon sides=9},
    colourA/.style={cr-node, sechseck, text=white, fill=Navy},
    colourB/.style={cr-node, viereck, text=white, fill=Purple},
    colourC/.style={cr-node, fill=DarkKhaki},
    colourD/.style={cr-node, dreieck, fill=Goldenrod},
    colourE/.style={cr-node, fuenfeck, fill=Turquoise},
    colourF/.style={cr-node, siebeneck, fill=Lime},
    colourG/.style={cr-node, sechseck, fill=Crimson, text=white},
    colourH/.style={cr-node, achteck, fill=HotPink},
    colourI/.style={cr-node, neuneck, fill=Salmon},
    colourJ/.style={cr-node, fill=DarkGreen, text=white},
    colourK/.style={cr-node, viereck, fill=SaddleBrown, text=white},
    colourL/.style={cr-node, dreieck, fill=SlateBlue},
    colourM/.style={cr-node, fuenfeck, fill=GreenYellow},
    colourN/.style={cr-node, siebeneck, fill=Plum},
    colourO/.style={cr-node, dreieck, fill=PeachPuff},
    colourP/.style={cr-node, fuenfeck, fill=Teal},
    colourQ/.style={cr-node, siebeneck, fill=SkyBlue},
    colourR/.style={cr-node, sechseck, fill=Grey},
    colourS/.style={cr-node, achteck, fill=DarkOrange},
    colourT/.style={cr-node, neuneck, fill=Tan},
    colourU/.style={cr-node, fill=Olive},
    colourV/.style={cr-node, viereck, fill=IndianRed},
}

\newcommand{\convexpath}[2]{
[
    create hullnodes/.code={
        \global\edef\namelist{#1}
        \foreach [count=\counter] \nodename in \namelist {
            \global\edef\numberofnodes{\counter}
            \node at (\nodename) [draw=none,name=hullnode\counter] {};
        }
        \node at (hullnode\numberofnodes) [name=hullnode0,draw=none] {};
        \pgfmathtruncatemacro\lastnumber{\numberofnodes+1}
        \node at (hullnode1) [name=hullnode\lastnumber,draw=none] {};
    },
    create hullnodes
]
($(hullnode1)!#2!-90:(hullnode0)$)
\foreach [
    evaluate=\currentnode as \previousnode using \currentnode-1,
    evaluate=\currentnode as \nextnode using \currentnode+1
    ] \currentnode in {1,...,\numberofnodes} {
-- ($(hullnode\currentnode)!#2!-90:(hullnode\previousnode)$)
  let \p1 = ($(hullnode\currentnode)!#2!-90:(hullnode\previousnode) - (hullnode\currentnode)$),
    \n1 = {atan2(\y1,\x1)},
    \p2 = ($(hullnode\currentnode)!#2!90:(hullnode\nextnode) - (hullnode\currentnode)$),
    \n2 = {atan2(\y2,\x2)},
    \n{delta} = {-Mod(\n1-\n2,360)}
  in
    {arc [start angle=\n1, delta angle=\n{delta}, radius=#2]}
}
-- cycle
}

\title{Color Refinement for Relational Structures}
\date{}

\author{Benjamin Scheidt\\\href{mailto:benjamin.scheidt@hu-berlin.de}{benjamin.scheidt@hu-berlin.de} \and Nicole Schweikardt\\\href{mailto:schweikn@informatik.hu-berlin.de}{schweikn@informatik.hu-berlin.de}}

\affil{Humboldt-Universität zu Berlin, Germany}

\newcommand*{\inlineTodo}[2]{\todo[inline, caption={}, color=#1]{#2}}

\newcommand*{\nat}{\mathbb{N}}
\newcommand*{\natpos}{\nat_{\geq{} 1}}

\newcommand*{\isdef}{\coloneqq}
\newcommand*{\defis}{\eqqcolon}
\newcommand*{\isomorphic}{\cong}

\newcommand{\img}{\operatorname{img}}

\newcommand*{\hati}{{\hat{\imath}}}

\newcommand*{\card}[1]{\lvert{} {#1} \rvert}
\newcommand*{\size}[1]{\lVert{} {#1} \rVert}

\renewcommand*{\hom}{\operatorname{hom}}
\newcommand*{\Hom}{\operatorname{Hom}}

\newcommand*{\bigO}{\mathcal{O}}

\let\oldphi\phi%
\let\phi\varphi%
\let\varphi\oldphi%

\let\geq\geqslant%
\let\leq\leqslant%

\let\emptyset\varnothing%

\newcommand\varcal[1]{\text{\usefont{U}{eus}{m}{n}#1}}

\newcommandx*{\set}[3][1={},3={\,}]{\ensuremath{#1\{#3 {#2} #3#1\}}}
\renewcommand*{\mid}{\;:\;}

\newcommandx*{\mset}[3][1={},3={\,}]{\ensuremath{
	#1\{\!\!#1\{#3 {#2} #3#1\}\!\!#1\}
}}
\newcommand*{\mult}[1]{\operatorname{mult}_{#1}}

\newcommand*{\union}{\cup}
\newcommand*{\intersect}{\cap}
\newcommand*{\bigunion}{\bigcup}

\newcommand*{\disunion}{\mathrel{\dot{\cup}}}
\newcommand*{\bigdisunion}{\dot{\bigcup}}

\newcommand*{\compl}[1]{\overline{#1}}

\newcommand*{\tup}[1]{{\bm{#1}}}
\newcommand*{\tset}{\operatorname{set}}

\newcommand*{\at}{\tup{a}}
\newcommand*{\bt}{\tup{b}}
\newcommand*{\ct}{\tup{c}}

\newcommand*{\sig}[1][\sigma]{#1}
\newcommand*{\ksig}[1][\sigma]{{\sig[#1]\vert_k}}
\newcommand*{\ar}{\operatorname{ar}}
\newcommand*{\str}[1]{\mathcal{#1}}
\newcommand*{\class}[1]{\mathsf{#1}}

\newcommand*{\A}{\str{A}}
\newcommand*{\B}{\str{B}}
\newcommand*{\C}{\varcal{C}}

\newcommand*{\tA}{\tup{A}}
\newcommand*{\tB}{\tup{B}}
\newcommand*{\tC}{\tup{C}}

\newcommand*{\cohesion}{\ensuremath{\Gamma}}
\newcommand*{\density}{\cohesion}

\newcommand*{\V}{V}
\newcommand*{\E}{E}

\newcommand*{\repsig}[1][\sigma]{\widehat{\sig[#1]}}
\newcommand*{\jtrep}[2]{\str{G}_{#1}^{#2}}
\newcommand*{\grep}[1]{\str{G}_{#1}}
\newcommand*{\vgrep}[1]{\str{H}_{#1}}
\newcommand{\slices}{\mathcal{S}}

\newcommand*{\slice}{\tup{s}}

\newcommand*{\atp}{\mathsf{atp}}
\newcommand*{\atype}{\rho}

\newcommand*{\stp}{\mathsf{stp}}
\newcommand*{\stype}{\tau}
\newcommand*{\sTypes}[1][\sigma]{\class{STP}_{#1}}
\newcommand*{\links}{l}
\newcommand*{\rechts}{r}
\newcommand*{\lpart}[1]{\mathsf{\links{}p}_{#1}}
\newcommand*{\rpart}[1]{\mathsf{\rechts{}p}_{#1}}
\newcommand*{\cpart}[1]{\mathsf{cp}_{#1}}
\newcommand*{\rc}{\operatorname{\varrho}}
\newcommand*{\col}[2]{\rc_{#1}(#2)} %
\newcommand*{\cols}[2]{\class{RC}_{#1}(#2)} %
\newcommand*{\ccount}[2]{\mult{#1}(#2)} %
\newcommand*{\Nset}[2]{N_{#1}^{#2}}

\newcommand*{\classicalCR}{\operatorname{\gamma}}
\newcommand*{\gcr}[2]{\classicalCR_{#1}(#2)} %

\newcommand*{\Logic}[1]{\mathsf{#1}}
\newcommand*{\CFO}{\Logic{C}}
\newcommand*{\GF}[1]{{\mathsf{GF}}(#1)} %
\newcommand*{\GFC}{\GF{\CFO}}

\newcommand*{\VAR}{\class{Var}}
\newcommand*{\var}{\mathtt{v}}
\newcommand*{\vartup}{\tup{\var}}

\newcommand*{\xvar}{x}
\newcommand*{\xt}{\tup{\xvar}}

\newcommand*{\LogGuard}{\Delta}

\newcommand*{\gd}{\operatorname{gd}}
\newcommand*{\free}{\operatorname{free}}
\newcommand*{\vset}{\operatorname{vars}}

\newcommand*{\Int}{\mathcal{I}}
\newcommand*{\assignment}{\alpha}
\newcommand*{\Isubst}[2]{\Int\!\tfrac{#1}{#2}}

\newcommand*{\bigland}{\bigwedge}

\newcommand*{\existsgeq}[1]{\exists^{\scalebox{0.6}{$\geq$}#1}}
\newcommand*{\existseq}[1]{\exists^{\scalebox{0.6}{$=$}#1}}
\newcommand*{\lequal}{\kern1pt{=}\kern1pt}
\newcommand*{\qsep}{\mathbin{.}}

\newcommand*{\sigEx}{\ensuremath{{\sigma_{1}}}}
\newcommand*{\SpExa}{\ensuremath{{\A_{1}}}}
\newcommand*{\SpExb}{\ensuremath{{\B_{1}}}}

\newcommand*{\AdsigEx}{\ensuremath{\sigma_{2}}}
\newcommand*{\AdExa}{\ensuremath{\A_2}}
\newcommand*{\AdExb}{\ensuremath{\B_2}} %

\renewcommand{\epsilon}{\varepsilon}

\newcommand{\nc}[1]{\newcommand{#1}}
\newcommand{\rnc}[1]{\renewcommand{#1}}

\nc{\myparagraph}[1]{\noindent\textbf{#1}}

\nc{\deff}{\isdef}
\nc{\ffed}{\defis}

\nc{\setsize}[1]{\ensuremath{|#1|}}
\nc{\Setsize}[1]{\ensuremath{\big|#1\big|}}
\nc{\Set}[1]{\ensuremath{\big\{#1\big\}}}
\nc{\setc}[2]{\set{#1 \, : \, #2}}
\nc{\Setc}[2]{\Set{#1 \, : \, #2}}

\nc{\aufgerundet}[1]{\ensuremath{\lceil #1 \rceil}}
\nc{\abgerundet}[1]{\ensuremath{\lfloor #1 \rfloor}}

\nc{\dcup}{\disunion}
\nc{\dCup}{\ensuremath{\dot\bigcup}}

\nc{\ov}[1]{\ensuremath{\overline{#1}}}

\nc{\NN}{\ensuremath{\mathbb{N}}}
\nc{\NNpos}{\ensuremath{\NN_{\scriptscriptstyle\geq 1}}}
\nc{\RR}{\ensuremath{\mathbb{R}}}
\nc{\RRpos}{\ensuremath{\RR_{\scriptscriptstyle\geq 0}}}
\nc{\QQ}{\ensuremath{\mathbb{Q}}}
\nc{\QQpos}{\ensuremath{\QQ_{\scriptscriptstyle\geq 0}}}

\nc{\und}{\ensuremath{\wedge}}
\nc{\Und}{\ensuremath{\bigwedge}}
\nc{\oder}{\ensuremath{\vee}}
\nc{\Oder}{\ensuremath{\bigvee}}
\nc{\nicht}{\ensuremath{\neg}}
\nc{\impl}{\ensuremath{\to}}
\nc{\gdw}{\ensuremath{\leftrightarrow}}

\nc{\Semijoin}{\ensuremath{\ltimes}}

\nc{\quant}{\ensuremath{\textrm{\upshape quant}}}

\nc{\isom}{\ensuremath{\isomorphic}}

\nc{\eval}[2]{\ensuremath{#1(#2)}}
\nc{\semantik}[1]{\ensuremath{\left\llbracket#1\right\rrbracket}}
\nc{\sem}[1]{{\semantik{#1}}}

\nc{\mypot}[1]{\ensuremath{2^{#1}}}
\nc{\mypoteq}[2]{\ensuremath{\binom{#1}{#2}}}
\nc{\mypotk}[1]{\mypoteq{#1}{k}}
\nc{\mypotedges}[1]{\mypoteq{#1}{2}}

\nc{\D}{\str{D}}
\nc{\tD}{\tup{D}}
\nc{\dt}{\tup{d}}
\nc{\vt}{\tup{v}}

\nc{\dirT}{\ensuremath{\vec{T}}}

\nc{\myatype}{\ensuremath{\textsf{at}}}
\nc{\myrho}{\ensuremath{\myatype}}

\nc{\mystype}{\ensuremath{\textsf{st}}}
\nc{\mytau}{\ensuremath{\mystype}}

\nc{\onesig}[1][\sigma]{{\sig[#1]\vert_{1}}}
\nc{\twosig}[1][\sigma]{{\sig[#1]\vert_{2}}}

\nc{\myC}{\ensuremath{\class{C}}}

\rnc{\vgrep}[1]{\ensuremath{\str{H}_{#1}}}
\rnc{\rc}{\operatorname{\varrho}}

\nc{\myc}{\ensuremath{f}}
\nc{\myColors}{\ensuremath{Z}}

\nc{\raus}[1]{\inlineTodo{gray!40}{\textbf{N\@: rausgeworfen, weil es
      (hier) nicht
      mehr benötigt wird: } #1}}

\nc{\PotGeq}[2]{\ensuremath{\mathcal{P}_{\geq #1}(#2)}}
\nc{\PotEq}[2]{\ensuremath{\mathcal{P}_{#1}(#2)}}
\nc{\Neighbors}[1]{\ensuremath{N(#1)}}
\nc{\NeighborsEq}[2]{\ensuremath{N_{#1}(#2)}}
\nc{\NeighborsGeq}[2]{\ensuremath{N_{\geq #1}(#2)}}
 
\begin{document}
	\maketitle

	\begin{abstract}
Color Refinement, also known as Naive Vertex Classification, is a classical method to distinguish graphs by iteratively computing a coloring of their vertices.
While it is mainly used as an imperfect way to test for isomorphism, the algorithm permeated many other, seemingly unrelated, areas of computer science.
The method is algorithmically simple, and it has a well-understood distinguishing power:
It is logically characterized by Cai, Fürer and Immerman (1992), who showed that it distinguishes precisely those graphs that can be distinguished by a sentence of first-order logic with counting quantifiers and only two variables.
A combinatorial characterization is given by Dvo\v{r}ák (2010), who shows that it distinguishes precisely those graphs that can be distinguished by the number of homomorphisms from some tree.

In this paper, we introduce Relational Color Refinement (RCR, for short), a generalization of the Color Refinement method from graphs to arbitrary relational structures, whose distinguishing power admits the equivalent combinatorial and logical characterizations as Color Refinement has on graphs:
We show that RCR distinguishes precisely those structures that can be distinguished by the number of homomorphisms from an acyclic relational structure.
Further, we show that RCR distinguishes precisely those structures that can be distinguished by a sentence of the guarded fragment of first-order logic with counting quantifiers.

Additionally, we show that for every fixed finite relational signature, RCR can be implemented to run on structures of that signature in time $\bigO(N\cdot \log N)$, where $N$ denotes the number of tuples present in the structure.  	\end{abstract}

	\section{Introduction}\label{sec:introduction}
\myparagraph{Color Refinement} (CR, for short) constitutes a simple procedure to classify the vertices of a graph $G$; it is well-understood and widely used in many areas of computer science.
The idea is a simple iteration: Given a coloring $\classicalCR$ of the vertices $\V(G)$ of a graph $G$, one computes a new coloring $\classicalCR'$ of $\V(G)$ following a certain procedure.
The new coloring $\classicalCR'$ is then used to compute another coloring $\classicalCR''$ following the same procedure, and so on, until the coloring stabilizes, i.e., the partitioning of $\V(G)$ induced by the new coloring is the same as the one induced by the previous coloring.
The procedure to compute the new coloring is very simple: Two vertices $u,v$ shall get different colors if they already have different colors, or they have a different number of neighbors of some color.
Otherwise, they receive the same color.
To start the iteration, one either uses the colors of the vertices (if $G$ is a colored graph), or the uniform coloring that assigns every vertex the same color.
This approach is sometimes also called \enquote{naive vertex classification} or the \enquote{1-dimensional Weisfeiler-Leman algorithm}.
CR is often formalized using multisets in the following way, see e.g.~\cite[Chapter 3]{Grohe2017}.
Considering an uncolored, undirected, simple graph $G$, we start with $\gcr{0}{v} = 0$ for all $v \in \V(G)$, and for all $i \in \nat$, we let
\begin{equation*}
	\gcr{i+1}{v} \;\isdef\; \bigl(
		\gcr{i}{v},\,
		\mset{ \gcr{i}{w} \mid \set{ v,w } \in \E(G)  }
	\bigr)\;.
\end{equation*}
Note that this formalizes the procedure described above: If $\gcr{i}{u} \neq \gcr{i}{v}$, then $\gcr{i+1}{u} \neq \gcr{i+1}{v}$; and if the number of neighbors of some color disagrees for $u$ and $v$, then $\mset{ \gcr{i}{w} \mid \set{ u,w } \in \E(G)  }[] \neq \mset{ \gcr{i}{w} \mid \set{ v,w } \in \E(G) }[]$, hence, $\gcr{i+1}{u} \neq \gcr{i+1}{v}$.
This formalization has the additional advantage that the colorings $\classicalCR_i$ assign \emph{canonical colors}, i.e., the colors themselves do not depend on the graph $G$.
Berkholz, Bonsma and Grohe~\cite{Berkholz2017} showed that such a canonical coloring can be computed in time $\bigO((n+m) \cdot \log(n))$, where $n$ denotes the number of vertices, and $m$ the number of edges.%
\medskip

\myparagraph{Applications.}
An obvious application of Color Refinement is to test for graph isomorphism: If there is an $i \in \nat$ and some $c$ such that $\card{\set{ v \in \V(G) \mid \gcr{i}{v} = c }} \neq \card{\set{ v \in \V(H) \mid \gcr{i}{v} = c }}$ (we say Color Refinement distinguishes $G$ and $H$ if this is true), then $G$ and $H$ cannot be isomorphic.
However, this test is not perfect, since there exist non-isomorphic pairs of graphs that are not distinguished by Color Refinement.
Nevertheless, it is a common subroutine in practical isomorphism testers and even plays a part in Babai's seminal result that graph isomorphism is solvable in quasi-polynomial time~\cite{Babai2016}.
In recent years, the classification of \enquote{similar vertices} that Color Refinement establishes, has been applied to other problems as well: It was used in~\cite{Grohe2014a,Grohe2021} to reduce the cost of solving linear programs, in~\cite{Riveros2024} it was used to speed up the evaluation of binary acyclic conjunctive queries, and in the area of machine learning, it is used as a graph kernel~\cite{Shervashidze2011} and was proven to be equivalent to so-called Graph Neural Networks (GNNs)~\cite{Xu2019,Grohe2020a}.
\medskip

\myparagraph{The power of Color Refinement} is well-understood~\cite{Arvind2017,Berkholz2017,Cai1992,Dvorak2010,Immerman1990,Kiefer2021} --- consult e.g.~\cite{Grohe2017,Kiefer2020} for an overview.
The key results on the distinguishing power of CR are a logical characterization due to Immerman and Lander~\cite{Immerman1990} and Cai, Fürer, Immerman~\cite{Cai1992}, and a combinatorial characterization w.r.t.\ the concept of \enquote{homomorphism indistinguishability} due to Dvo\v{r}ák~\cite{Dvorak2010} and Dell, Grohe and Rattan~\cite{Dell2018}:
In~\cite{Cai1992,Immerman1990} it is shown that CR distinguishes $G$ and $H$ if, and only if, there is a sentence $\phi$ of first-order logic with counting quantifiers with at most 2 variables ($\CFO^2$, for short) such that $G \models \phi$ and $H \not\models \phi$.
In~\cite{Dvorak2010,Dell2018} it is shown that CR distinguishes $G$ and $H$ if, and only if, $G$ and $H$ are distinguished by the number of homomorphisms from some tree $T$ into $G$ and $H$.
If no such $T$ exists, one says that $G$ and $H$ are \emph{homomorphism indistinguishable} over the class of trees.
This result sparked active research in recent years exploring the concept of homomorphism indistinguishability over various graph classes, see e.g.~\cite{Dawar2021,Fluck2024,Grohe2020,Mancinska2020,Montacute2022,Roberson2022,Seppelt2023,Scheidt2024,Scheidt2023}.
These characterizations can explain the success of the Color Refinement method, and in particular, they give us a hint on why the vertex classification produced by CR is so powerful: Two vertices $u, v \in \V(G)$ get classified as \enquote{similar} by CR if the number of homomorphisms from every rooted tree $T$ into $G$ that map the root to $u$ is equal to the number of homomorphisms from $T$ into $G$ that map the root to $v$.
\medskip

\myparagraph{Contributions.} 
With the success of Color Refinement in mind, it is an obvious question how one could devise a method to color arbitrary finite relational structures, not just graphs.
In particular, we would like a method that admits a combinatorial characterization w.r.t.\ homomorphism counts from the class of acyclic relational structures (for a sensible, broad definition of acyclicity) and a logical characterization for a sensible logic, analogously to the ones CR has.
We propose \emph{Relational Color Refinement} (RCR, for short) as such a method and show that it indeed admits the desired characterizations.
We prove the following two main theorems for every finite relational signature $\sig$:

\begin{restatable*}{alphtheorem}{mainhoms}\label{thm:main-homs}
	For all $\sig$-structures $\A$ and $\B$, the following statements are equivalent.
	\begin{statements}[topsep=0pt, noitemsep]
		\item Relational Color Refinement distinguishes $\A$ and $\B$.
		\item There exists an acyclic and connected $\sig$-structure $\C$ such that $\hom(\C, \A) \neq \hom(\C, \B)$.
	\end{statements}%
\end{restatable*}
\begin{restatable*}{alphtheorem}{mainlogic}\label{thm:main-logic}
	For all $\sig$-structures $\A$ and $\B$, the following statements are equivalent.
	\begin{statements}[topsep=0pt, noitemsep]
		\item Relational Color Refinement distinguishes $\A$ and $\B$.
		\item There exists a sentence $\phi \in \GFC$ such that $\A \models \phi$ and $\B \not\models \phi$.
		\item Spoiler wins the Guarded-Game on $\A, \B$.
	\end{statements}
\end{restatable*}
\noindent
Here, $\GFC$ denotes the \emph{guarded fragment} of the logic $\CFO$ (see~\cref{sec:gfc}),
and the \emph{Guarded-Game} is a particular variant of Ehrenfeucht-Fra\"\i{}ss\'{e} games (see \cref{sec:game}).

Furthermore, we obtain the following additional main theorem, where $\size{\A}$ denotes the number of tuples present in relations of a $\sig$-structure $\A$.

\begin{restatable*}{alphtheorem}{mainruntime}\label{thm:main-runtime}
  Relational Color Refinement (RCR) can be implemented to run in time $\bigO(\size{\A}\cdot\log\size{\A})$ upon input of a $\sig$-structure $\A$.
\end{restatable*}

\Cref{sec:preliminaries} introduces the necessary basic concepts and notation used throughout the paper.
In~\cref{sec:RCR} we first consider "quick and dirty" ways to extend Color Refinement to relational structures and show how their distinguishing power falls short, even for simple structures.
Then, we introduce Relational Color Refinement, discuss its connection to CR, and give a first upper bound on the running time.
\Cref{sec:main-homomorphism} is devoted to the proof of~\cref{thm:main-homs}.
\Cref{sec:main-logic} starts by introducing the logic $\GF{\CFO}$ and the Guarded-Game, followed by the proof of~\cref{thm:main-logic}.
\Cref{sec:main-runtime} is devoted to the proof of~\cref{thm:main-runtime}.
Finally, we conclude in~\cref{sec:conclusion} with a summary and an outlook on future work.
\medskip

\myparagraph{Related work on structures and hypergraphs.}
The work by Dell, Grohe and Rattan~\cite{Dell2018} has been generalized to relational structures by Butti and Dalmau~\cite{Butti2021}.
However, they apply classical CR on the incidence graph of the relational structure and use the weaker notion of Berge-acyclicity that is subsumed by our notion of acyclicity.
There is also related work in this direction on hypergraphs, which are conceptually similar to relational structures.
Böker~\cite{Boeker2019} introduced a variant of Color Refinement that works on hypergraphs, and showed that it distinguishes two hypergraphs if, and only if, there is a Berge-acyclic hypergraph that has a different number of homomorphisms to them.
The connection between the logic $\CFO^2$ and homomorphism counts from trees due to~\cite{Dvorak2010,Dell2018} was generalized to the logic $\Logic{GC}$ (a logic similar to the guarded fragment $\GFC$, but tailored towards hypergraphs) and homomorphism counts from acyclic hypergraphs by Scheidt and Schweikardt~\cite{Scheidt2023}. 
	\section{Preliminaries}\label{sec:preliminaries}
\myparagraph{Basic notation.}
We write $\nat$ for the set of non-negative integers, and we let $\natpos \isdef \nat \setminus \set{0}[]$.
For $n \in \NN$ we write $[n]$ to denote the set $\set{ i \in \NN \mid 1 \leq i \leq n }$ (i.e., $[0] = \emptyset$, $[1] = \set{1}[]$, and $[n] = \set{1,\ldots,n}$ for $n \geq 2$).
For a set $S$ we write $\mypot{S}$ to denote the power set (i.e., the set of all subsets) of $S$; and for $k \in \NN$ we let $\mypotk{S} \deff \set{ X \subseteq S \mid |X|=k }$.
We use bold letters $\at$ to denote a tuple $(a_1, \dots, a_k)$.
The tuple's \emph{arity} $\ar(\at)$ is $k$, and $a_i$ denotes that tuple's $i$-th entry (for $i\in[k]$).
We let $\tset(\at)=\set{a_1,\ldots,a_k}$.
We write $\set{a_1\to b_1, \; \ldots, \; a_k\to b_k}$ to describe the function $f\colon\set{a_1,\ldots,a_k}\to \set{b_1,\ldots,b_k}$ with $f(a_i)=b_i$ for $i\in[k]$.

A \emph{multiset} $M$ is a tuple $(S,f)$, where $S$ is a set and $f$ is a function $f\colon S \to \NNpos$; the number $f(s)$ indicates the multiplicity with which the element $s\in S$ occurs in the multiset $M$.
We write $\mult{M}(x)$ to denote the multiplicity of $x$ in the multiset $M$;
in particular, $\mult{M}(x) = 0$ denotes that $x \not\in S$.
We adopt the usual notation for multisets using brackets $\mset{\cdots}$ in which each $s\in S$ is listed exactly $\mult{M}(s)$ times.
E.g., $\mset{ a, a, b }$ denotes the multiset $(\set{ a,b }, \set{ a \to 2, b \to 1 })$.

A \emph{coloring} of a set $S$ is a function $\gamma\colon S \to C$ for some set $C$.
Let $\alpha\colon S \to C_\alpha$ and $\beta\colon S \to C_\beta$ be two colorings of the same set $S$.
We say that  $\alpha$ \emph{refines} $\beta$, if for all $u, v \in S$: $\alpha(u) = \alpha(v) \implies \beta(u) = \beta(v)$.
The colorings $\alpha$ and $\beta$ are \emph{equivalent}, if for all $u, v \in S$: $\alpha(u) = \alpha(v) \iff \beta(u) = \beta(v)$.

An (uncolored, undirected, simple) graph is a tuple $G \isdef (\V(G),\E(G))$, where $\V(G)$ is a finite set of vertices and $\E(G) \subseteq \mypotedges{\V(G)}$ is a set of edges.
A \emph{forest} is an acyclic graph; and a \emph{tree} is a connected forest.
\smallskip

\myparagraph{Relational Structures.} A \emph{(finite, relational) signature} $\sig$ is a finite, non-empty set;
the elements in $\sig$ are called \emph{relation symbols}.
Every $R\in\sig$ has an associated \emph{arity} $\ar(R)\in\NNpos$.
For a fixed number $k\in\NNpos$ we write $R/k$ to indicate that $R$ is a relation symbol of arity $k$.
The \emph{arity of} $\sig$ is defined as $\ar(\sig) \isdef \max\set{ \ar(R) \mid R \in \sig }$.
By $\ksig$ we denote the subset $\set{ R \in \sig \mid \ar(R) = k }$ of relation symbols of $\sig$ with arity exactly $k$.

A structure $\str{A}$ of signature $\sig$ (for short, \emph{$\sig$-structure}) consists of a finite, non-empty set $V(\str{A})$ (called the \emph{universe} of $\str{A}$), and a relation $R^{\str{A}}\subseteq {V(\str{A})}^{\ar(R)}$ for every $R\in\sig$.
We additionally require that every $v\in V(\str{A})$ occurs as an entry in at least one tuple of at least one relation of $\str{A}$ --- note that this assumption can easily be met, e.g.\ by inserting into $\sig$ a new relation symbol $U$ of arity 1 with $U^{\str{A}}=V(\str{A})$ (here, we identify a tuple $(v)$ of arity~1 with the element $v$).
By $\tup{A}$ we denote the set $\bigunion_{R \in \sig} R^\str{A}$ of all tuples that belong to some relation of $\str{A}$.
We define the \emph{size} of $\str{A}$ as $\size{\str{A}} \isdef \card{\tup{A}}$.
We say that two $\sig$-structures $\A$, $\B$ have \emph{strictly equal size}, if $\card{R^\A} = \card{R^\B}$ for every $R \in \sig$.
The \emph{cohesion} of $\A$ is defined as the number $\density(\A)$ of all tuples $(\at,\bt)\in \tup{A}\times\tup{A}$ with $\at\neq\bt$ and $\tset(\at)\cap\tset(\bt)\neq \emptyset$.
Obviously, $\density(\A) < \size{\A}^2$.

The \emph{Gaifman graph} of $\A$ is defined as the (undirected, simple) graph $G$ with $\V(G) = V(\A)$ and where $\E(G)$ consists of all $\set{u,v}\in \mypotedges{\V(G)}$ for which there is a tuple $\at \in \tA$ with $u,v \in \tset(\at)$.
A $\sig$-structure is called \emph{connected} if its Gaifman graph is connected.

A \emph{binary} signature is a signature $\sig$ where every $R\in\sig$ has arity $\leq 2$.
A \emph{colored multigraph} $\str{G}$ is a structure of a binary signature.
The binary relations of $\str{G}$ can be viewed as directed edge relations that carry specific labels, and the unary relations of $\str{G}$ can be viewed as assigning specific labels to the vertices of $\str{G}$.
\smallskip

\myparagraph{Color Refinement} can be adapted to colored multigraphs by including the vertex colors and the loops in the base color, and the edge colors in the iteration.
This can be formalized as follows:
Let $\sig$ be a binary signature, and let $\str{G}$ be a $\sig$-structure.
For every $v \in \V(\str{G})$, let $\gcr{0}{v} = (\set{ C \in\onesig \mid v \in C^{\str{G}} }[],\, \set{ E_j \in\twosig \mid (v,v) \in {(E_j)}^{\str{G}} }[])$, and for all $i \in \nat$ let
\[
  \gcr{i+1}{v} \;\isdef\; \bigl(
    \gcr{i}{v},\,
    \mset{
      (\lambda(v,w), \gcr{i}{w})
      \mid
      \set{ v,w }[] \in \E(G)
    }[]
  \bigr),
\]
where $G$ denotes the Gaifman graph of $\str{G}$ and
\[
  \lambda(v,w)
  \ \isdef \
  \set[\big]{
    E_j^{+}
    \mid
    E_j\in\twosig,\ (v,w) \in {(E_j)}^{\str{G}}
  }
  \;\union\;
  \set[\big]{
    E_j^{-}
    \mid
    E_j\in\twosig,\
    (w,v) \in {(E_j)}^{\str{G}}
  }.
\]
\begin{example}\label{ex:simple-structures}
Consider $\sigEx \isdef \set{ E/2, R/6 }$ and let $\SpExa$, $\SpExb$ be $\sigEx$-structures with the same universe $\set{ 1,2,3,u,v,w }$, where $E^\SpExa \isdef \{ (1,2), (2,3), (3,1), (u,v), (v,w), (w,u) \}$, $E^\SpExb \isdef \{ (1,2), (2,w), (w,u),\allowbreak (u,v), (v,3), (3,1) \}$ and $R^\SpExa \isdef \set{ (1,2,3,u,v,w) } \defis R^\SpExb$.
A graphical representation of $\SpExa$, $\SpExb$ as \enquote{hypergraphs} can be found in \cref{fig:ex:simple-hypergraph-a,fig:ex:simple-hypergraph-b}, and their Gaifman graph is depicted in~\cref{fig:ex:simple-gaifman}.
We will use these two structures as a running example throughout the rest of the paper.\endExample{}

\newdimen\figrasterwd%
\figrasterwd\textwidth%
\begin{figure}
	\parbox{\figrasterwd}{
		\parbox[b]{0.22\figrasterwd}{
			\subcaptionbox{%
				Representation of $\SpExa$.%
				\label{fig:ex:simple-hypergraph-a}%
			}[\linewidth]{%
\begin{tikzpicture}[
	edge from parent/.style={draw=none},
	every node/.style={mynode},
	]
	\node (1) at (0,0) {$1_{\text{\tiny \textcolor{myblue}{(1)}}}$}
	child {
		node (2) {$2_{\text{\tiny \textcolor{myblue}{(2)}}}$}
		child {
			node (w) {$w_{\text{\tiny \textcolor{myblue}{(6)}}}$}
			child[missing]
			child {
				node (u) {$u_{\text{\tiny \textcolor{myblue}{(4)}}}$}
			}
		}
	}
	child {
		node (3) {$3_{\text{\tiny \textcolor{myblue}{(3)}}}$}
		child {
			node (v) {$v_{\text{\tiny \textcolor{myblue}{(5)}}}$}
		}
	};
	\draw[RelE] (1) edge (2);
	\draw[RelE] (3) edge (1);
	
	\draw[RelE] (u) edge (v);
	\draw[RelE] (w) edge (u);
	
	\draw[RelE] (2) edge (3);
	\draw[RelE] (v) edge (w);

	\draw[RelR] \convexpath{1,3,v,u,w,2}{15pt};
	\draw[RelR] \convexpath{1,3,v,u,w,2}{15pt};
\end{tikzpicture} 			}
		}
		\hfill
		\parbox[b]{0.22\figrasterwd}{
			\subcaptionbox{%
				Representation of $\SpExb$.%
				\label{fig:ex:simple-hypergraph-b}%
			}[\linewidth]{%
\begin{tikzpicture}[
	edge from parent/.style={draw=none},
	every node/.style={mynode},
	]
	\node (1) at (0,0) {$1_{\text{\tiny \textcolor{myblue}{(1)}}}$}
	child {
		node (2) {$2_{\text{\tiny \textcolor{myblue}{(2)}}}$}
		child {
			node (w) {$w_{\text{\tiny \textcolor{myblue}{(6)}}}$}
			child[missing]
			child {
				node (u) {$u_{\text{\tiny \textcolor{myblue}{(4)}}}$}
			}
		}
	}
	child {
		node (3) {$3_{\text{\tiny \textcolor{myblue}{(3)}}}$}
		child {
			node (v) {$v_{\text{\tiny \textcolor{myblue}{(5)}}}$}
		}
	};
	\draw[RelE] (1) edge (2);
	\draw[RelE] (3) edge (1);
	
	\draw[RelE] (u) edge (v);
	\draw[RelE] (w) edge (u);

	\draw[RelE] (2) edge (w);
	\draw[RelE] (v) edge (3);

	\draw[RelR] \convexpath{1,3,v,u,w,2}{15pt};
\end{tikzpicture} 			}
		}
		\hfill
		\parbox[b]{0.22\figrasterwd}{
			\subcaptionbox{%
				Their Gaifman graph.%
				\label{fig:ex:simple-gaifman}%
			}[\linewidth]{
\begin{tikzpicture}[
	every node/.style={mynode},
	]
	\node (1) at (0,0) {$1$}
	child {
		node (2) {$2$}
		child {
			node (w) {$w$}
			child[missing]
			child {
				node (u) {$u$}
			}
		}
	}
	child {
		node (3) {3}
		child {
			node (v) {$v$}
		}
	};
	\draw (1) -- (u);
	\draw (1) -- (v);
	\draw (1) -- (w);

	\draw (2) -- (3);
	\draw (2) -- (v);
	\draw (2) -- (u);

	\draw (3) -- (u);
	\draw (3) -- (w);

	\draw (v) -- (w);
	\draw (v) -- (u);
\end{tikzpicture} 				\vspace{4.5mm}
			}
		}
		\hfill
		\parbox[b]{0.3\figrasterwd}{
			\subcaptionbox{%
				The incidence graph of $\SpExa$.%
				\label{fig:ex:simple-incidence-a}%
			}[\linewidth]{
\begin{tikzpicture}[
	every node/.style={circle, fill=black},
	re/.style={draw=black, regular polygon, regular polygon sides=4, inner sep=3.25pt, fill=myred},
	be/.style={draw=black, regular polygon, regular polygon sides=5, fill=myblue},
	xscale=0.8,
	yscale=0.8]
	\node (1) at (0,0) {};
	\node (2) at (1,0) {};
	\node (3) at (2,0) {};
	\node (u) at (3,0) {};
	\node (v) at (4,0) {};
	\node (w) at (5,0) {};

	\node[re] (12) at (0,1) {};
	\node[re] (23) at (1,1) {};
	\node[re] (31) at (2,1) {};
	\node[re] (uv) at (3,1) {};
	\node[re] (vw) at (4,1) {};
	\node[re] (wu) at (5,1) {};
	\node[be] (123uvw) at (2.5,-.75) {};

	\draw[-latex] (1) -- (123uvw);
	\draw[-latex] (1) -- (12);
	\draw[-latex] (1) -- (31);

	\draw[-latex] (2) -- (123uvw);
	\draw[-latex] (2) -- (12);
	\draw[-latex] (2) -- (23);

	\draw[-latex] (3) -- (123uvw);
	\draw[-latex] (3) -- (23);
	\draw[-latex] (3) -- (31);

	\draw[-latex] (u) -- (123uvw);
	\draw[-latex] (u) -- (uv);
	\draw[-latex] (u) -- (wu);

	\draw[-latex] (v) -- (123uvw);
	\draw[-latex] (v) -- (uv);
	\draw[-latex] (v) -- (vw);

	\draw[-latex] (w) -- (123uvw);
	\draw[-latex] (w) -- (vw);
	\draw[-latex] (w) -- (wu);
\end{tikzpicture} 				\vspace{4.5mm}
			}
			\vspace{4.5mm}

			\subcaptionbox{%
				The incidence graph of $\SpExb$.%
				\label{fig:ex:simple-incidence-b}%
			}[\linewidth]{
\begin{tikzpicture}[
	every node/.style={circle, fill=black},
	re/.style={draw=black, regular polygon, regular polygon sides=4, inner sep=3.25pt, fill=myred},
	be/.style={draw=black, regular polygon, regular polygon sides=5, fill=myblue},
	xscale=0.8,
	yscale=0.8]
	\node (1) at (0,0) {};
	\node (2) at (1,0) {};
	\node (w) at (2,0) {};
	\node (u) at (3,0) {};
	\node (v) at (4,0) {};
	\node (3) at (5,0) {};

	\node[re] (12) at (0,1) {};
	\node[re] (2w) at (1,1) {};
	\node[re] (wu) at (2,1) {};
	\node[re] (uv) at (3,1) {};
	\node[re] (v3) at (4,1) {};
	\node[re] (31) at (5,1) {};
	\node[be] (123uvw) at (2.5,-.75) {};

	\draw[-latex] (1) -- (123uvw);
	\draw[-latex] (1) -- (12);
	\draw[-latex] (1) -- (31);

	\draw[-latex] (2) -- (123uvw);
	\draw[-latex] (2) -- (12);
	\draw[-latex] (2) -- (2w);

	\draw[-latex] (3) -- (123uvw);
	\draw[-latex] (3) -- (v3);
	\draw[-latex] (3) -- (31);

	\draw[-latex] (u) -- (123uvw);
	\draw[-latex] (u) -- (wu);
	\draw[-latex] (u) -- (uv);

	\draw[-latex] (v) -- (123uvw);
	\draw[-latex] (v) -- (uv);
	\draw[-latex] (v) -- (v3);

	\draw[-latex] (w) -- (123uvw);
	\draw[-latex] (w) -- (2w);
	\draw[-latex] (w) -- (wu);
\end{tikzpicture} 				\vspace{4.5mm}
			}
		}
	}
	\caption{The two $\sigEx$-structures $\SpExa$, $\SpExb$ according to~\cref{ex:simple-structures}, and their Gaifman and incidence graphs.}%
	\label{fig:ex:simple}
\end{figure}
 \end{example}

\myparagraph{Types.}
The notion of \emph{atomic type} and the \emph{similarity type} of tuples will be crucial for the definition of Relational Color Refinement.
For a $\sig$-structure $\A$ and a tuple $\at \in {V(\A)}^k$ of arity $k$, the \emph{atomic type} $\atp(\at)$ is the set $\set{ R \in \sig \mid \at \in R^\A }$.
For every tuple $\bt \in {V(\A)}^\ell$ of arity $\ell$, the \emph{similarity type} $\stp(\at, \bt)$ between $\at$ and $\bt$ is defined as the set $\set{ (i, j)\in[k]\times[\ell] \mid a_i = b_j }$.
We use $\stp(\at)$ as shorthand for $\stp(\at, \at)$.

In general, an \emph{atomic type $\atype$ of arity $k$} (over signature $\sig$) is a subset of $\ksig$.
A \emph{similarity type $\stype$ of arity $(k,\ell)$} (over signature $\sig$) is a subset of $[k] \times [\ell]$ that satisfies the following condition of \enquote{transitivity}:
\begin{equation*}
  \text{For all } \ i,i' \in [k],\ j,j' \in [\ell]\ \text{ we have:} \quad
	\set{ (i,j), (i',j), (i,j') } \subseteq \stype \implies (i', j') \in \stype\;.
\end{equation*}
We write $\ar(\atype)$ and $\ar(\stype)$ to denote the arity of $\atype$ and $\stype$.
If $\stype$ has arity $(k,k)$, we simply say that $\stype$ has arity $k$ and write $\ar(\stype) = k$.
We let $\sTypes$ be the set of all similarity types of arity $(\ar(R), \ar(S))$ for any $R,S \in \sig$.

We say that a tuple $\at\in {V(\A)}^k$ has atomic type $\atype$ if $\rho = \atp(\at)$.
Analogously, we say that $\at$ has similarity type $\stype$, if $\stype = \stp(\at)$; and we say that $\at$, $\bt$ have similarity type $\stype$, if $\stype = \stp(\at, \bt)$. 
	\section{Color Refinement on Relational Structures}\label{sec:RCR}
The goal of this section is to introduce Relational Color Refinement as a generalization of Color Refinement from graphs to relational structures.
Let $\sig$ be an arbitrary (relational) signature; this $\sig$ will be fixed throughout the rest of the paper.%
\smallskip

\myparagraph{Naive approaches.}
There are two obvious ways to lift the Color Refinement method from graphs to $\sig$-structures $\A$, $\B$.
The first is to run Color Refinement on the Gaifman graphs of $\A$ and $\B$.
Recall the $\sigEx$-structures $\SpExa$, $\SpExb$ from \cref{ex:simple-structures}: as their Gaifman graphs are identical, this will not distinguish $\SpExa$ and $\SpExb$.
The second way is to run Color Refinement on the \emph{incidence graphs} of $\A$ and $\B$.
The incidence graph of $\A$ is the colored multigraph $\str{I}$ of signature $\sig_{I} \isdef \set{E} \union \set{ U_R \mid R \in \sig }$, where $\ar(E)=2$ and $\ar(U_R)=1$ for all $R \in \sig$.
The universe $V(\str{I})$ of $\str{I}$ consists of all elements in $V(\A)$ and a new element $v_{(R,\at)}$ for each $R\in\sig$ and each $\at\in R^{\A}$.
Furthermore, $E^{\str{I}}$ is the set of all tuples $(u,v_{(R,\at)})$ with $R\in\sig$, $\at\in R^{\A}$, and $u\in\tset(\at)$;
and $(U_R)^{\str{I}} = \set{ v_{(R,\at)} \mid \at \in R^{\A} }$, for every $R\in\sig$. %
The incidence graphs of $\SpExa$ and $\SpExb$ from~\cref{ex:simple-structures} are depicted in~\cref{fig:ex:simple-incidence-a,fig:ex:simple-incidence-b}; the black (red, blue) vertices are those vertices of the incidence graph that belong to the original structure's universe (the relation $U_E$, the relation $U_R$, respectively).
Running the classical Color Refinement method on these incidence graphs will not distinguish between $\SpExa$ and $\SpExb$, because the initial coloring of the vertices of the incidence graphs is already a stable coloring.

In \cref{app:advanced-color-refinement} we discuss two further naive approaches that use Color Refinement on an enhanced version of the Gaifman graph or the incidence graph, respectively, where additional information is encoded as edge colors; and we present examples which demonstrate that the distinguishing power of these approaches still is rather weak.

This paper aims for a stronger coloring method.
The key idea is to color the tuples in $\tup{A}$ and $\tup{B}$ and take into account the tuples' atomic type and their mutual overlap (the latter is formalized by their similarity type).
The details are as follows.
\smallskip

\myparagraph{Relational Color Refinement (RCR, for short).}
Consider an arbitrary $\sig$-structure $\A$.
We iteratively compute colors for all tuples $\at\in\tup{A}$.
For all $\at \in \tA$, the initial color consists of the atomic type and the similarity type of $\at$, i.e., it is $\col{0}{\at} \isdef (\atp(\at), \stp(\at))$.
For $i \in \natpos$, the color after $i$ iterations is defined as $\col{i}{\at} \isdef (\col{i-1}{\at}, \Nset{i}{\A}(\at))$, where
\begin{equation*}
	\Nset{i}{\A}(\at) \:\isdef\: \mset[\big]{
		\bigl( 
			\stp(\at,\bt),\,
			\col{i-1}{\bt}
		\bigr)
		\mid
		\bt \in \tA,\, \stp(\at, \bt) \neq \emptyset 
	}[].
\end{equation*}
Note that $\stp(\at,\bt)\neq\emptyset$ $\iff$ $\tset(\at)\cap\tset(\bt)\neq\emptyset$.
By definition, $\rc_i$ refines $\rc_{i-1}$ for all $i \in \natpos$.
The $i$-th coloring is \emph{stable}, if for all $\at, \bt \in \tA$ we have that $\col{i}{\at} = \col{i}{\bt} \iff \col{i+1}{\at} = \col{i+1}{\bt}$.
It is easy to see that for every $\sig$-structure $\A$ there is an $i\in\NN$ such that the $i$-th coloring is stable;
we let $i_\A$ be the smallest such number, and we write $\col{\infty}{\at}$ to denote $\col{i_\A}{\at}$.

For $i\in\NN$ we write  $\cols{i}{\A}$ to denote the set of colors produced in the $i$-th refinement round, i.e., $\cols{i}{\A} = \set{ \col{i}{\at} \mid \at \in \tA }$.
We let $\cols{}{\A}\deff\cols{\infty}{\A}\deff \cols{i_\A}{\A}$; and we will call this the set of \emph{stable colors on $\A$ produced by RCR}.

For each $i\in\NN$ and each $c \in \cols{i}{\A}$ we let $\ccount{\A}{c} \isdef \card{\set{ \at \in \tA \mid \col{i}{\at} = c}}$, i.e., $\ccount{\A}{c}$ is the number of tuples with color $c$.
We say that RCR \emph{distinguishes $\A$ and $\B$ in round $i$}, if there is a color $c \in \cols{i}{\A} \union \cols{i}{\B}$ such that $\ccount{\A}{c} \neq \ccount{\B}{c}$.
Furthermore, we say that RCR \emph{distinguishes $\A$ and $\B$} if there is an $i\leq\max\set{i_\A,i_\B}$ such that RCR distinguishes $\A$ and $\B$ in round $i$.
It is straightforward to see that if $\A$ and $\B$ are \emph{not} of strictly equal size, then RCR distinguishes $\A$ and $\B$ in round 0.

\begin{example}
Consider the $\sigEx$-structures $\SpExa$, $\SpExb$ from~\cref{ex:simple-structures}.
RCR distinguishes $\SpExa$, $\SpExb$ --- see~\cref{app:color-refinement-run} for a complete run of Relational Color Refinement on $\SpExa$, $\SpExb$.
\cref{app:advanced-color-refinement} presents a more sophisticated pair
of structures that are also distinguished by Relational Color Refinement.
\end{example}

\myparagraph{Connection to Color Refinement on graphs.}
We represent a $\sig$-structure $\A$ by a colored multigraph $\grep{\A}$ of the signature $\repsig\deff \set{ E_{i,j} \mid i,j\in [\ar(\sig)] } \union \set{ U_R \mid R \in \sig }$, where $\ar(E_{i,j})=2$ for all $i,j\in [\ar(\sig)]$ and $\ar(U_R)=1$ for all $R\in\sig$.
The universe $V(\grep{A})$ of $\grep{\A}$ consists of a new element $w_{\at}$ for every tuple $\at\in\tup{A}$.
Furthermore, ${(U_R)}^{\grep{\A}} = \set{ w_{\at} \mid \at \in R^\A}$, for all $R\in\sig$.
And for all $i,j\in [\ar(\sig)]$ we have ${(E_{i,j})}^{\grep{\A}} \isdef \set{ (w_{\at}, w_{\bt}) \mid \at,\bt\in\tup{A}, \ (i,j) \in \stp(\at, \bt) }$.
\begin{example}
The representations $\grep{\SpExa}$, $\grep{\SpExb}$ of $\SpExa$ and
$\SpExb$ as colored multigraphs are depicted in~\cref{fig:ex:simple-graph-repr-a,fig:ex:simple-graph-repr-b}, respectively.
To keep the figure easy to grasp, we labeled the vertices with the tuples they represent, omitted the self-loops, and contracted multi-edges into a single one with combined edge label, where $xy$ denotes the tuple $(x,y)$. %
\endExample{}

\begin{figure}
	\begin{subfigure}{0.495\textwidth}
		\centering
\begin{tikzpicture}[
	every node/.style={myvertex},
	xscale=2.4,
	yscale=1.2]
	\clip (-1.7,-2.252) rectangle (1.7, 2.252);
	\tikzmath{
		\rad=1.5;
		\offset=.075;
		\small=8;
		\medium=15;
		\myangle=40;
		\distA=38pt;
		\distB=43pt;
	}
	\node[colorR, densely dashed] (123uvw) at (0:0) {$1, 2, 3, u, v, w$};
	\node[colorE] (12) at (0:\rad) {$1, 2$};
	\node[colorE] (23) at (60:\rad) {$2, 3$};
	\node[colorE] (31) at (120:\rad) {$3, 1$};
	\node[colorE] (uv) at (180:\rad) {$\vphantom{1}u, v$};
	\node[colorE] (vw) at (240:\rad) {$\vphantom{1}v, w$};
	\node[colorE] (wu) at (300:\rad) {$\vphantom{1}w, u$};

	\draw 
		($(123uvw.east)-(0,\offset)$)
		edge[tpFor]
		node[edgeAnnotation, below=-2pt] {$11, 22$}
	($(12.west)-(0,\offset)$);
	\draw[tpFor, bend right=\small] 
		(123uvw)
		edge
		node[edgeAnnotation, rotate=\myangle, below=-2pt] {$21, 32$}
	(23);
	\draw[tpFor, bend right=\small] 
		(123uvw)
		edge
		node[edgeAnnotation, rotate=-\myangle, above=-2pt] {$12, 31$}
	(31);
	\draw 
		($(123uvw.west)+(0,\offset)$)
		edge[tpFor]
		node[edgeAnnotation, above=-2pt] {$41, 52$}
	($(uv.east)+(0,\offset)$);
	\draw[tpFor, bend right=\small] 
		(123uvw)
		edge
		node[edgeAnnotation, rotate=\myangle, above=-2pt] {$51, 62$}
	(vw);
	\draw[tpFor, bend right=\small] 
		(123uvw)
		edge
		node[edgeAnnotation, rotate=-\myangle, below=-2pt] {$42, 61$}
	(wu);

	\draw
		($(12.west)+(0,\offset)$)
		edge[tpFor]
		node[edgeAnnotation, above=-2pt] {$11, 22$}
	($(123uvw.east)+(0,\offset)$);
	\draw[tpFor, bend right=\small]
		(12)
		edge
		node[edgeAnnotation, rotate=-\myangle, above=-2pt] {$21$}
	(23);
	\draw[tpBack, in=75, out=85, distance=\distB]
		(31)
		edge
		node[edgeAnnotation, rotate=-15, above=-2pt] {$12$}
	(12);

	\draw[tpFor, bend right=\small]
		(23)
		edge
		node[edgeAnnotation, rotate=\myangle, above=-2pt] {$12, 23$}
	(123uvw);
	\draw[tpFor, bend right=\small]
		(23)
		edge
		node[edgeAnnotation, rotate=-\myangle, below=-2pt] {$12$}
	(12);
	\draw[tpFor, bend right=\medium]
		(23)
		edge
		node[edgeAnnotation, above=-1pt] {$21$}
	(31);

	\draw[tpFor, bend right=\small]
		(31)
		edge
		node[edgeAnnotation, rotate=-\myangle, below=-2pt] {$13, 21$}
	(123uvw);
	\draw[tpFor, in=85, out=75, distance=\distA]
		(31)
		edge
		node[edgeAnnotation, rotate=-15, below=-2pt] {$21$}
	(12);
	\draw[tpFor, bend right=\medium]
		(31)
		edge
		node[edgeAnnotation, below=-1pt] {$12$}
	(23);

	\draw
		($(uv.east)-(0,\offset)$)
		edge[tpFor]
		node[edgeAnnotation, below=-1pt] {$14, 25$}
	($(123uvw.west)-(0,\offset)$);
	\draw[tpFor, bend right=\small]
		(uv)
		edge
		node[edgeAnnotation, rotate=-\myangle, below=-2pt] {$21$}
	(vw);
	\draw[tpFor, in=265, out=255, distance=\distB]
		(uv)
		edge
		node[edgeAnnotation, rotate=-15, below=-2pt] {$12$}
	(wu);

	\draw[tpFor, bend right=\small]
		(vw)
		edge
		node[edgeAnnotation, rotate=\myangle, below=-2pt] {$15, 26$}
	(123uvw);
	\draw[tpFor, bend right=\small]
		(vw)
		edge
		node[edgeAnnotation, rotate=-\myangle, above=-2pt] {$12$}
	(uv);
	\draw[tpFor, bend right=\medium]
		(vw)
		edge
		node[edgeAnnotation, below=-1pt] {$21$}
	(wu);

	\draw[tpFor, bend right=\small]
		(wu)
		edge
		node[edgeAnnotation, rotate=-\myangle, above=-2pt] {$16, 24$}
	(123uvw);
	\draw[tpBack, in=255, out=265, distance=\distA]
		(uv)
		edge
		node[edgeAnnotation, rotate=-15, above=-2pt] {$21$}
	(wu);
	\draw[tpFor, bend right=\medium]
		(wu)
		edge
		node[edgeAnnotation, above=-1pt] {$12$}
	(vw);
\end{tikzpicture} 		\vspace{-\abovedisplayskip} %
		
		\caption{Colored multigraph representation $\grep{\SpExa}$.}%
		\label{fig:ex:simple-graph-repr-a}
	\end{subfigure}
	\hfill
	\begin{subfigure}{0.495\textwidth}
		\centering
\begin{tikzpicture}[
	every node/.style={myvertex},
	xscale=2.4,
	yscale=1.2,
	]
	\clip (-1.7,-2.252) rectangle (1.7, 2.252);
	\tikzmath{
		\rad=1.5;
		\offset=.075;
		\small=8;
		\medium=15;
		\myangle=40;
	}
	\node[colorR, densely dashed] (123uvw) at (0:0) {$1, 2, 3, u, v, w$};
	\node[colorE] (12) at (0:\rad) {$1, 2$};
	\node[colorE] (2w) at (60:\rad) {$2, w$};
	\node[colorE] (wu) at (120:\rad) {$\vphantom{1}w, u$};
	\node[colorE] (uv) at (180:\rad) {$\vphantom{1}u, v$};
	\node[colorE] (v3) at (240:\rad) {$v, 3$};
	\node[colorE] (31) at (300:\rad) {$3, 1$};

	\draw 
		($(123uvw.east)-(0,\offset)$)
		edge[tpFor]
		node[edgeAnnotation, below=-2pt] {$11, 22$}
	($(12.west)-(0,\offset)$);
	\draw[tpFor, bend right=\small] 
		(123uvw)
		edge
		node[edgeAnnotation, rotate=\myangle, below=-2pt] {$21, 62$}
	(2w);
	\draw[tpFor, bend right=\small] 
		(123uvw)
		edge
		node[edgeAnnotation, rotate=-\myangle, above=-2pt] {$42, 61$}
	(wu);
	\draw
		($(123uvw.west)+(0,\offset)$)
		edge[tpFor]
		node[edgeAnnotation, above=-2pt] {$41, 52$}
	($(uv.east)+(0,\offset)$);
	\draw[tpFor, bend right=\small] 
		(123uvw)
		edge
		node[edgeAnnotation, rotate=\myangle, above=-2pt] {$32, 51$}
	(v3);
	\draw[tpFor, bend right=\small] 
		(123uvw)
		edge
		node[edgeAnnotation, rotate=-\myangle, below=-2pt] {$12, 31$}
	(31);

	\draw
		($(12.west)+(0,\offset)$)
		edge[tpFor]
		node[edgeAnnotation, above=-2pt] {$11, 22$}
	($(123uvw.east)+(0,\offset)$);
	\draw[tpFor, bend right=\small]
		(12)
		edge
		node[edgeAnnotation, rotate=-\myangle, above=-2pt] {$21$}
	(2w);
	\draw[tpFor, bend right=\small]
		(12)
		edge
		node[edgeAnnotation, rotate=\myangle, above=-2pt] {$12$}
	(31);

	\draw[tpFor, bend right=\small]
		(2w)
		edge
		node[edgeAnnotation, rotate=\myangle, above=-2pt] {$12, 26$}
	(123uvw);
	\draw[tpFor, bend right=\small]
		(2w)
		edge
		node[edgeAnnotation, rotate=-\myangle, below=-2pt] {$12$}
	(12);
	\draw[tpFor, bend right=\medium]
		(2w)
		edge
		node[edgeAnnotation, above=-1pt] {$21$}
	(wu);

	\draw[tpFor, bend right=\small]
		(wu)
		edge
		node[edgeAnnotation, rotate=-\myangle, below=-2pt] {$16, 24$}
	(123uvw);
	\draw[tpFor, bend right=\medium]
		(wu)
		edge
		node[edgeAnnotation, below=-2pt] {$12$}
	(2w);
	\draw[tpFor, bend right=\small]
		(wu)
		edge
		node[edgeAnnotation, rotate=\myangle, above=-1pt] {$21$}
	(uv);

	\draw
		($(uv.east)-(0,\offset)$)
		edge[tpFor]
		node[edgeAnnotation, below=-1pt] {$14, 25$}
	($(123uvw.west)-(0,\offset)$);
	\draw[tpFor, bend right=\small]
		(uv)
		edge
		node[edgeAnnotation, rotate=\myangle, below=-2pt] {$12$}
	(wu);
	\draw[tpFor, bend right=\small]
		(uv)
		edge
		node[edgeAnnotation, rotate=-\myangle, below=-2pt] {$21$}
	(v3);

	\draw[tpFor, bend right=\small]
		(v3)
		edge
		node[edgeAnnotation, rotate=\myangle, below=-2pt] {$15, 23$}
	(123uvw);
	\draw[tpFor, bend right=\small]
		(v3)
		edge
		node[edgeAnnotation, rotate=-\myangle, above=-2pt] {$12$}
	(uv);
	\draw[tpFor, bend right=\medium]
		(v3)
		edge
		node[edgeAnnotation, below=-1pt] {$21$}
	(31);

	\draw[tpFor, bend right=\small]
		(31)
		edge
		node[edgeAnnotation, rotate=-\myangle, above=-2pt] {$13, 21$}
	(123uvw);
	\draw[tpBack, bend right=\small]
		(31)
		edge
		node[edgeAnnotation, rotate=\myangle, below=-2pt] {$21$}
	(12);
	\draw[tpFor, bend right=\medium]
		(31)
		edge
		node[edgeAnnotation, above=-1pt] {$12$}
	(v3);
\end{tikzpicture} 		\vspace{-\abovedisplayskip}
		
		\caption{Colored multigraph representation $\grep{\SpExb}$.}%
		\label{fig:ex:simple-graph-repr-b}
	\end{subfigure}
	\caption{}%
	\label{fig:ex:graph-representation}
\end{figure} \end{example}

It is easy to see (cf.~\cref{app:rcr-is-cr-on-rep} for a proof) that running RCR on $\A$ produces (in the same number of rounds) a stable coloring of $\tA$ that is equivalent (via identifying $\at\in\tA$ with $w_{\at}\in V(\grep{\A})$) to the stable coloring produced by classical Color Refinement on the colored multigraph $\grep{\A}$.

Berkholz, Bonsma and Grohe~\cite{Berkholz2017} showed that classical Color Refinement can be implemented to run in time $\bigO((n+m) \cdot \log(n))$ on colored multigraphs, where $n$ denotes the number of vertices and $m$ denotes the total number of edges.
Thus, on the colored multigraph $\grep{\A}$ representing $\A$, classical Color Refinement runs in time $\bigO((n+m) \cdot \log(n))$ where $n=\card{\tA}=\size{\A}$ and $m=\sum_{i,j\in [\ar(\sig)]}\card{{(E_{i,j})}^{\grep{\A}}}$.
Therefore, Relational Color Refinement can be implemented to run, for each fixed signature $\sig$, in time $\bigO\bigl( (\size{\A} + \density(\A)) \cdot \log(\size{\A}) \bigr)$ (recall from Section~\ref{sec:preliminaries} the definition of the \emph{cohesion} $\density(\A)$).
In \cref{sec:main-runtime} we will improve the running time to $\bigO(\size{\A}\cdot\log\size{\A})$ by using a representation of $\A$ by a colored multigraph different from $\grep{\A}$.

For the special case where $\sigma=\set{E/2,U/1}$, the following is straightforward to see (consult~\cref{app:rel-cr-generalizes-cr} for a proof).
For any (simple, undirected) graph $G$ let $\A_G$ be the $\sig$-structure $\A$ that represents $G$ as follows.
$V(\A)=U^\A=\V(G)$ and $E^\A$ consists of the tuples $(u,v)$ and $(v,u)$ for all $\set{u,v}\in \E(G)$.
RCR on $\A_G$ produces a stable coloring $\gamma$ that is equivalent to the stable coloring $\gamma'$ produced by classical Color Refinement on $G$ in the sense that for any two $u,v\in \V(G)$ we have $\gamma(u)=\gamma(v)$ $\iff$ $\gamma'(u)=\gamma'(v)$.
Therefore, RCR can be viewed as a generalization of classical Color Refinement from graphs to $\sig$-structures for arbitrary (relational) signatures $\sig$.

In the next two sections, we will prove that the distinguishing power of RCR on $\sig$-structures has several natural characterizations that are analogous to those known for classical Color Refinement on graphs.
 	
	\section{Connection to Homomorphism Counts}\label{sec:main-homomorphism}
In this section we relate the distinguishing power of Relational Color Refinement to distinguishability via homomorphism counts from \emph{acyclic} $\sig$-structures.
\smallskip

\myparagraph{Acyclic $\sig$-structures.}
Let $\C$ be a $\sig$-structure.
A \emph{join-tree} for $\C$ is a tree (i.e., an undirected, simple
graph that is connected and acyclic) $T$ with vertex set $\V(T) \deff \tC$ (i.e., the tuples in $\tC$ serve as vertices of $T$) and which
satisfies the following connectivity condition: for all $u\in V(\C)$
the set $\set{ {\ct}\in \V(T) \mid u \in \tset(\ct) }$ induces a
connected subgraph of $T$.

We call a $\sig$-structure $\C$ \emph{acyclic} if there exists a join-tree for $\C$.
This definition of acyclicity of $\sig$-structures is equivalent to acyclicity as defined in the textbook~\cite{Abiteboul1995}, it is equivalent to the notion of \emph{alpha-acyclicity} as defined in~\cite{Beeri1983,DBLP:journals/siamcomp/BernsteinG81} and, finally, is also equivalent to $\C$ having (generalized or fractional) hypertree width~1 as defined in~\cite{Gottlob2002,Gottlob2003,Grohe2014}.
In the literature, also other notions of acyclicity for relational structures (and hypergraphs) have been considered; but alpha-acyclicity arguably is the most common and the least restrictive one.
Consult~\cite{BraultBaron2016} for a detailed survey on this topic.

For the special case of \emph{binary} signatures $\repsig$, i.e., where $\repsig$-structures are colored multigraphs, it is known (see e.g.~\cite{BraultBaron2016}) that a $\repsig$-structure $\C$ is acyclic if, and only if, its Gaifman graph is acyclic (w.r.t.\ the usual notion of acyclicity of undirected simple graphs).
The following example shows that for non-binary signatures $\sig$ there exist acyclic $\sig$-structures whose Gaifman graph is not acyclic.
\begin{example}\label{ex:join-trees}

Both $\sigEx$-structures $\SpExa$, $\SpExb$ from~\cref{ex:simple-structures} are acyclic, as witnessed by the join-trees depicted in~\cref{fig:ex:join-tree}.
The Gaifman graphs of $\SpExa$ and $\SpExb$ (see \cref{fig:ex:simple-gaifman}) are not acyclic.%
\endExample{}

\begin{figure}
	\begin{subfigure}{0.495\textwidth}
		\centering
\settowidth{\myA}{${(1,2)}$}%
\begin{tikzpicture}[
	every node/.style={draw, rectangle}, 
	level/.style = {level distance=1cm, sibling distance=1.3cm},
]
	\node (123uvw) at (0,0) {${(1,2,3,u,v,w)}$}
	child[edge from parent path={
		(\tikzparentnode\tikzparentanchor) edge [bend right=8]  (\tikzchildnode\tikzchildanchor)
	}] { 
		node {\makebox[\myA][c]{${(1,2)}$}} 
	}
	child { node {\makebox[\myA][c]{${(2,3)}$}} }
	child { node {\makebox[\myA][c]{${(3,1)}$}} }
	child { node {\makebox[\myA][c]{$\vphantom{1}{(u,v)}$}} }
	child { node {\makebox[\myA][c]{$\vphantom{1}{(v,w)}$}} }
	child[edge from parent path={
		(\tikzparentnode\tikzparentanchor) edge [bend left=8]  (\tikzchildnode\tikzchildanchor)
	}] { node {\makebox[\myA][c]{$\vphantom{1}{(w,u)}$}} };
\end{tikzpicture} 		\caption{A join-tree for $\SpExa$.}%
		\label{fig:ex:simple-join-tree-a}
	\end{subfigure}
	\hfill
	\begin{subfigure}{0.495\textwidth}
		\centering
\settowidth{\myA}{${(1,2)}$}%
\begin{tikzpicture}[
	every node/.style={draw, rectangle}, 
	level/.style = {level distance=1cm, sibling distance=1.3cm},
]
	\node (123uvw) at (0,0) {${(1,2,3,u,v,w)}$}
	child[edge from parent path={
		(\tikzparentnode\tikzparentanchor) edge [bend right=8]  (\tikzchildnode\tikzchildanchor)
	}] { 
		node {\makebox[\myA][c]{${(1,2)}$}} 
	}
	child { node {\makebox[\myA][c]{${(2,w)}$}} }
	child { node {\makebox[\myA][c]{$\vphantom{1}{(w,u)}$}} }
	child { node {\makebox[\myA][c]{$\vphantom{1}{(u,v)}$}} }
	child { node {\makebox[\myA][c]{${(v,3)}$}} }
	child[edge from parent path={
		(\tikzparentnode\tikzparentanchor) edge [bend left=8]  (\tikzchildnode\tikzchildanchor)
	}] { node {\makebox[\myA][c]{${(3,1)}$}} };
\end{tikzpicture} 		\caption{A join-tree for $\SpExb$.}%
		\label{fig:ex:simple-join-tree-b}
	\end{subfigure}
	\caption{}%
	\label{fig:ex:join-tree}
\end{figure} \end{example}

In this section, we will use a slightly different (but equivalent) notion of a join-tree that we call \enquote{relaxed join-tree}.
Let $\C$ be a $\sig$-structure.
A \emph{relaxed join-tree} $J \isdef (T, \beta)$ for $\C$ consists 
of a tree $T$ and a surjective function $\beta\colon \V(T) \to \tC$ which satisfies the following connectivity condition: for all $u \in \V(\C)$ the set $\set{ v \in \V(T) \mid u \in \tset(\beta(v)) }$ induces a connected subgraph of $T$.

\begin{proposition}\label{prop:relaxedjointree}
	For every $\sig$-structure $\C$, the following are equivalent:
	\begin{enumerate}[topsep=0pt, noitemsep,label=(\arabic*)]
		\item There exists a join-tree for $\C$.
		\item There exists a relaxed join-tree for $\C$.
	\end{enumerate}
\end{proposition}
\begin{proof}
	It is easy to see that (1) implies (2) because a join-tree $T$
        can be turned into a relaxed join-tree $J=(T,\beta)$ by
        letting $\beta(\ct) \deff\ct$ for all $\ct\in \tC=\V(T)$.

	To see that (2) implies (1), note that every relaxed join-tree  satisfies the requirements for a hypertree decomposition of width 1 (cf.\ \cite{Gottlob2002}).
        Thus, (2) implies
        that $\C$ has hypertree width~1. From~\cite{Gottlob2002}
        it
        is known that a structure has hypertree width 1
        if, and only if, it has a join tree.
        This completes the proof of \cref{prop:relaxedjointree}.
\end{proof}

\myparagraph{Homomorphisms.}
A homomorphism from a $\sig$-structure $\C$ to a $\sig$-structure $\A$ is a mapping $h: V(\C) \to V(\A)$ such that for all $R \in \sig$, for $k\deff \ar(R)$, and for all $\ct = (c_1, \dots, c_k) \in R^{\C}$ we have $(h(c_1), \dots, h(c_k)) \in R^\A$.
We write $\Hom(\C, \A)$ for the set of all homomorphisms from $\C$ to $\A$, and we let $\hom(\C, \A) \isdef \card{\Hom(\C, \A)}$ denote the number of homomorphisms from $\C$ to $\A$.
\begin{example}\label{ex:homomorphisms}
The identity is a homomorphism from $\SpExa$ to $\SpExa$; thus $\hom(\SpExa, \SpExa) \geq 1$.
There is no homomorphism from $\SpExa$ to $\SpExb$: Assume for contradiction that $h\in\Hom(\SpExa,\SpExb)$.
Since $R^\SpExa=R^\SpExb=\set{(1,2,3,u,v,w)}$, $h$ must be the identity; however, $(2,3)\in E^\SpExa$ but $(h(2),h(3))=(2,3)\not\in E^\SpExb$.
Hence, $\hom(\SpExa,\SpExb)=0$.
Recall that $\SpExa$ is acyclic (see~\cref{ex:join-trees}).
Hence, there exists an acyclic $\sig$-structure $\C$ (namely, $\C=\SpExa$) such that $\hom(\C,\SpExa)\neq \hom(\C,\SpExb)$.
\endExample{}
 \end{example}

This section's main result is \cref{thm:main-homs}:
\mainhoms{}
This theorem can be viewed as a generalization of the following result by Dvo\v{r}ák~\cite{Dvorak2010} and Dell, Grohe and Rattan~\cite{Dell2018} to arbitrary signatures $\sig$.
While~\cite{Dvorak2010,Dell2018} state the theorem just for graphs, it easily extends to colored multigraphs (as noted in~\cite{Dell2018}).
A \emph{colored multitree} is an acyclic and connected colored multigraph, i.e., a colored multigraph whose Gaifman graph is a tree.

\begin{theorem}[\cite{Dvorak2010,Dell2018}]\label{thm:dvorak}
	Let $\str{G}$ and $\str{H}$ be colored multigraphs.
	The following statements are equivalent.
	\begin{enumerate}[topsep=0pt, noitemsep]
		\item Color Refinement distinguishes $\str{G}$ and $\str{H}$.
		\item There exists a colored multitree $\str{T}$ such that $\hom(\str{T}, \str{G}) \neq \hom(\str{T}, \str{H})$.
	\end{enumerate}
\end{theorem}
\cref{thm:dvorak} will serve as the first key ingredient of our proof of \cref{thm:main-homs}.
The second key ingredient is to use the following notion of a \emph{relaxed join-tree representation} $\jtrep{\C}{J}$.
Recall from \cref{sec:RCR} the binary signature $\repsig \deff \set{ E_{i,j} \mid i,j \in [\ar(\sig)] } \union \set{ U_R \mid R \in \sig }$ and the colored multigraph $\grep{\A}$ of signature $\repsig$ that represents a $\sig$-structure $\A$.
For an acyclic $\sig$-structure $\C$ and a relaxed join-tree $J=(T,\beta)$ for $\C$ we define the colored multigraph $\jtrep{\C}{J}$ of signature $\repsig$ to have universe $\V(\jtrep{\C}{J})\deff \V(T)$ and
\begin{equation*}
	{(U_R)}^{\jtrep{\C}{J}} \isdef \set{ v \mid \beta(v) \in R^\C }
	\;\;\text{and}\;\;
	{(E_{i,j})}^{\jtrep{\C}{J}} \isdef \set{ (v, w) \mid \set{v, w} \in \E(T) \text{ and } (i,j) \in \stp(\beta(v), \beta(w)) }
\end{equation*}
for all \ $R \in \sig$ and all $i,j \in [\ar(\sig)]$.
Because $T$ is a tree, it is easy to see that $\jtrep{\C}{J}$ is acyclic.

The last two ingredients for our proof of \cref{thm:main-homs} are the following lemmas:
\begin{lemma}\label{lem:homs-structure-to-graph}
	For $\sig$-structures $\A$ and $\C$ and any relaxed join-tree $J = (T, \beta)$ for $\C$ we have: $\hom(\C, \A) \;=\; \hom(\jtrep{\C}{J}, \grep{\A})$.
\end{lemma}
\begin{lemma}\label{lem:homs-tree-to-acyclic-structure}
	Let $\A$ and $\B$ be $\sig$-structures, and let $\str{T}$ be a colored multitree of signature $\repsig$ such that $\hom(\str{T}, \grep{\A}) \neq \hom(\str{T}, \grep{\B})$.
	There exists an acyclic and connected $\sig$-structure $\C$ and a relaxed join-tree $J$ for $\C$ such that $\hom(\jtrep{\C}{J}, \grep{\A}) \neq \hom(\jtrep{\C}{J}, \grep{\B})$.
\end{lemma}
Before proving these lemmas, we first show how to use the four key ingredients for proving \cref{thm:main-homs}.

\begin{proof}[Proof of~\cref{thm:main-homs}]
As pointed out in \cref{sec:RCR}, running RCR on a $\sig$-structure $\A$ produces a stable coloring that is equivalent (via identifying $\at\in\tA$ with $w_\at\in V(\grep{\A})$) to the stable coloring produced by the classical Color Refinement on the colored multigraph $\grep{\A}$.
Thus, RCR distinguishes the $\sig$-structures $\A$ and $\B$ if, and only if, classical Color Refinement distinguishes $\grep{\A}$ and $\grep{\B}$.
According to~\cref{thm:dvorak} the latter is the case if, and only if, there is a colored multitree $\str{T}$ with $\hom(\str{T}, \grep{\A}) \neq \hom(\str{T}, \grep{\B})$.

Hence, for the direction \enquote{1 $\Rightarrow$ 2}, if RCR distinguishes $\A$ and $\B$, then there exists a colored multitree $\str{T}$ with $\hom(\str{T}, \grep{\A}) \neq \hom(\str{T}, \grep{\B})$.
By \cref{lem:homs-tree-to-acyclic-structure}, there also exists an acyclic and connected $\sig$-structure $\C$ and a relaxed join-tree $J$ for $\C$ such that $\hom(\jtrep{\C}{J}, \grep{\A}) \neq \hom(\jtrep{\C}{J}, \grep{\B})$.
According to \cref{lem:homs-structure-to-graph}, this implies that $\hom(\C, \A) \neq \hom(\C, \B)$.

For the direction \enquote{2 $\Rightarrow$ 1}, by assumption there exists an
acyclic and connected $\sig$-structure $\C$ such that $\hom(\C, \A)
\neq \hom(\C, \B)$.
Since $\C$ is acyclic, there exists a  join-tree $T$ for
$\C$. Clearly, $J\deff (T,\beta)$ with $\beta(\ct)\deff\ct$ for all
$\ct\in V(T)=\tC$ is a relaxed join-tree for $\C$.
According to
\cref{lem:homs-structure-to-graph} we have $\hom(\jtrep{\C}{J},
\grep{\A}) \neq \hom(\jtrep{\C}{J}, \grep{\B})$.
By the definition of $\jtrep{\C}{J}$, the Gaifman graph of
$\jtrep{\C}{J}$ is a subgraph of $T$.
In fact, the Gaifman graph of $\C$ being connected implies that the Gaifman graph of $\jtrep{\C}{J}$ is exactly $T$.
Hence, $\jtrep{\C}{J}$ is a colored multitree.
Thus, as pointed out in the first paragraph of the proof, RCR distinguishes $\A$ and $\B$.
\end{proof}
The remainder of this section is devoted to proving the \cref{lem:homs-structure-to-graph,lem:homs-tree-to-acyclic-structure}.
The following notation will be convenient: If $f$ is a mapping from at set $V$ to a set $V'$, and $\at=(a_1,\ldots,a_k)$ is a tuple in $V^k$ for some $k\in\NNpos$ (i.e., $a_i\in V$ for $i\in[k]$), then we write $f(\at)$ for the tuple $(f(a_1),\ldots,f(a_k))$.
\begin{proof}[Proof of~\cref{lem:homs-structure-to-graph}]
We prove the lemma by providing a bijection $\pi$ between the sets $\Hom(\C, \A)$ and $\Hom(\jtrep{\C}{J}, \grep{\A})$.
Recall that $V(\jtrep{\C}{J}) = \V(T)$ and $V(\grep{\A}) = \set{ w_\at \mid \at\in \tA }$.

For all $h \in \Hom(\C, \A)$ let $\pi(h) \isdef h'$, where $h'$ is defined by $h'(v) \isdef w_{h(\ct)}$ for all $v \in \V(T)$ with $\ct \isdef \beta(v)$ (recall that if $\ct$ is of the form $(c_1, \dots, c_k)$, then $h(\ct)$ is the tuple $(h(c_1), \dots, h(c_k))$).
Since $\ct \in \tC$ and $h \in \Hom(\C, \A)$, we obtain that $h(\ct) \in \tA$.
Hence, $h'(v) =w_{h(\beta(v))} = w_{h(\ct)} \in V(\grep{\A})$.
To prove the lemma, it suffices to verify that:
\begin{statements}[label=(\alph*), noitemsep, topsep=1ex] %
	\item\label{lem:homs-structure-to-graph:stmt:proper-homs}
	$\img(\pi) \subseteq \Hom(\jtrep{\C}{J}, \grep{\A})$, i.e.,
	$\pi(h)$ is a homomorphism for every $h \in \Hom(\C, \A)$;
	
	\item\label{lem:homs-structure-to-graph:stmt:injective}
	$\pi$ is injective;\; and

	\item\label{lem:homs-structure-to-graph:stmt:surjective}
	$\pi$ is surjective.
\end{statements}

\myparagraph{Proof of~\cref{lem:homs-structure-to-graph:stmt:proper-homs}.}
Consider an arbitrary $h \in \Hom(\C, \A)$ and let $h' = \pi(h)$.
We have to show that $h'$ is a homomorphism from $\jtrep{\C}{J}$ to $\grep{\A}$.

First consider an $R \in \sigma$ and the relation symbol $U_R$ of arity $1$.
Let $v$ be an arbitrary element in ${(U_R)}^{\jtrep{\C}{J}}$.
We have to show that $h'(v) \in {(U_R)}^{\grep{\A}}$.
By definition of $\jtrep{\C}{J}$, from $v \in {(U_R)}^{\jtrep{\C}{J}}$ we obtain that $\ct \in R^{\C}$ where $\ct \isdef \beta(v)$.
Since $h \in \Hom(\C,\A)$, this implies that $h(\ct) \in R^{\A}$.
By definition of $\grep{\A}$ and $h'$ we obtain: $h'(v) = w_{h(\beta(v))} = w_{h(\ct)}\in {(U_R)}^{\grep{\A}}$.

Next, consider arbitrary $i, j \in [\ar(\sig)]$ and the relation $E_{i,j}$ of arity~2.
Let $(v, w)$ be an arbitrary tuple in ${(E_{i,j})}^{\jtrep{\C}{J}}$.
We have to show that $(h'(v), h'(w)) \in {(E_{i,j})}^{\grep{\A}}$.
Let $\bt \isdef \beta(v)$ and $\ct \isdef \beta(w)$.
By definition of $\jtrep{\C}{J}$, from $(v, w) \in
{(E_{i,j})}^{\jtrep{\C}{J}}$ we obtain that
$\set{v,w}\in E(T)$
and $(i,j) \in \stp(\bt,\ct)$.
Hence, $b_i = c_j$, where $b_i$ is the $i$-th component $\bt$ and $c_j$ is the $j$-th component of $\ct$.
Clearly, $h(b_i) = h(c_j)$, and therefore the $i$-th component of the tuple $h(\bt)$ is equal to the $j$-th component of the tuple $h(\ct)$.
Hence, $(i,j) \in \stp(h(\bt), h(\ct))$.
Furthermore, $\bt, \ct \in \tC$, and hence $h(\bt), h(\ct) \in \tA$.
By definition of $\grep{\A}$ all this implies that $(w_{h(\bt)},w_{h(\ct)}) \in {(E_{i,j})}^{\grep{\A}}$.
We are done by noting that $(h'(v), h'(w)) = (w_{h(\beta(v))},w_{h(\beta(w))}) = (w_{h(\bt)}, w_{h(\ct)})$.
This proves~\cref{lem:homs-structure-to-graph:stmt:proper-homs}.
\bigskip

\myparagraph{Proof of~\cref{lem:homs-structure-to-graph:stmt:injective}.}
Let $h_1, h_2 \in \Hom(\C, \A)$ with $h_1 \neq h_2$.
Let $h'_1 = \pi(h_1)$ and $h'_2 = \pi(h_2)$; our aim is to show that $h'_1 \neq h'_2$.

Since $h_1 \neq h_2$, there exists a $u \in \V(\C)$ such that $h_1(u) \neq h_2(u)$.
By definition of $\C$ (recall our assumption on $\sig$-structures
described in \cref{sec:preliminaries}),
there exists an $R \in \sig$ and a tuple $\ct \in R^{\C}$ such that $u \in \tset(\ct)$.
Clearly, $h_1(u) \neq h_2(u)$ implies that $h_1(\ct) \neq h_2(\ct)$.
By the definition of relaxed join-trees, there exists a 
$v \in \V(T)$ such that $\beta(v) = \ct$.
Thus, $h'_1(v) = w_{h_1(\beta(v))}=w_{h_1(\ct)} \neq w_{h_2(\ct)}=w_{h_2(\beta(v))} = h'_2(v)$.
Hence, $h'_1 \neq h'_2$.
This proves~\cref{lem:homs-structure-to-graph:stmt:injective}.
\bigskip

\myparagraph{Proof of~\cref{lem:homs-structure-to-graph:stmt:surjective}.}
Let $h'' \in \Hom(\jtrep{\C}{J}, \grep{\A})$.
Our aim is to find an $h \in \Hom(\C,\A)$ such that $h'' = \pi(h)$.
By definition of $\C$ (recall our assumption on $\sig$-structures described in \cref{sec:preliminaries}), for every $z \in V(\C)$ there exists an $R \in \sig$ and a tuple $\ct \in R^\C$ such that $z \in \tset(\ct)$.
For each $z \in V(\C)$ let us choose arbitrary, but from now on fixed such $R$ and $\ct$ which we henceforth will denote by $R_z$ and $\ct_z$, and let us fix an $i_z \in [\ar(R_z)]$ such that $z$ is the $i_z$-th component of the tuple $\ct_z$. 
By definition of $J=(T,\beta)$, there is at least one node $v$ of $T$ such that $\beta(v) = \ct_z$. We choose an arbitrary, but from now on fixed such $v$ and denote it by $v_z$.
Since $\ct_z \in {(R_z)}^{\C}$, by definition of $\jtrep{\C}{J}$ we have $v_{z} \in {(U_{R_z})}^{\jtrep{\C}{J}}$.
Since $h'' \in \Hom(\jtrep{\C}{J}, \grep{\A})$, we obtain $h''(v_{z}) \in {(U_{R_z})}^{\grep{\A}}$.
By definition of $\grep{\A}$ there is a tuple $\at_{z} \in {(R_z)}^\A$ such that $h''(v_{z})= w_{\at_{z}}$.
Let us write $x_{z}$ to denote the $i_z$-th component of the tuple $\at_{z}$.
Clearly, $x_{z} \in \V(\A)$.
Using these notions, we define the mapping $h\colon V(\C)\to V(\A)$ by letting $h(z) \isdef x_z$ for every $z \in \V(\C)$.
\medskip

\noindent
\emph{Claim 1:} 
	For all $v \in \V(T)$ with $\ct \in \tC$ and $\at \in \tA$ such that $\ct = \beta(v)$ and $h''(v) = w_\at$ we have: $h(\ct) = \at$.
\medskip

\noindent
\emph{Proof:}
Consider $v \in \V(T)$ and let $\ct \in \tC$ and $\at \in \tA$ such that $\ct = \beta(v)$ and $h''(v_\ct) = w_\at$.
By definition of $\jtrep{\C}{J}$ and $\grep{\A}$, and since $h'' \in \Hom(\jtrep{\C}{J},\grep{\A})$, we have: $\ar(\ct) = \ar(\at)$.
Let $k = \ar(\ct)$ and let $\ct$ and $\at$ be of the form $\ct = (c_1, \ldots, c_k)$ and $\at = (a_1, \ldots, a_k)$.
Consider an arbitrary $\ell \in [k]$.
We have to show that $h(c_\ell) = a_\ell$.

Let $z \isdef c_\ell$.
By definition of $h$ we have: $h(c_\ell) = h(z) = x_z$.
Here, $x_z$ is the $i_z$-th component of the tuple $\at_z \in \tA$;
furthermore, $w_{\at_z} = h''(v_{z})$, where $\beta(v_z) = \ct_z$ is a tuple in $\tC$ such that $z$ is the $i_z$-th component of $\ct_z$.

Since $z \in \tset(\ct)$ and $z \in \tset(\ct_z)$, the connectivity
condition of relaxed join-trees implies that $z \in \tset(\beta(w))$ for every node $w$ of $T$ that lies on the path in $T$ from $v$ to $v_z$.
Let $(w_0, \ldots, w_m)$ be this path, starting with $w_0 = v$ and ending with $w_m = v_z$ for some $m \geq 0$ and let $\ct^{(j)} \isdef \beta(w_j)$ for $0 \leq j \leq m$.
Let $i_0 \isdef \ell$ and $i_m \isdef i_z$ and for every $j$ with $1 \leq j<m$ consider any $i_j$ such that $z$ occurs in the $i_j$-th component of $\ct^{(j)}$. Note that this way for all $0 \leq j \leq m$ it holds that $z$ occurs in the $i_j$-th component of $\ct^{(j)}$.
For every $j \in [m]$ we have: $\set{ w_{j-1}, w_{j} } \in \E(T)$ and $(i_{j-1},i_j) \in \stp(\ct^{(j-1)}, \ct^{(j)})$; and hence, by definition of $\jtrep{\C}{J}$ we have $(w_{j-1},w_{j}) \in {(E_{i_{j-1},i_j})}^{\jtrep{\C}{J}}$.
Since $h'' \in \Hom(\jtrep{\C}{J},\grep{\A})$ we obtain:
$(h''(w_{j-1}),h''(w_{j})) \in {(E_{i_{j-1},i_j})}^{\grep{\A}}$.

For each $j \in \set{ 0, \ldots, m }$ let $\at^{(j)} \in \tA$ such that $h''(w_{j}) = w_{\at^{(j)}}$.
Note that $\at^{(0)} = \at$ and $\at^{(m)} = \at_z$.
We have shown for all $j \in [m]$ that $(w_{\at^{(j-1)}}, w_{\at^{(j)}}) \in {(E_{i_{j-1},i_j})}^{\grep{\A}}$ which by the definition of $\grep{\A}$ implies that $(i_{j-1},i_j) \in\stp(\at^{(j-1)},\at^{(j)})$.
Hence, the $i_0$-th component of $\at^{(0)}$ is the same as the $i_1$-th component of $\at^{(1)}$, \dots, is the same as the $i_m$-th component of $\at^{(m)}$.
Recall that $i_0 = \ell$, $\at^{(0)} = \at$, $i_m = i_z$, and $\at^{(m)} = \at_z$.
Hence, the $\ell$-th component of $\at$ (i.e., $a_\ell$) is the same as the $i_z$-th component of $\at_z$ (i.e., $x_z$).
Thus, $a_\ell = x_z = h(z) = h(c_\ell)$.

This completes the proof of Claim~1.\qed
\bigskip

\noindent
\emph{Claim 2:} $h\in\Hom(\C,\A)$.
\medskip

\noindent
\emph{Proof:}
Let $R \in \sig$ and let $\ct \in R^{\C}$ and consider $v \in \V(T)$ with $\beta(v) = \ct$. 
Then, $v \in {(U_R)}^{\jtrep{\C}{J}}$ and since $h'' \in \Hom(\jtrep{\C}{J},\grep{\A})$, we obtain $h''(v) \in {(U_R)}^{\grep{\A}}$.
By definition of $\grep{\A}$ there exists an $\at \in R^\A$ such that $h''(v_{\ct}) = w_\at$.
Claim 1 yields $h(\ct) = \at$.
Hence, $h(\ct) \in R^\A$.
This completes the proof of Claim~2.\qed
\bigskip

To complete the proof of \cref{lem:homs-structure-to-graph} it suffices to show that $\pi(h) = h''$.
By definition of $\pi$ we have: $\pi(h) = h'$, where $h'$ is defined by $h'(v) \isdef w_{h(\beta(v))}$ for all $v \in \V(T)$.
From Claim~1 we obtain that $h''(v) = h'(v)$ for all $v \in \V(T)$. Hence, $h'' = h' = \pi(h)$.
This completes the proof of~\cref{lem:homs-structure-to-graph:stmt:surjective} and the proof of \cref{lem:homs-structure-to-graph}.\qedhere \end{proof}
\begin{proof}[Proof of~\cref{lem:homs-tree-to-acyclic-structure}]
By assumption, $\hom(\str{T}, \grep{\A}) \neq \hom(\str{T}, \grep{\B})$, where $\str{T}$ is a colored multitree of signature $\repsig$ and $\A$, $\B$ are $\sig$-structures.
We will show that the homomorphisms from $\str{T}$ provide us with \enquote{templates} for acyclic and connected $\sig$-structures, one of which must have a number of homomorphisms into $\A$ that is different from its number of homomorphisms into $\B$.

Let us write $T$ to denote the Gaifman graph of $\str{T}$.
Since $\str{T}$ is a colored multitree, $T$ is a tree.
For an  $h \in \Hom(\str{T}, \grep{\A})$ the \emph{print $P_h$ of $h$ in $\grep{\A}$} is the colored multigraph of signature $\repsig$ defined by $\V(P_h) \isdef \V(\str{T})=\V(T)$ and, for all $R\in\sig$ and all $i,j\in[\ar(\sig)]$:
\begin{equation*}
	\begin{array}{rcl}
		{(U_R)}^{P_h} & \deff  & \set{\, v\in \V({T}) \,\mid\, h(v) \in {(U_R)}^{\grep{\A}} \,}\,,
		\smallskip\\
		{(E_{i,j})}^{P_h} & \deff & \set{\, (u, v) \,\mid \, (h(u), h(v)) \in {(E_{i,j})}^{\grep{\A}} \text{ \;and \;} \big(\; u{=}v \text{ \,or \,}\set{u,v}\in \E(T)\;\big)\;}\,.
	\end{array} 
\end{equation*}
Note that $T$ is also the Gaifman graph of $P_h$.
We let $P_\A=\setc{P_h}{h\in\Hom(\str{T}, \grep{\A})}$.
The notion of the print $P_h$ of $h$ in $\grep{\B}$ for $h\in\Hom(\str{T},\grep{\B})$ and the set $P_\B$ are defined analogously.
Note that $P_\A$ and $P_\B$  are not necessarily disjoint, and that different homomorphisms may have the same print.

For every $P\in P_\A\cup P_\B$ we let $\#(P, \A)\deff |\setc{h\in\Hom(\str{T}, \grep{\A})}{P_h=P}|$.
The number $\#(P,\B)$ is defined analogously.
Notice that $\hom(\str{T}, \grep{\A}) = \sum_{P\in P_\A} \#(P, \A)$ and $\hom(\str{T}, \grep{\B}) = \sum_{P \in P_\B} \#(P, \B)$ --- see~\cref{claim:homomorphisms:SecondLemma:FirstClaim} in~\cref{sec:appendix:homomorphisms:SecondLemma} for a proof.

For any two prints $P, P'$ we say that $P$ is a \emph{subprint of $P'$} (for short: $P \preceq P'$) if ${(U_R)}^{P} \subseteq {(U_R)}^{P'}$ and ${(E_{i,j})}^{P} \subseteq {(E_{i,j})}^{P'}$, for all $R \in \sig$ and all $i,j \in [\ar(\sig)]$.
Obviously, $\preceq$ is a partial order on $P_\A \union P_\B$.

It can be verified that for every print $P$ we have $\hom(P, \grep{\A}) = \sum_{P': P \preceq P'} \#(P', \A)$ and $\hom(P, \grep{\B}) = \sum_{P':P\preceq P'} \#(P', \B)$ --- see~\cref{claim:homomorphisms:SecondLemma:SecondClaim} in~\cref{sec:appendix:homomorphisms:SecondLemma} for a proof.

Since $\hom(\str{T}, \grep{\A}) = \sum_{P\in P_\A} \#(P, \A)$ and $\hom(\str{T}, \grep{\B}) = \sum_{P \in P_\B} \#(P, \B)$ and, by assumption, \allowbreak{} $\hom(\str{T}, \grep{\A})\neq\hom(\str{T}, \grep{\B})$, there must be a $P\in P_\A\cup P_\B$ such that $\#(P, \A) \neq \#(P, \B)$.
We choose a largest such $P$ w.r.t.\ the partial order $\preceq$.
I.e., $\#(P, \A) \neq \#(P, \B)$, but $\#(P', \A) = \#(P', \B)$ for all $P'$ with $P\preceq P'$ and $P'\neq P$.
Combining this with the fact that $\hom(P, \grep{\A}) = \sum_{P': P \preceq P'} \#(P', \A)$ and $\hom(P, \grep{\B}) = \sum_{P':P\preceq P'} \#(P', \B)$, we obtain: $\hom(P, \grep{\A}) \neq \hom(P, \grep{\B})$.

Now, all that remains to be done is to show that there exists an acyclic and connected $\sig$-structure $\C$ and a relaxed join-tree $J$ for $\C$ ($T$ will be the tree associated with $J$) such that $\jtrep{\C}{J}$ is isomorphic to $P$ --- then $\hom(P, \grep{\A}) \neq \hom(P, \grep{\B})$ implies that $\hom(\jtrep{\C}{J}, \grep{\A}) \neq \hom(\jtrep{\C}{J}, \grep{\B})$ and the proof is complete.

Recall that $T$ is the Gaifman graph of $P$.
We pick an arbitrary vertex $r\in\V(T)$ as the root, and we let $\dirT$ be the directed tree obtained from $T$ by orienting the edges to be directed \emph{away} from the root.

We use $P$ to assign to each $v\in \V(T)$ an atomic type $\myrho_v$, an arity $k_v$, and a similarity type $\mytau_v$ as follows.
Let $\myrho_v\deff \setc{R\in\sig}{v\in {(U_R)}^P}$.
Note that $\myrho_v\neq\emptyset$ and all $R\in\myrho_v$ have the same arity; we denote this arity by $k_v$.
Let $\mytau_v\deff\setc{(i,j)}{i,j\in[k_v], \ (v,v)\in {(E_{i,j})}^P}$.
Furthermore, we use $P$ to assign to each edge $(u,v)\in \E(\dirT)$ a similarity type $\mytau_{(u,v)}$ via $\mytau_{(u,v)}\deff \setc{(i,j)\in [k_u]\times [k_v]}{(u,v)\in {(E_{i,j})}^{P}}$.

We start with an empty $\C$, and we perform a top-down pass of $\dirT$ during which we insert elements and tuples into $\C$:
For the root $r$ of $\dirT$ we create a tuple $\tup{t}_r$ of arbitrary values and of arity $k_r$, such that $\stp(\tup{t}_r) = \mytau_r$ and $\atp(\tup{t}_r)=\myrho_r$.
Then, whenever we encounter a node $v$ of $\dirT$, we consider this node's parent $p$ in $\dirT$ (for which the tuple $\tup{t}_p$ has already been created), and create a (not necessarily new) tuple $\tup{t}_v$ and ensure that $\stp(\tup{t}_p,\tup{t}_v)=\mytau_{(p,v)}$ and $\stp(\tup{t}_v)=\mytau_v$ and $\atp(\tup{t}_v)=\myrho_v$.
Note that $\mytau_{(p,v)}$ and $\mytau_v$ provide us with the information on which components of $\tup{t}_p$ have to be copied to which components in $\tup{t}_v$, and which components in $\tup{t}_v$ have to carry the same entries.
To fill the components that do not inherit values from $\tup{t}_p$, we simply invent new values.

The details are carried out as follows:
We start with the root $r$ of $\dirT$.
Let $k\deff k_r$.
Let $\tup{t}_r=(u_1,\ldots,u_k)$ be a tuple of arbitrary values such that $\stp(\tup{t}_r) = \mytau_r$.
We insert $u_1,\ldots,u_k$ into the universe $V(\C)$ of $\C$; and for each $R\in\myrho_r$ we insert $\tup{t}_r$ into $R^{\C}$ (thus, $\atp(\tup{t}_r)=\myrho_r$).

During the top-down pass of $\dirT$, when we reach a node $v\in V(T)$ we proceed as follows.
Let $p$ be the parent of $v$ in $\dirT$, i.e., $(p,v)\in \E(\dirT)$.
We already have created a tuple $\tup{t}_p$ for $p$ such that $\ar(\tup{t}_p)=k_p$ and $\stp(\tup{t}_p)=\mytau_p$.
Our goal is to create a tuple $\tup{t}_v$ of arity $k_v$ such that $\stp(\tup{t}_p,\tup{t}_v)=\mytau_{(p,v)}$ and $\stp(\tup{t}_v)=\mytau_v$.

To see that this is indeed possible, note that $P\in P_\A\cup P_\B$, and hence $P=P_h$ for a homomorphism $h$ from $\str{T}$ to $\A$ or $\B$.
Thus, according to the definition of $P_h$, there exist tuples $\at,\bt$ (either both in $\tA$ or both in $\tB$) such that $\stp(\at)=\mytau_p$ and $\stp(\at,\bt)=\mytau_{(p,v)}$ and $\stp(\bt)=\mytau_v$.
Hence, $\mytau_p$, $\mytau_{(p,v)}$, $\mytau_v$ are consistent in the sense that we can proceed as follows.

Let $(u_1,\ldots,u_{k_p})\deff\tup{t}_p$ and set  $\tup{t}_v=(w_1,\ldots,w_{k_v})$, where $w_j\deff u_i$ for all $(i,j)\in\mytau_{(p,v)}$ and $w_j\deff \bot$ for all the remaining $j$.
Note that for all $(i,j) \in \mytau_v$ where $w_i \neq \bot$ or $w_j \neq \bot$, it already holds that $w_i = w_j$; the reasoning is as follows.
Consider the case where $w_i\neq \bot$ (the case where $w_j\neq\bot$ is analogous).
Because $w_i\neq\bot$, there is an $\hati$ such that $(\hati, i) \in \mytau_{(p,v)}$.
Recall that $\mytau_{(p,v)} = \stp(\at, \bt)$ for the $\at, \bt$ from above, hence, $a_{\hati} = b_{i}$.
Combined with $\mytau_v = \stp(\bt)$, this yields $a_{\hati} = b_i = b_j$.
Hence, $(\hati, j) \in \stp(\at,\bt)$, i.e., $(\hati, j) \in \mytau_{(p,v)}$.
Thus, $w_j = u_{\hati}=w_i$.

Let $I$ be the set of remaining $j$ where $w_j=\bot$.
Note that $I$ could be empty.
For each $j\in I$ let $w_j\deff\bot_j$ for a new symbol $\bot_j$.
Afterwards, repeatedly replace $w_j$ by $w_i$ for all $(i,j)\in\mytau_v$ with $i<j$ and $i,j\in I$.
Finally, let $X$ be the set of symbols $\bot_j$ that still occur in $\tup{t}_v$.
We treat these as $|X|$ distinct new symbols and insert them into $V(\C)$.
Furthermore, for every $R\in \myrho_v$ we insert $\tup{t}_v$ into $R^\C$ (thus, $\atp(\tup{t}_v)=\myrho_v$).
Note that $\stp(\tup{t}_v)=\mytau_v$ and $\stp(\tup{t}_p,\tup{t}_v)=\mytau_{(p,v)}$.

After having finished the top-down pass of $\dirT$, we have produced a $\sig$-structure $\C$.
The following is straightforward to see:
Firstly, $J \isdef (T, \beta)$, where $\beta(v) \isdef \tup{t}_v$ for all $v \in \V(T)$,
is a relaxed join-tree for $\C$.
In particular, by \cref{prop:relaxedjointree}, $\C$ is acyclic.
Secondly, the Gaifman graph of $\C$ is connected.
Hence, $\C$ is an acyclic and connected $\sig$-structure.
Finally, $\jtrep{\C}{J}$ is isomorphic to $P$.
This completes the proof of \cref{lem:homs-tree-to-acyclic-structure}.\qedhere \end{proof} 
	
	\section{Connection to Logic}\label{sec:main-logic}
	This section's goal is to provide a logical characterization of the distinguishing power of Relational Color Refinement.
	We aim for a theorem that is analogous to the following result due to Immerman and Lander~\cite{Immerman1990} and Cai, Fürer, Immerman~\cite{Cai1992}.%
	\begin{theorem}[\cite{Cai1992,Immerman1990}]\label{thm:CFIL}
		Let $G$ and $H$ be graphs.
		The following statements are equivalent:
		\begin{statements}[topsep=0pt, noitemsep]
			\item Color Refinement distinguishes $G$ and $H$.
			\item There exists a sentence $\phi \in \CFO^2$ such that $G \models \phi$ and $H \not\models \phi$.
		\end{statements}
	\end{theorem}
	\noindent
	Here, $\CFO^2$ is the restriction of $\CFO$ to two variables, and $\CFO$ is a syntactic extension of first-order logic with counting quantifiers of the form $\existsgeq{n} x\, \psi$ (for every fixed $n\in\NNpos$), expressing \enquote{there exist at least $n$ values for $x$ such that $\psi$ holds}.
	\subsection{The Guarded Fragment of Counting Logic}\label{sec:gfc}

This section introduces the \emph{guarded fragment of the logic $\CFO$}, for short: $\GFC$.
Its definition is in the same spirit as the logic $\Logic{GF}$ (the guarded fragment of first-order logic; see~\cite{Graedel1999}) and the logic $\Logic{GF(L)}$ (the guarded fragment of any logic $\Logic{L}$ that is a subset of first-order logic; see~\cite{Gottlob2003}).
Here, we use a similar notation as in~\cite{Gottlob2003}, but adapt it in order to obtain a reasonable notion of \enquote{guarded fragment of $\CFO$}.

Let $\VAR \isdef \set{ \var_i \mid i \in \natpos }$ be the set of \emph{variables}.
We call a tuple $\vartup$ of $m$ \emph{distinct} variables of the form $(\var_{i_1}, \dots, \var_{i_m}) \in \VAR^{m}$ a \emph{variable tuple}, and we let $\vset(\vartup) \isdef \set{ \var_{i_1}, \dots, \var_{i_m} }$.
Recall that at the beginning of \cref{sec:RCR} we have chosen an arbitrary (relational) signature that is fixed throughout the rest of this paper.

\begin{definition}[Syntax of $\GFC$]\label{def:logic-syntax}
	The logic $\GFC$ is inductively defined along with the \emph{free variables} and the \emph{guard-depth}, formalized by the functions $\free\colon \GFC \to \mypot{\VAR}$ and $\gd\colon \GFC \to \nat$.
	\begin{description}
		\settowidth{\myA}{$\existsgeq{n}\, \vartup \qsep (\LogGuard \land \psi)$\; }%
		\item[Atomic Formulas:] For all $R \in \sig$ with $\ell \isdef \ar(R)$, all $x_1, \dots, x_\ell \in \VAR$ and all $x,y \in \VAR$, the following formulas $\phi$  (of signature $\sig$) are in $\GFC$: \ $\phi$ is of the form
		\begin{rules}[resume=logic-syntax,noitemsep,topsep=0.5ex]
			\item\label{def:logic-syntax:rules:relation}
			\makebox[\myA][l]{$R(x_1, \dots, x_\ell)$ } 
			with\,
			$\free(\phi) \isdef \set{ x_1, \dots, x_\ell }$ 
			\,and\,
			$\gd(\phi) = 0$;

			\item\label{def:logic-syntax:rules:equality}
			\makebox[\myA][l]{$x \lequal y$ } 
			with\,
			$\free(\phi) \isdef \set{ x, y }$ 
			\,and\,
			$\gd(\phi) = 0$.
		\end{rules}
		\item[Inductive Rules:] Let $\chi, \psi$ be formulas (of signature $\sig$) in $\GFC$.
		The following formulas $\phi$ (of signature $\sig$) are in $\GFC$: \ $\phi$ is of the form
		\begin{rules}[resume=logic-syntax, noitemsep,topsep=0.5ex]
			\item\label{def:logic-syntax:rules:negation}
			\makebox[\myA][l]{ $\lnot \chi$ }
			with\,
			$\free(\phi) \isdef \free(\chi)$
			\,and\,
			$\gd(\phi) = \gd(\chi)$;

			\item\label{def:logic-syntax:rules:conjunction}
			\makebox[\myA][l]{ $(\chi \land \psi)$ } 
			with\,
			$\free(\phi) \isdef \free(\chi) \union \free(\psi)$
			\,and\,
			$\gd(\phi) = \max(\gd(\chi), \gd(\psi))$.
		\end{rules}
		An atomic formula $\LogGuard$ (of signature $\sig$) of the form $R(x_1,\ldots,x_\ell)$ in $\GFC$ is called \emph{a guard for $\psi$}, if $\free(\psi) \subseteq \free(\LogGuard)$.
		Let $n \in \natpos$ and let $\LogGuard$ be a guard for $\psi$.
		For every \emph{variable tuple} $\vartup$ with $\vset(\vartup) \subseteq \free(\LogGuard)$, the following formula $\phi$ (of signature $\sig$) is in $\GFC$: \ $\phi$ is of the form
		\begin{rules}[resume=logic-syntax]
			\item\label{def:logic-syntax:rules:quantification}
			\makebox[\myA][l]{
				$\existsgeq{n}\, \vartup \qsep (\LogGuard \land \psi)$
			}
			with\,
			$\free(\phi) \isdef \free(\LogGuard) \setminus \vset(\vartup)$
			\,and\,
			$\gd(\phi) = \gd(\psi) + 1$.
			\endDefinition%
		\end{rules}
	\end{description}
\end{definition}
In this paper we assume w.l.o.g.\ that the variable tuple $\vartup = (\var_{i_1}, \dots, \var_{i_m})$ after a quantifier $\existsgeq{n}$ is \emph{ordered}, i.e., $i_1 < \cdots < i_m$.
This has no effect on the semantics, but simplifies some arguments.
We write $\existseq{n}\, \vartup \qsep (\LogGuard \land \phi)$ as shorthand for $\big(\,\existsgeq{n}\, \vartup \qsep (\LogGuard \land \phi) \ \;\land \;\lnot \,\existsgeq{n{+}1}\, \vartup \qsep (\LogGuard \land \phi)\,\big)$.
We omit parentheses in the usual way.

\begin{definition}[Semantics of $\GFC$]\label{def:logic-semantics}
	A \emph{$\sig$-interpretation} is a tuple $\Int = (\A, \assignment)$ consisting of a $\sig$-structure $\A$ and a function $\assignment: \VAR \to V(\A)$.
	Formulas (of signature $\sig$) in $\GFC$ are evaluated on $\sig$-interpretations $\Int$.
	We write $\Int \models \phi$ to denote that $\Int$ \emph{satisfies} $\phi$, and $\Int \not\models \phi$ to denote that $\Int$ does not satisfy $\phi$.
	By $\Isubst{(a_1, \dots, a_\ell)}{(\var_{i_1}, \dots, \var_{i_\ell})}$ we denote the $\sig$-interpretation $(\A, \assignment')$ with $\assignment'(\var_{i_j}) \isdef a_j$ for all $j\in[\ell]$, and $\assignment'(x) \isdef \assignment(x)$ for all $x \in \VAR \setminus \set{ \var_{i_1}, \dots, \var_{i_\ell} }$.
	The semantics of formulas in $\GFC$ are inductively defined as follows:
	\begin{description}[font={\mdseries\itshape},noitemsep,topsep=1ex]
		\item[\Cref{def:logic-syntax:rules:relation}:]
		$\Int \models R(x_1, \dots, x_\ell) \iff (\assignment(x_1), \dots, \assignment(x_\ell)) \in R^\A$.

		\item[\Cref{def:logic-syntax:rules:equality}:]
		$\Int \models x \lequal y \iff \assignment(x) = \assignment(y)$.

		\item[\Cref{def:logic-syntax:rules:negation,def:logic-syntax:rules:conjunction}:]
		$\Int \models \lnot \chi \iff \Int \not\models \chi$.
		\;
		$\Int \models (\chi \land \psi) \iff \Int \models \chi$ and $\Int \models \psi$.

		\item[\Cref{def:logic-syntax:rules:quantification}:]
		$\Int \models \existsgeq{n} \vartup \qsep (\LogGuard \land \psi)$ $\iff$ there are at least $n$ tuples $\at\in {V(\A)}^{\ar(\vartup)}$ such that $\Isubst{\at}{\vartup} \models (\LogGuard \land \psi)$.
		\endDefinition%
	\end{description}
\end{definition}
We will use the following conventions throughout the paper:
$\phi(x_1, \dots, x_k)$ denotes that $\free(\phi) \subseteq \set{ x_1, \dots, x_k }$; and $\A, (a_1, \dots, a_k) \models \phi(x_1, \dots, x_k)$ denotes that $(\A, \assignment) \models \phi$ where $\assignment$ is \emph{any} assignment with $\assignment(x_i) \isdef a_i$ for all $i \in [k]$.
A \emph{sentence} is a formula $\phi \in \GFC$ that has no free variables, i.e., $\free(\phi) = \emptyset$.
If $\phi$ is a sentence, we write $\A \models \phi$ to denote that $(\A, \assignment) \models \phi$ for \emph{any} assignment $\assignment$ (since $\assignment$ does not matter in this case).
We write $\A \equiv_{\GFC} \B$ to denote that the $\sig$-structures $\A$ and $\B$ satisfy the same sentences (of signature $\sig$) in $\GFC$.
Finally, we say that $\A$ and $\B$ are \emph{distinguishable} in $\GFC$, if $\A \not\equiv_{\GFC} \B$.
\begin{example}

Consider the formula $
	\phi \isdef \existsgeq{1}(\var_1, \var_2, \var_3, \var_4, \var_5, \var_6) \qsep \bigl(
			R(\var_1, \var_2, \var_3, \var_4, \var_5, \var_6)
			\land \bigl(
				E(\var_1, \var_2) 
				\land (\,E(\var_2, \var_3) 
				\land E(\var_3, \var_1)\,)
			\bigr)
		\bigr).
$
Clearly, $\psi \isdef \bigl( E(\var_1, \var_2) \land (\,E(\var_2, \var_3) \land E(\var_3, \var_1)\,)\bigr)$ is a formula in $\GFC$ with $\free(\psi) = \set{ \var_1, \var_2, \var_3 }$ and $\gd(\psi) = 0$.

Thus, $R(\var_1, \var_2, \var_3, \var_4, \var_5, \var_6)$ is a guard for $\psi$.
Hence, $\phi \in \GFC$ with $\free(\phi) = \emptyset$, and $\gd(\phi) = 1$.
The formula $\phi$ states that there is at least one tuple (of arity $6$) in $R$ such the first 3 entries of this tuple form a triangle w.r.t.\ relation $E$.
Hence, $\SpExa \models \phi$ but $\SpExb \not\models \phi$, where $\SpExa$, $\SpExb$ are the structures from \cref{ex:simple-structures}.
Thus, $\SpExa$ and $\SpExb$ are distinguishable in $\GFC$.\endExample{} \end{example}

\subsection{The Guarded-Game}\label{sec:game}
Our ultimate goal in \cref{sec:main-logic} is to prove for any two $\sig$-structures $\A$ and $\B$ that RCR distinguishes $\A$ and $\B$ $\iff$ $\A\not\equiv_{\GFC} \B$.
Similarly to the proof of \cref{thm:CFIL}, our proof will use, as an intermediate step, a game characterization of (in)distinguishability of two $\sig$-structures in $\GFC$. %
We call this game the \emph{Guarded-Game}; it is defined as follows.
It is played on two $\sig$-structures $\A$, $\B$.
A \emph{configuration} of the Guarded-Game is a tuple of the form $((\A, \at), (\B, \bt))$, where $\A$ and $\B$ are the given $\sig$-structures and $\at \in {V(\A)}^k$, $\bt \in {V(\B)}^k$ for some $k \in \nat$.
A configuration $((\A, \at), (\B, \bt))$ is called \emph{distinguishing}, if $\stp(\at) \neq \stp(\bt)$ or there are a $\ell\in[\ar(\sig)]$ and indices $i_1,\ldots,i_\ell\in[k]$ such that $\atp((a_{i_1}, \dots, a_{i_\ell})) \neq \atp((b_{i_1}, \dots, b_{i_\ell}))$.
We may omit parentheses if they are clear from the context.
If $k=0$, we write $\A, \B $ for the configuration $((\A, ()), (\B, ()))$, and we call this the \emph{empty} configuration; note that this configuration is \emph{not} distinguishing.

A \emph{round} of the Guarded-Game is played as follows:
Consider $(\A, \at)$, $(\B, \bt)$ to be the configuration at the beginning of the round.
Spoiler picks a relation symbol $R \in \sig$.
Then, Duplicator provides a bijection $\pi$ between $R^\A$ and $R^\B$.
Finally, Spoiler picks some $\at' \in R^\A$ and creates the new configuration $(\A, \at')$, $(\B, \bt')$ where $\bt' \isdef \pi(\at')$.
Duplicator wins this round if the new configuration is not distinguishing and $\stp(\at, \at') = \stp(\bt, \bt')$.
Otherwise, Spoiler has won.
In particular, Spoiler wins the round if Duplicator is unable to provide a bijection because $\card{R^\A} \neq \card{R^\B}$.

Duplicator has a $0$-round \emph{winning strategy} on $(\A, \at)$, $(\B, \bt)$ if the configuration is not distinguishing.
For $i\geq 1$, Duplicator has an $i$-round winning strategy on $(\A, \at)$, $(\B, \bt)$ if this configuration is not distinguishing, and she can provide a bijection $\pi$ for every $R \in \sig$ that Spoiler may pick, such that for every $\at'$ and $\bt' \deff \pi(\at')$ that Spoiler may choose, she wins the current round and has an $(i{-}1)$-round winning strategy on the resulting configuration $(\A, \at')$, $(\B, \bt')$.
Spoiler has an $i$-round winning-strategy on $(\A, \at)$, $(\B, \bt)$, if Duplicator does not have one.
In particular, if Spoiler has a winning strategy for $i$ rounds, he also has a winning strategy for more than $i$ rounds.
If $\A$ and $\B$ are not of strictly equal size, then Spoiler has a trivial $1$-round winning strategy, because Duplicator is unable to give a bijection in the first round.
We say that Duplicator \emph{wins} the Guarded-Game on $(\A, \at)$, $(\B, \bt)$ if she has an $i$-round winning strategy on $(\A, \at)$, $(\B, \bt)$ for every $i \in \nat$.
\begin{example}

The pair $(\SpExa,\, (1,2))$, $(\SpExb, (u,v))$ is a configuration.
Note that $\stp((1,2)) = \set{ (1,1), (2,2) }[] = \stp((u,v))$ and $\atp((1,2)) = \set{ E }[] = \atp((u,v))$.
Furthermore, $\atp(\at') = \emptyset$ for every tuple $\at'$ of arity $\leq 6 = \ar(\sigEx)$ with entries in $\set{ 1, 2 }[]$ or $\set{ u, v }[]$ such that $\at'\neq (1,2)$ and $\at' \neq (u,v)$.
Therefore, this configuration \emph{not} distinguishing.
Hence, Duplicator has a 0-round winning strategy on $(\SpExa, (1,2))$, $(\SpExb, (u,v))$.

Consider the (empty) configuration $\SpExa$, $\SpExb$, which is obviously not distinguishing.
Hence, Duplicator has a 0-round winning strategy.
But, Spoiler has a 1-round winning strategy that goes as follows:
Spoiler picks $R$ as relation.
Since $R^{\SpExa} = \set{ (1,2,3,u,v,w) }[] = R^{\SpExb}$, Duplicator can only give the identity as bijection.
Spoiler creates the configuration $(\SpExa, (1,2,3,u,v,w))$, $(\SpExb, (1,2,3,u,v,w))$.
Since the atomic type of $(2,3)$ is $\set{ E }[]$ in $\SpExa$ and $\emptyset$ in $\SpExb$, this configuration is distinguishing.
Thus, Spoiler wins this round, and wins the Guarded-Game on $\SpExa$, $\SpExb$.%
\endExample{} \end{example}
 	\subsection{Relational Color Refinement is equivalent to \texorpdfstring{$\GFC$}{GF(C)}} %
In this section we prove \cref{thm:main-logic}.
\mainlogic{}
We prove the theorem by showing that the implication chain $1 \Rightarrow 2 \Rightarrow 3 \Rightarrow 1$ holds.
For this, we use the following three lemmas;
their proofs are inductive and quite similar to the way the analogous result on graphs is shown, thus we defer their proofs to~\cref{app:main-logic}.
The \emph{arity} $\ar(c)$ (atomic type $\atp(c)$, similarity type $\stp(c)$) of a color $c\in \cols{i}{\A}$ is defined as the arity (atomic type, similarity type) of the tuples in $\tA$ that receive this color.
\begin{restatable}{lemma}{describingColorInGFC}\label{lem:describing-color-in-gfc}
	For every $\sig$-structure $\A$, every $i \in \nat$, and every $c \in \cols{i}{\A}$ of arity $k$, there exists a formula $\phi^{i}_{c}(\xt) \in \GFC$ with $\xt = (x_1, \dots, x_k)$ such that for every $\sig$-structure $\B$ of size strictly equal to $\A$ and every $\bt \in \tB$ of arity $k$ we have:
	$\B, \bt \models \phi^{i}_{c}(\xt)$ $\iff$ $\col{i}{\bt} = c$.
\end{restatable}
\vspace{-1ex}

\begin{restatable}{lemma}{formulaIsSpoilerStrategy}\label{lem:formula-is-spoiler-strategy}
	Let $\A$ and $\B$ be $\sig$-structures of strictly equal size and let $\at \in {V(\A)}^k$, $\bt\in {V(\B)}^k$ be arbitrary tuples of arity $k$.
	Let $\xt=(x_1,\ldots,x_k)$ be a tuple of $k$ distinct variables.
	If there exists a formula $\phi \in \GFC$ with $\free(\phi)\subseteq\set{x_1,\ldots,x_k}$ such that $\A, \at \models \phi(\xt) {\iff}\; \B, \bt \not\models \phi(\xt)$, then Spoiler has a $\gd(\phi)$-round winning strategy for the Guarded-Game on $(\A, \at)$, $(\B, \bt)$.
\end{restatable}
\vspace{-1ex}

\begin{restatable}{lemma}{sameColorIsDupStrategy}\label{lem:same-color-is-dup-strategy}
	Let $i \in \nat$, let $\A$ and $\B$ be $\sig$-structures of strictly equal size, and let $\at \in \tA$, $\bt \in \tB$ be tuples of arity $k$.
	If $\col{1}{\at} = \col{1}{\bt}$, then the configuration $(\A, \at)$, $(\B, \bt)$ is not distinguishing.
	Further, if $\col{i+1}{\at} = \col{i+1}{\bt}$ and $\ccount{\A}{c} = \ccount{\B}{c}$ for all $c \in \cols{i}{\A} \union \cols{i}{\B}$, then Duplicator has an $i$-round winning strategy for the Guarded-Game on $(\A, \at)$, $(\B, \bt)$.
\end{restatable}
\begin{proof}[Proof of~\cref{thm:main-logic}]
	If $\A$ and $\B$ are \emph{not} of strictly equal size, it is easy to see, that RCR distinguishes them, that a suitable sentence of the form $\existsgeq{n} (\var_1, \dots, \var_k) \qsep (R(\var_1, \dots, \var_k) \land\, \var_1{=}\var_1)$ distinguishes $\A$ and $\B$, and that Spoiler has a $1$-round winning strategy in the Guarded-Game on $\A$, $\B$.
	Thus, it suffices to consider the case where $\A$ and $\B$ are of strictly equal size.

\enquote{1$\implies$2}:
If RCR distinguishes $\A$ and $\B$, then there exists an $i\in\NN$ and a color $c \in \cols{i}{\A} \union \cols{i}{\B}$ such that $n \isdef \card{\set{ \at \in \tA \mid \col{i}{\at} = c }} \neq \card{\set{ \bt \in \tB \mid \col{i}{\bt} = c }}$.
W.l.o.g.\ we assume that $n\neq 0$.
Note that $\A \models \phi$ and $\B \not\models \phi$ for $\phi \isdef \existseq{n} \xt \qsep (R(\xt) \land \phi^{i}_{c}(\xt))$, where $R \in \atp(c)$ and $\phi^{i}_{c}(\xt)$ is chosen according to \cref{lem:describing-color-in-gfc}.

\enquote{2$\implies$3}: This is obtained from \cref{lem:formula-is-spoiler-strategy} for $k=0$ and $\at=\bt=()$.

\enquote{3$\implies$1}:
We show the contraposition.
By assumption, RCR does \emph{not} distinguish $\A$ and $\B$.
Let $j\deff\max\set{i_\A,i_\B}$.
Let us consider the Guarded-Game starting with configuration $\A$, $\B$.
In round 0, Duplicator trivially wins the  because the configuration is empty.
In round 1, for any $R \in \sig$ that Spoiler might choose, Duplicator can provide a bijection $\pi$ between $R^\A$ and $R^\B$ such that for every $\at\in R^\tA$ we have $\col{j}{\at} = \col{j}{\pi(\at)}$ (this follows since RCR does not distinguish $\A$ and $\B$).
For any $\at\in R^\A$ and $\bt\deff\pi(\at)$ that Spoiler might choose, we obtain from the first statement of \cref{lem:same-color-is-dup-strategy} that $(\A, \at)$, $(\B, \bt)$ is not distinguishing.
Hence, Duplicator can win round 1.
By our choice of $j$ we have $\col{i+1}{\at} = \col{i+1}{\bt}$ for all $i\in\NN$.
Since RCR does not distinguish $\A$ and $\B$, the assumption for the second statement of \cref{lem:same-color-is-dup-strategy} is satisfied for \emph{every} $i\in\NN$.
Hence, Duplicator has an $i$-round winning strategy on $(\A, \at), (\B, \bt)$, for every $i\in\NN$.
Consequently, Duplicator wins the Guarded-Game on $\A$, $\B$.
\end{proof}
  
	\section{Implementing RCR in \texorpdfstring{$\bigO(\size{\A} \cdot \log(\size{\A}))$}{|A| log|A|}}\label{sec:main-runtime}
Let $\sig$ be a fixed relational signature.
This section is devoted to proving our following third main result:
\mainruntime{}%
The factor hidden by the $\bigO$-notation is of size
$2^{\bigO(k\cdot \log k)}$, where $k$ is the maximum arity of a relation symbol in $\sig$.

Recall from \cref{sec:RCR} that by performing classical Color Refinement on the colored multigraph $\grep{\A}$ we can implement RCR on a $\sig$-structure $\A$ with runtime $\bigO((n+m) \cdot \log(n))$, where $n$ denotes the number of vertices and $m$ denotes the total number of edges of $\grep{\A}$.

Note that the nodes $w_\at$ for all tuples $\at\in\tA$ that share an element, form a clique in $\grep{\A}$.
This causes a blow-up of the number of edges in $\grep{\A}$.
We will alleviate this by resolving every clique by inserting a constant number of \emph{fresh} vertices that are connected to all tuples participating in a clique.
This will drastically reduce the number of edges, yielding a new colored multigraph $\vgrep{\A}$ whose number of nodes and edges is in $\bigO(\size{\A})$.
For the precise definition of $\vgrep{\A}$ we need the following notation.

\begin{definition}\label{def:slice}
	Let $\A$ be a $\sig$-structure.
	We call a tuple $\slice = (s_1, \dots, s_\ell) \in {\V(\A)}^\ell$ a \emph{slice (over $\V(\A)$)}, if its elements are pairwise distinct (i.e., $\ell = \card{\tset(\slice)}$) and $\ell\geq 1$.
	We call $\slice$ a \emph{slice of $\at\in {\V(\A)}^k$}, if $\slice$ is a slice over $V(\A)$ and $\tset(\slice) \subseteq \tset(\at)$.
	For an $\at \in \tA$, we write $\slices(\at)$ for the set of all slices of $\at$, i.e., $\slices(\at) = \set{ \slice \in {\V(\A)}^\ell \mid \tset(\slice)\subseteq\tset(\at),\ \ell\geq 1,\ \card{\tset(\slice)} = \ell }$.
	Conversely, for a slice $\slice$ over $V(\A)$ we denote by $\slices^{-1}(\slice) \isdef \set{ \at \in \tA \mid \slice \in \slices(\at) }$ the set of tuples in $\tA$ that $\slice$ is a slice of.
\end{definition}

\begin{definition}\label{def:efficient-rep}
	\let\qed\relax%
	Let $\A$ be a $\sig$-structure.
	Let $\vgrep{\A}$ be the colored multigraph of signature $\repsig$ defined as follows.
	The universe $\V(\vgrep{\A})$ consists of the nodes $w_{\at}$ for all $\at \in \tA$ and a new node $v_{\slice}$ for every slice $\slice \in \slices(\tA)\deff \bigunion_{\at\in \tA} \slices(\at)$.
  I.e., $\V(\vgrep{\A}) = \set{ w_{\at} \mid \at \in \tA } \disunion \set{ v_{\slice} \mid \slice\in \slices(\tA) }$.
	Furthermore, ${(U_R)}^{\vgrep{\A}} \isdef \set{ w_{\at} \mid \at \in R^{\A} }$, for all $R \in \sig$.
	And for all $i,j \in [\ar(\sig)]$ we let ${(E_{i,j})}^{\vgrep{\A}} \isdef$
	\begin{alignat*}{3}
		&\bigl\{\, &(w_{\at}, v_{\slice}) 
			&\mid
			\at \in \tA,\,
			\slice \in \slices(\at),\,
			&(i,j) \in \stp(\at, \slice)
		\,&\bigr\}
		\;\;\union \\
		&\bigl\{\, &(v_{\slice}, w_{\bt}) 
			&\mid
			\bt \in \tA,\,
			\slice \in \slices(\bt),\,
			&(i,j) \in \stp(\slice, \bt) 
		\,&\bigr\}\;.\tag*{$\lrcorner$}
	\end{alignat*}
\end{definition}

\begin{figure}
  \centering
\begin{tikzpicture}[
	every node/.style={draw, rectangle},
	xscale=1,
	yscale=1]
	\clip (-3.5,-2.252) rectangle (3.5, 2.252);
	\tikzmath{
		\hght=1.25;
		\wdth=1.5;
		\offset=.075;
		\small=8;
		\medium=15;
		\myangle=-35;
		\distA=38pt;
		\distB=43pt;
	}
	\node[slice] (2) at (0,0) {$v_{(2)}$};
	\node[slice, above =\hght of 2] (1) {$v_{(1)}$};
	\node[slice, below =\hght of 2] (3) {$v_{(3)}$};
	\node[colorR, very thick, right =\wdth of 2] (112) {$w_{(1,1,2)}$};
	\node[colorR, very thick, left =\wdth of 2] (232) {$w_{(2,3,2)}$};
	\node[slice, above =\hght of 112] (12) {$v_{(1,2)}$};
	\node[slice, below =\hght of 112] (21) {$v_{(2,1)}$};
	\node[slice, above =\hght of 232] (23) {$v_{(2,3)}$};
	\node[slice, below =\hght of 232] (32) {$v_{(3,2)}$};

	\draw[tpFor, bend right=\small] 
		(112)
		edge
		node[midway, edgeAnnotation, above=-2pt, sloped] {$11, 21$}
	(1);
	\draw[tpFor, bend right=\small] 
		(1)
		edge
		node[midway, edgeAnnotation, below=-2pt, sloped] {$11, 12$}
	(112);

	\draw[bend right=\small] 
		($(112.west)+(0,\offset)$)
		edge[tpFor]
		node[edgeAnnotation, above=-2pt] {$31$}
	($(2.east)+(0,\offset)$);
	\draw[bend right=\small] 
		($(2.east)-(0,\offset)$)
		edge[tpFor]
		node[edgeAnnotation, below=-2pt] {$13$}
		($(112.west)-(0,\offset)$);

	\draw[bend right=\small] 
		($(232.east)-(0,\offset)$)
		edge[tpFor]
		node[edgeAnnotation, below=-2pt] {$11, 31$}
	($(2.west)-(0,\offset)$);
	\draw[bend right=\small] 
		($(2.west)+(0,\offset)$)
		edge[tpFor]
		node[edgeAnnotation, above=-2pt] {$11, 13$}
	($(232.east)+(0,\offset)$);

	\draw[tpFor, bend right=\small] 
		(232)
		edge
		node[midway, edgeAnnotation, below=-2pt, sloped] {$21$}
	(3);
	\draw[tpFor, bend right=\small] 
		(3)
		edge
		node[midway, edgeAnnotation, above=-2pt, sloped] {$12$}
	(232);

	\draw[tpFor, bend right=\small] 
		(112)
		edge
		node[midway, edgeAnnotation, above, sloped] {$11, 21, 32$}
	(12);
	\draw[tpFor, bend right=\small] 
		(12)
		edge
		node[midway, edgeAnnotation, above, sloped] {$11, 12, 23$}
	(112);

	\draw[tpFor, bend right=\small] 
		(112)
		edge
		node[midway, edgeAnnotation, below, sloped] {$12, 22, 31$}
	(21);
	\draw[tpFor, bend right=\small] 
		(21)
		edge
		node[midway, edgeAnnotation, below, sloped] {$13, 21, 22$}
	(112);

	\draw[tpFor, bend right=\small] 
		(232)
		edge
		node[midway, edgeAnnotation, above, sloped] {$11, 22, 31$}
	(23);
	\draw[tpFor, bend right=\small] 
		(23)
		edge
		node[midway, edgeAnnotation, above, sloped] {$11, 13, 22$}
	(232);

	\draw[tpFor, bend right=\small] 
		(232)
		edge
		node[midway, edgeAnnotation, below, sloped] {$12, 21, 32$}
	(32);
	\draw[tpFor, bend right=\small] 
		(32)
		edge
		node[midway, edgeAnnotation, below, sloped] {$12, 21, 23$}
	(232);

\end{tikzpicture}   %
	\caption{$\vgrep{\A}$ for the structure $\A$ considered in \cref{example:vgrepA}.}%
	\label{fig:ex:efficient-runtime-representation}
\end{figure}

\begin{example}\label{example:vgrepA}
	Consider the signature $\sig \isdef \set{ R }$ with $\ar(R)=3$.
	Then, $\repsig$ consists of the unary relation symbol $U_R$ and binary relation symbols $E_{i,j}$ for $i,j\in \set{1,2,3}$.

	Let $\A$ be the $\sig$-structure where $\V(\A) \isdef \set{ 1,2,3 }$ and $R^{\A} \isdef \set{ (1,1,2),\; (2,3,2) }$.
	Then, $\tA = R^{\A}$ and $\slices(\tA) = \set{ (1),\; (2),\; (3),\;\allowbreak (1,2),\; (2,1),\; (2,3),\; (3,2) }$.
	The colored multigraph $\vgrep{\A}$ is the $\repsig$-structure with universe $V(\vgrep{\A}) = \set{ w_{(1,1,2)},\; w_{(2,3,2)} }\allowbreak \union \set{ v_\slice \mid \slice\in\slices(\tA) }$.
	The unary symbol $U_R$ is interpreted by the set ${(U_R)}^{\vgrep{\A}}=\set{ w_{(1,1,2)},\; w_{(2,3,2)} }$.
	See \cref{fig:ex:efficient-runtime-representation} for an illustration of $\vgrep{\A}$.
\end{example}

Note that for each $\at\in\tA$ of arity $k$, the number of slices of $\at$ is $\card{\slices(\at)} \leq k {\cdot} {k!}$.
Thus, for $k\deff \ar(\sig) = \max\setc{\ar(R)}{R\in\sig}$ we have 
	$\card{\V(\vgrep{\A})} = \card{\tA} + \card{\slices(\tA)} \leq (1 + k {\cdot} {k!}) {\cdot} \card{\tA}$
and
	$\card{{(E_{i,j})}^{\vgrep{\A}}} \leq 2 {\cdot} k {\cdot} {k!} {\cdot} \card{\tA}$, for all $i, j \in [k]$.
Hence, the total number of edges of $\vgrep{\A}$ is at most $2 {\cdot} k^3{\cdot} k! {\cdot} \card{\tA}$.
I.e., for the fixed relational signature $\sig$, the number of nodes and edges of the colored multigraph $\vgrep{\A}$ associated with a $\sig$-structure $\A$ is of size $\bigO(\size{\A})$, where the factor hidden by the $\bigO$-notation is bounded by $2 {\cdot} k^3{\cdot} k! = 2^{\bigO(k\cdot \log k)}$ for $k\deff\ar(\sig)$.

From Berkholz, Bonsma and Grohe~\cite{Berkholz2017} we know that classical Color Refinement can be implemented to run in time $\bigO((n+m) \cdot \log(n))$ on the colored multigraph $\vgrep{\A}$, where $n$ denotes the number of vertices and $m$ denotes the total number of edges.
Therefore, since $n$ and $m$ both are of size $\bigO(\size{\A})$, \cref{thm:main-runtime} is obtained as an immediate consequence of the following \cref{thm:rcr-is-cr-on-vgrep}:

\begin{theorem}\label{thm:rcr-is-cr-on-vgrep}
	Let $\A$ be a $\sig$-structure.
	Let $\col{i}{\at}$ be the color assigned to tuple $\at\in\tA$ in the $i$-th round of Relational Color Refinement RCR on $\A$, and let $\gcr{i}{u}$ be the color assigned to node $u$ of $\vgrep{\A}$ in  round $i$ of conventional Color Refinement CR on the colored multigraph $\vgrep{\A}$.
	For all $i \in \nat$ and all $\at, \bt \in \tA$ we have: \ $\col{i}{\at} = \col{i}{\bt}$ $\iff$ $\gcr{2i+1}{w_{\at}} = \gcr{2i+1}{w_{\bt}}$.
\end{theorem}

The remainder of this section is devoted to the proof of \cref{thm:rcr-is-cr-on-vgrep}.
See \cref{app:alt-rep} for proofs of all subsequent lemmas and claims.
We start with a lemma that summarizes some obvious facts.

\begin{restatable}{lemma}{sliceObservations}\label{lem:slice-observations}
	Let $\A$ be a $\sig$-structure.
	Let $k,k',\ell\geq 1$ and let $\at=(a_1,\ldots,a_k)$ and $\bt=(b_1,\ldots,b_{k'})$ be elements in $\tA$, and let $\slice=(s_1,\ldots,s_\ell)$ be a slice over $\V(\A)$.
	\begin{enumerate}[label={(\alph*)}, ref=\thelemma\,(\alph*), topsep=0pt, noitemsep]
		\crefalias{enumi}{lemma}

		\item\label{lem:slice-observations:stp-nonempty-iff-intersect}
		$\stp(\at, \bt) \neq \emptyset \iff \tset(\at) \intersect \tset(\bt) \neq \emptyset$.

		\item\label{lem:slice-observations:self-stp-equal-iff-bijection}
		$\stp(\at) = \stp(\bt)$ $\iff$ $\ar(\at) = \ar(\bt)$ and the function $\beta\colon \tset(\at) \to \tset(\bt)$ with $\beta(a_i) \isdef b_i$ for all $i \in [k]$ is well-defined and bijective.

		\item\label{lem:slice-observations:every-j-has-i}
		$\slice \in \slices(\at)$ $\iff$ for every $i \in [\ell]$ there exists a $j \in [k]$ such that $(i, j) \in \stp(\slice, \at)$ $\iff$ for every $j \in [\ell]$ there exists an $i \in [k]$ such that $(i, j) \in \stp(\at, \slice)$.

		\item\label{lem:slice-observations:slice-stp-is-unique}
		Let $\slice\in\slices(\at)$.
		For all $\slice' \in \slices(\at)$ we have:\;
		$\stp(\at, \slice) = \stp(\at, \slice')$ $\iff$ $\slice = \slice'$.
	\end{enumerate}
\end{restatable}

Two nodes $u,u'\in V(\vgrep{\A})$ are called \emph{neighbors in $\vgrep{\A}$} if $(u,u')\in {(E_{i,j})}^{\vgrep{\A}}$ for some $i,j\in [\ar(\sig)]$.

\begin{remark}\label{remark:neighbors-in-vgrep}
Note that for all $\at\in\tA$ and $\slice\in \slices(\tA)$ we have: $w_\at$ and $v_{\slice}$ are neighbors in $\vgrep{\A}$ $\iff$ $\slice\in\slices(\at)$.
By \cref{lem:slice-observations:slice-stp-is-unique} we obtain: For all $\at\in\tA$ and all  $\slice,\slice'\in\slices(\at)$ with $\slice\neq\slice'$, we have 
$\stp(\at, \slice)\neq\stp(\at,\slice')$.
\end{remark}

The following characterization of tuples $\at,\bt$ with $\stp(\at)=\stp(\bt)$ will be crucial for our
proof of \cref{thm:rcr-is-cr-on-vgrep}.

\begin{restatable}{lemma}{selfStpIdentifiesSlices}\label{lem:self-stp-identifies-stp-of-slices}
	Let $\A$ be a $\sig$-structure.
	For all $\at, \bt \in \tA$ we have:\;
	$\stp(\at) = \stp(\bt)$ $\iff$ there is a bijection $\pi_{\slices}\colon \slices(\at) \to \slices(\bt)$ such that for all $\slice \in \slices(\at)$ we have $\stp(\at, \slice) = \stp(\bt, \pi_{\slices}(\slice))$.
\end{restatable}
It follows from \cref{remark:neighbors-in-vgrep} that the bijection $\pi_{\slices}$ is unique, if it exists.
The following lemma summarizes straightforward properties of the mapping $\pi_\slices$ obtained from \cref{lem:self-stp-identifies-stp-of-slices}.

\begin{restatable}{lemma}{easyPropertiesOfPiSlices}\label{lemma:easy-properties-of-pi-slices}
	Let $\A$ be a $\sig$-structure, let $\at,\bt\in\tA$ with $\stp(\at)=\stp(\bt)$, and let $\pi_\slices\colon \slices(\at) \to \slices(\bt)$ be the bijection obtained from \cref{lem:self-stp-identifies-stp-of-slices}.
	For all $\slice, \slice' \in \slices(\at)$ and for $\tup{t} \isdef \pi_\slices(\slice)$ and $\tup{t}' \isdef \pi_\slices(\slice')$ we have:
	\begin{statements}[topsep=0pt, noitemsep]
		\item
		$\ar(\slice)=\ar(\tup{t})$, \ and
		\item 
		$\tset(\slice)\subseteq \tset(\slice')$ $\iff$
		$\tset(\tup{t})\subseteq\tset(\tup{t}')$, \ and
		\item
		$\tset(\slice)=\tset(\slice')$ $\iff$ $\tset(\tup{t})=\tset(\tup{t}')$.
	\end{statements}
\end{restatable}  

For $\at \in \tA$ let $N(\at) \isdef \set{ \ct'\in\tA \mid \stp(\at, \ct') \neq \emptyset } = \set{ \ct'\in\tA \mid \tset(\at)\cap\tset(\ct')\neq \emptyset }$
(the last equality is due to \cref{lem:slice-observations:stp-nonempty-iff-intersect}).
We proceed with the main technical lemma that will enable us to prove \cref{thm:rcr-is-cr-on-vgrep}.

\begin{restatable}{lemma}{vgrepEquivRCR}\label{lem:vgrep-equiv-rcr}
	Let $\A$ be a $\sig$-structure.
	Let $\myColors$ be a non-empty set and let $\myc$ be a mapping $\myc\colon \tA \to \myColors$.
	\\
	Consider $\at, \bt \in \tA$ with $\myc(\at)=\myc(\bt)$ and $\stp(\at) = \stp(\bt)$.
	Let $\pi_{\slices}\colon \slices(\at) \to \slices(\bt)$ be the bijection obtained by \cref{lem:self-stp-identifies-stp-of-slices}.
	The following are equivalent:
	\begin{statements}[topsep=0pt, noitemsep]
		\item\label{statement1-equiv} 
		$\mset{ (\stp(\at, \ct), \myc(\ct)) \mid \ct \in N(\at) } \;=\; \mset{ (\stp(\bt, \ct), \myc(\ct)) \mid \ct \in N(\bt) }$.
		
		\item\label{statement2-equiv}
		For all $\slice \in \slices(\at)$ we have:\\
		$\mset{ (\stp(\slice, \ct), \myc({\ct})) \mid \ct \in \slices^{-1}(\slice)}
		\;=\;
		\mset{ (\stp(\pi_{\slices}(\slice), \dt), \myc({\dt})) \mid \dt \in \slices^{-1}(\pi_\slices(\slice)) }$.
	\end{statements}%
\end{restatable}
The combinatorially quite involved proof of \cref{lem:vgrep-equiv-rcr} can be found in \cref{appendix:lem:vgrep-equiv-rcr}.

The following facts will be helpful for the proofs.
\begin{fact}\label{fact:easy-rcr-cr}
	For all $\at,\bt\in\tA$ we have
	\begin{enumerate}[label={(\alph*)}, ref={\thefact~{(\alph*)}}, topsep=0pt, noitemsep]
		\crefalias{enumi}{fact}%
		\item\label{fact:easy-rcr-cr:atp-is-initial-color}
		$\atp(\at)=\atp(\bt)$ $\iff$ $\gcr{0}{w_{\at}} = \gcr{0}{w_{\bt}}$.
		\item\label{fact:easy-rcr-cr:lambda-tuple-stp}
		For all $\slice \in \slices(\at)$ and $\tup{t} \in \slices(\bt)$ we have:\;
		$\stp(\at,\slice) = \stp(\bt,\tup{t})$ $\iff$ $\lambda(w_\at,v_\slice)=\lambda(w_\bt,v_{\tup{t}})$
		\item\label{fact:easy-rcr-cr:lambda-slice-stp}
		For all $\slice \in \slices(\at)$, $\tup{t} \in \slices(\bt)$ and for all $\ct \in \slices^{-1}(\slice)$, $\tup{d} \in \slices^{-1}(\tup{t})$ we have:\\
		$\stp(\slice, \ct) = \stp(\tup{t}, \tup{d})$ $\iff$ $\lambda(v_\slice, w_\ct)=\lambda(v_{\tup{t}}, w_{\tup{d}})$.

		\item\label{fact:easy-rcr-cr:neighbors-of-w}
		For all nodes $v$ of $\vgrep{\A}$ we have: \ $v$ and $w_\at$ are neighbors in $\vgrep{\A}$ $\iff$ $v=v_\slice$ for some $\slice\in\slices(\at)$.
	\end{enumerate}
\end{fact}
The first three statements follow immediately from the definition of $\vgrep{\A}$, see \cref{def:efficient-rep}.
The last two statements follow from \cref{remark:neighbors-in-vgrep} and the definition of $\vgrep{\A}$.
The following lemma relates the color that CR assigns to the node $w_\at$ to the colors it assigns to the nodes $v_\slice$ for the slices $\slice$ of $\at$.
\begin{restatable}{lemma}{slicesDetermineCR}\label{lem:slices-determine-cr}
	For all $\at, \bt \in \tA$ with $\stp(\at) = \stp(\bt)$ and all $i \in \natpos$ we have:\; $\gcr{i}{w_{\at}} = \gcr{i}{w_{\bt}} \iff$
	\begin{enumerate}[label={(\alph*)}, topsep=0pt, noitemsep]
		\item $\gcr{i-1}{w_{\at}} = \gcr{i-1}{w_{\bt}}$\; and
		\item for all $\slice \in \slices(\at)$ and $\tup{t} \isdef \pi_{\slices}(\slice)$:\; $\gcr{i-1}{v_{\slice}} = \gcr{i-1}{v_{\tup{t}}}$.
	\end{enumerate}
	Here, $\pi_{\slices}\colon \slices(\at) \to \slices(\bt)$ is the bijection from \cref{lem:self-stp-identifies-stp-of-slices}.
\end{restatable}

Finally, we are ready for the proof of \cref{thm:rcr-is-cr-on-vgrep}.
\begin{proof}[Proof of \cref{thm:rcr-is-cr-on-vgrep}]
	We proceed by induction over $i$.
	\smallskip

\myparagraph{Base case:} Consider $i = 0$ and $\at, \bt \in \tA$.

\noindent
By the definition of RCR, $\col{0}{\at} = \col{0}{\bt}$ is true iff
\begin{conditions}[label={(\arabic*)}, topsep=0pt, noitemsep]
	\item\label{cond:proof:rcr-is-cr-on-vgrep:base:1}
	$\atp(\at) = \atp(\bt)$, and
	\item\label{cond:proof:rcr-is-cr-on-vgrep:base:2}
	$\stp(\at) = \stp(\bt)$.
\end{conditions}
From \cref{fact:easy-rcr-cr:atp-is-initial-color} we know that \cref*{cond:proof:rcr-is-cr-on-vgrep:base:1} holds $\iff \gcr{0}{w_{\at}} = \gcr{0}{w_{\bt}}$.
\smallskip

\noindent
According to \cref{lem:self-stp-identifies-stp-of-slices}, \cref{cond:proof:rcr-is-cr-on-vgrep:base:2} holds $\iff$ there is a bijection $\pi_{\slices}\colon \slices(\at) \to \slices(\bt)$ such that $\stp(\at, \slice) = \stp(\bt, \pi_{\slices}(\slice))$ for all $\slice \in \slices(\at)$.
By using \cref{fact:easy-rcr-cr:neighbors-of-w}, we obtain that such a bijection exists $\iff$ there is a bijection $\pi'$ from the set $N_\at$ of all neighbors $v$ of $w_{\at}$ in $\vgrep{\A}$ to the set $N_\bt$ of all neighbors $v$ of $w_{\bt}$ in $\vgrep{\A}$ such that $\lambda(w_\at,v) = \lambda(w_\bt,\pi'(v))$ holds for all $v\in N_\at$.
\smallskip

\noindent
On the other hand, by the definition of CR, $\gcr{1}{w_{\at}} = \gcr{1}{w_{\bt}}$ is true iff
\begin{conditions}[label={(\arabic*)}, topsep=0pt, noitemsep]
	\item $\gcr{0}{w_{\at}} = \gcr{0}{w_{\bt}}$, and
	\item there is a bijection $\pi''\colon N_\at \to N_\bt$ such that $\lambda(w_\at,v) = \lambda(w_\bt, \pi''(v))$ and $\gcr{0}{v} = \gcr{0}{\pi''(v)}$.
\end{conditions}

Note that condition (1') is equivalent to condition (1).
Furthermore, by the definition of $\vgrep{\A}$, all nodes $v \in N_\at \union N_\bt$ have the same color $\gcr{0}{v} = \emptyset$.
Thus, we obtain that condition (2') is equivalent to condition (2).
In summary, we, in particular, obtain that $\col{0}{\at} = \col{0}{\bt} \iff \gcr{1}{w_{\at}} = \gcr{1}{w_{\bt}}$.
This completes the proof for the induction base $i=0$.
 	
	\myparagraph{Inductive Step:} Consider an $i \in \natpos$, and let $\at, \bt \in \tA$.

	\noindent\emph{Induction hypothesis}:\;
		For all  $j < i$ and all $\ct, \tup{d} \in \tA$ we have:\
		$\col{j}{\ct} = \col{j}{\tup{d}} \iff \gcr{2j+1}{w_{\ct}} = \gcr{2j+1}{w_{\tup{d}}}$.\smallskip

	\noindent\emph{Induction Claim}:\;
		$\col{i}{\at} \,{=}\mskip1mu \col{i}{\bt} \,{\iff}\mskip1mu \gcr{2i+1}{w_{\at}} \,{=}\, \gcr{2i+1}{w_{\bt}}$.\smallskip

		In case that $\stp(\at) \neq \stp(\bt)$, we have $\col{0}{\at} \neq \col{0}{\bt}$, and hence, by the definition of RCR, $\col{j}{\at} \neq \col{j}{\bt}$ holds for all $j \geq 0$.
	 	Furthermore, by the induction base, $\col{0}{\at} \neq \col{0}{\bt}$ implies that $\gcr{1}{w_{\at}} \neq \gcr{1}{w_\bt}$. Hence, by the definition of CR,  $\gcr{j}{w_{\at}} \neq \gcr{j}{w_\bt}$ holds for all $j \geq 1$.
		In particular, this yields: $\col{i}{\at} \neq \col{i}{\bt}$ and $\gcr{2i+1}{w_{\at}} \neq \gcr{2i+1}{w_\bt}$, completing the induction step.

		In the following, we consider the case where $\stp(\at)=\stp(\bt)$.
		From \cref{lem:self-stp-identifies-stp-of-slices} we obtain a bijection $\pi_\slices\colon\slices(\at)\to\slices(\bt)$ satisfying
		$\stp(\at,\slice)=\stp(\bt,\pi_\slices(\slice))$ for all $\slice\in\slices(\at)$.

		If $\col{i-1}{\at} \neq \col{i-1}{\bt}$, then $\col{i}{\at} \neq \col{i}{\bt}$ by definition of RCR and $\gcr{2(i-1)+1}{w_{\at}} \neq \gcr{2(i-1)+1}{w_{\bt}}$ by induction hypothesis. It follows from the definition of CR that $\gcr{2i+1}{w_{\at}} \neq \gcr{2i+1}{w_{\bt}}$ as well. Thus, from now on we consider the case that $\col{i-1}{\at} = \col{i-1}{\bt}$.

		Using \cref{lem:slices-determine-cr} we get that $\gcr{2i+1}{w_{\at}} = \gcr{2i+1}{w_{\bt}}$ if, and only if, $\gcr{2i}{w_{\at}} = \gcr{2i}{w_{\bt}}$ and for all $\slice \in \slices(\at)$ and $\tup{t} \isdef \pi_{\slices}(\slice)$ it holds that $\gcr{2i}{v_{\slice}} = \gcr{2i}{v_{\tup{t}}}$. Applying the same lemma again yields that $\gcr{2i+1}{w_{\at}} = \gcr{2i+1}{w_{\bt}}$ if, and only if $\gcr{2i-1}{w_{\at}} = \gcr{2i-1}{w_{\bt}}$ and for all $\slice \in \slices(\at)$ and $\tup{t} \isdef \pi_{\slices}(\slice)$ it holds that
		$\gcr{2i}{v_{\slice}} = \gcr{2i}{v_{\tup{t}}}$ 
		and
		$\gcr{2i-1}{v_{\slice}} = \gcr{2i-1}{v_{\tup{t}}}$.
		\smallskip
		
		Thus, we must show that $\col{i}{\at} = \col{i}{\bt}$ if, and only if, $\gcr{2i-1}{w_{\at}} = \gcr{2i-1}{w_{\bt}}$ and for all $\slice \in \slices(\at)$ and $\tup{t} \isdef \pi_{\slices}(\slice)$ it holds that
		$\gcr{2i}{v_{\slice}} = \gcr{2i}{v_{\tup{t}}}$ 
		and
		$\gcr{2i-1}{v_{\slice}} = \gcr{2i-1}{v_{\tup{t}}}$.
		Recall that we have $\col{i-1}{\at} = \col{i-1}{\bt}$. Since $2i-1 = 2(i-1)+1$, we get that $\gcr{2i-1}{w_{\at}} = \gcr{2i-1}{w_{\bt}}$ from the induction hypothesis. Hence, it remains to show that $\col{i}{\at} = \col{i}{\bt}$ if, and only if, for all $\slice \in \slices(\at)$ and $\tup{t} \isdef \pi_{\slices}(\slice)$ it holds that
		$\gcr{2i}{v_{\slice}} = \gcr{2i}{v_{\tup{t}}}$ 
		and
		$\gcr{2i-1}{v_{\slice}} = \gcr{2i-1}{v_{\tup{t}}}$. It is now easy to see that the following two claims finish the proof.
		\begin{restatable}{claim}{intermediateSliceCR}\label{claim:intermediate-slice-cr-implied}
			If $\col{i}{\at} = \col{i}{\bt}$, then the following holds for all $\slice \in \slices(\at)$ and $\tup{t} \isdef \pi_{\slices}(\slice)$:
			\begin{quote}
				If\; $\gcr{2i-1}{v_{\slice}} = \gcr{2i-1}{v_{\tup{t}}}$, then\; $\gcr{2i}{v_{\slice}} = \gcr{2i}{v_{\tup{t}}}$.
			\end{quote}
		\end{restatable}
		\begin{restatable}{claim}{RCRiffSliceCR}\label{claim:rcr-iff-slice-cr}
			$\col{i}{\at} = \col{i}{\bt}$ \;$\iff$\; for all $\slice \in \slices(\at)$ and $\tup{t} \isdef \pi_{\slices}(\slice)$ we have $\gcr{2i-1}{v_{\slice}} = \gcr{2i-1}{v_{\tup{t}}}$.
		\end{restatable}
		\noindent See \cref{appendix:main-runtime-thm-claims} for a proof of \cref{claim:intermediate-slice-cr-implied,claim:rcr-iff-slice-cr}.
\end{proof}   
	\section{Final Remarks}\label{sec:conclusion}
We introduced Relational Color Refinement (RCR) as an adaptation of the classical Color Refinement (CR) procedure for arbitrary relational structures, and we showed that its distinguishing power admits an analogous combinatorial (\cref{thm:main-homs}) and logical (\cref{thm:main-logic}) characterization as classical CR\@.
Combining the Theorems~\ref{thm:main-homs} and~\ref{thm:main-logic} yields the following theorem.
\begin{alphtheorem}\label{thm:main}
	Let $\sig$ be a finite relational signature and let $\A$ and $\B$ be $\sig$-structures.\\
	The following statements are equivalent.
	\begin{statements}[topsep=0pt, noitemsep]
		\item Relational Color Refinement does not distinguish $\A$ and $\B$.
		\item $\A$ and $\B$ are homomorphism indistinguishable over the class of acyclic $\sig$-structures.
		\item $\A$ and $\B$ are indistinguishable by sentences of the logic  $\GFC$.
		\item Duplicator has a winning strategy for the Guarded-Game on $\A, \B$.
	\end{statements}
\end{alphtheorem}
Alpha-acyclicity, which is equivalent to our notion of acyclicity~\cite{BraultBaron2016}, is a natural and widely used choice for the notion of acyclicity for relational structures and
hypergraphs~\cite{Abiteboul1995,Bagan2007,Beeri1983,Berkholz2020,DBLP:journals/siamcomp/BernsteinG81,Idris2017,Yannakakis1981}.

Further, in \cref{sec:RCR} we showed that RCR can be implemented to run in time $\bigO((\size{\A} + \density(\A)) \cdot \log(\size{\A}))$ on any relational structure $\A$, where $\density(\A)$ denotes the \emph{cohesion} of $\A$, which is bounded by $\size{\A}^{2}$.
In \cref{sec:main-runtime} we improved this to
$\bigO(\size{\A}\cdot\log\size{\A})$ (\cref{thm:main-runtime}).

\smallskip
\noindent There are multiple directions for further research:
\begin{enumerate}
	\item
	Considering that Color Refinement is equivalent to the 1-dimensional Weisfeiler-Leman (WL) algorithm, RCR might be a good basis to devise a generalization of the $k$-dimensional WL to arbitrary relational structures.
	
	\item
	Given the close relationship between relational structures and hypergraphs, a coloring method similar to RCR should exist for hypergraphs too.
	However, notice that RCR relies heavily on the order that the tuples provide --- and this order is absent in hyperedges.
	Thus, it is not clear, how to adapt the refinement of the colors in an iteration step from tuples to hyperedges.

	\item
	Another promising direction is to think about the applications of RCR, given that there are so many applications of classical CR, also apart from isomorphism testing~(cf.\ \cref{sec:introduction}).
	In particular, we conjecture that RCR can be used to lift the
	result by Riveros, Scheidt and Schweikardt~\cite{Riveros2024}
	from binary structures to arbitrary structures, i.e., that RCR
	can be used to speed up the evaluation of free-connex acyclic conjunctive queries over arbitrary relational structures.
	Further, the tight connection between Color Refinement and Graph Neural Networks suggests interesting applications for RCR in machine learning as well.
	We plan to investigate this.
\end{enumerate}

\paragraph{Acknowledgement.} 
We thank Panagiotis Aivasiliotis for pointing out an issue in~\cref{sec:main-homomorphism} of a previous version of this paper; this issue has now been resolved. 
	\bibliography{references}

	\clearpage
	\appendix

\section*{APPENDIX}

\section{Details omitted in Section~\ref{sec:RCR}}

\subsection{A Complete Run of Relational Color Refinement}\label{app:color-refinement-run}
We show that Relational Color Refinement distinguishes $\SpExa$, $\SpExb$.
\paragraph*{Colors computed on $\SpExa$.}
We first compute the coloring for $\SpExa$.
Initially, we get the following colors:
\begin{align*}\arraycolsep=4pt\def\arraystretch{1.4}
	\begin{array}{c|ccccccc}
		\tup{t} & (1,2,3,u,v,w) & (1,2) & (2,3) & (3,1) & (u,v) & (v,w) & (w,u) \\\hline
		\col{0}{\tup{t}}		&
		\bigl(\set{ R }, \hat{n}_0\bigr)	&
		\bigl(\set{ E }, n_0\bigr)	&
		\bigl(\set{ E }, n_0\bigr)	&
		\bigl(\set{ E }, n_0\bigr)	&
		\bigl(\set{ E }, n_0\bigr)	&
		\bigl(\set{ E }, n_0\bigr)	&
		\bigl(\set{ E }, n_0\bigr)
	\end{array}
\end{align*}
where $n_0 = \set{ (1,1), (2,2) }$ and $\hat{n}_0 = \set{ (1,1), (2,2), (3,3), (4,4), (5,5), (6,6) }$.
Let $c_0 \isdef (\set{ E }, n_0)$ and $\hat{c}_0 \isdef (\set{ R }, \hat{n}_0)$.
The colors after one iteration look like this:
\begin{align*}\arraycolsep=4pt\def\arraystretch{1.4}
	\begin{array}[pos]{c|ccccccc}
		\tup{t} & (1,2,3,u,v,w) & (1,2) & (2,3) & (3,1) & (u,v) & (v,w) & (w,u) \\\hline
		\col{1}{\tup{t}}		&
		(\hat{c}_0, \widehat{N}_1)	&
		(c_0, N^1_1)	&
		(c_0, N^2_1)	&
		(c_0, N^3_1)	&
		(c_0, N^4_1)	&
		(c_0, N^5_1)	&
		(c_0, N^6_1)
	\end{array}
\end{align*}
where
\begin{align*}
	\widehat{N}_1 \;&\isdef\; \bigl\{\!\!\bigl\{ %
		\begin{aligned}[t]
			&\bigl( \set{ (1,1), (2,2), (3,3), (4,4), (5,5), (6,6) }, \hat{c}_0 \bigr),\\ %
			&\bigl( \set{ (1,1), (2,2) }, c_0 \bigr), %
			\bigl( \set{ (2,1), (3,2) }, c_0 \bigr), %
			\bigl( \set{ (1,2), (3,1) }, c_0 \bigr),\\ %
			&\bigl( \set{ (4,1), (5,2) }, c_0 \bigr), %
			\bigl( \set{ (5,1), (6,2) }, c_0 \bigr), %
			\bigl( \set{ (4,2), (6,1) }, c_0 \bigr) %
			\bigr\}\!\!\bigr\}, 
		\end{aligned}\\
	N^1_1 \;&\isdef\; \mset[\big]{ %
		\bigl( \set{ (1,1), (2,2) }, \hat{c}_0 \bigr), %
		\bigl( \set{ (1,1), (2,2) }, c_0 \bigr), %
		\bigl( \set{ (2,1) }, c_0 \bigr), %
		\bigl( \set{ (1,2) }, c_0 \bigr) %
	},\\
	N^2_1 \;&\isdef\; \mset[\big]{ %
		\bigl( \set{ (1,2), (2,3) }, \hat{c}_0 \bigr), %
		\bigl( \set{ (1,2) }, c_0 \bigr), %
		\bigl( \set{ (1,1), (2,2) }, c_0 \bigr), %
		\bigl( \set{ (2,1) }, c_0 \bigr) %
	},\\
	N^3_1 \;&\isdef\; \mset[\big]{ %
		 \bigl( \set{ (1,3), (2,1) }, \hat{c}_0 \bigr), %
		 \bigl( \set{ (2,1) }, c_0 \bigr), %
		 \bigl( \set{ (1,2) }, c_0 \bigr), %
		 \bigl( \set{ (1,1), (2,2) }, c_0 \bigr) %
	},\\
	N^4_1 \;&\isdef\; \mset[\big]{ %
		 \bigl( \set{ (1,4), (2,5) }, \hat{c}_0 \bigr), %
		 \bigl( \set{ (1,1), (2,2) }, c_0 \bigr), %
		 \bigl( \set{ (2,1) }, c_0 \bigr), %
		 \bigl( \set{ (1,2) }, c_0 \bigr) %
	},\\
	N^5_1 \;&\isdef\; \mset[\big]{ %
		 \bigl( \set{ (1,5), (2,6) }, \hat{c}_0 \bigr), %
		 \bigl( \set{ (1,2) }, c_0 \bigr), %
		 \bigl( \set{ (1,1), (2,2) }, c_0 \bigr), %
		 \bigl( \set{ (2,1) }, c_0 \bigr) %
	},\\
	N^6_1 \;&\isdef\; \mset[\big]{ %
		 \bigl( \set{ (1,6), (2,4) }, \hat{c}_0 \bigr), %
		 \bigl( \set{ (2,1) }, c_0 \bigr), %
		 \bigl( \set{ (1,2) }, c_0 \bigr), %
		 \bigl( \set{ (1,1), (2,2) }, c_0 \bigr) %
	}.
\end{align*}
Since every tuple has its own color after this round, the refinement is already stable.

\paragraph*{Colors computed on $\SpExb$.}
Initially, we get the following colors:
\begin{align*}\arraycolsep=4pt\def\arraystretch{1.4}
	\begin{array}{c|ccccccc}
		\tup{t} & (1,2,3,u,v,w) & (1,2) & (2,w) & (w,u) & (u,v) & (v,3) & (3,1) \\\hline
		\col{0}{\tup{t}}		&
		\bigl(\set{ R }, \hat{n}_0\bigr)	&
		\bigl(\set{ E }, n_0\bigr)	&
		\bigl(\set{ E }, n_0\bigr)	&
		\bigl(\set{ E }, n_0\bigr)	&
		\bigl(\set{ E }, n_0\bigr)	&
		\bigl(\set{ E }, n_0\bigr)	&
		\bigl(\set{ E }, n_0\bigr)
	\end{array}
\end{align*}
where $n_0 = \set{ (1,1), (2,2) }$ and $\hat{n}_0 = \set{ (1,1), (2,2), (3,3), (4,4), (5,5), (6,6) }$.
Let $c_0 \isdef (\set{ E }, n_0)$ and $\hat{c}_0 \isdef (\set{ R }, \hat{n}_0)$.
The colors after one iteration look like this:
\begin{align*}\arraycolsep=4pt\def\arraystretch{1.4}
	\begin{array}[pos]{c|ccccccc}
		\tup{t} & (1,2,3,u,v,w) & (1,2) & (2,w) & (w,u) & (u,v) & (v,3) & (3,1) \\\hline
		\col{1}{\tup{t}}		&
		(\hat{c}_0, \widehat{M}_1)	&
		(c_0, M^1_1)	&
		(c_0, M^2_1)	&
		(c_0, M^3_1)	&
		(c_0, M^4_1)	&
		(c_0, M^5_1)	&
		(c_0, M^6_1)
	\end{array}
\end{align*}
where
\begin{align*}
	\widehat{M}_1 \;&\isdef\; \bigl\{\!\!\bigl\{ %
		\begin{aligned}[t]
			&\bigl( \set{ (1,1), (2,2), (3,3), (4,4), (5,5), (6,6) }, \hat{c}_0 \bigr),\\ %
			&\bigl( \set{ (1,1), (2,2) }, c_0 \bigr), %
			\bigl( \set{ (2,1), (6,2) }, c_0 \bigr), %
			\bigl( \set{ (4,2), (6,1) }, c_0 \bigr),\\ %
			&\bigl( \set{ (4,1), (5,2) }, c_0 \bigr), %
			\bigl( \set{ (3,2), (5,1) }, c_0 \bigr), %
			\bigl( \set{ (1,2), (3,1) }, c_0 \bigr) %
			\bigr\}\!\!\bigr\}, 
		\end{aligned}\\
	M^1_1 \;&\isdef\; \mset[\big]{ %
		\bigl( \set{ (1,1), (2,2) }, \hat{c}_0 \bigr), %
		\bigl( \set{ (1,2) }, c_0 \bigr), %
		\bigl( \set{ (1,1), (2,2) }, c_0 \bigr), %
		\bigl( \set{ (2,1) }, c_0 \bigr) %
	},\\
	M^2_1 \;&\isdef\; \mset[\big]{ %
		\bigl( \set{ (1,2), (2,6) }, \hat{c}_0 \bigr), %
		\bigl( \set{ (1,2) }, c_0 \bigr), %
		\bigl( \set{ (1,1), (2,2) }, c_0 \bigr), %
		\bigl( \set{ (2,1) }, c_0 \bigr) %
	},\\
	M^3_1 \;&\isdef\; \mset[\big]{ %
		\bigl( \set{ (1,6), (2,4) }, \hat{c}_0 \bigr), %
		\bigl( \set{ (1,2) }, c_0 \bigr), %
		\bigl( \set{ (1,1), (2,2) }, c_0 \bigr), %
		\bigl( \set{ (2,1) }, c_0 \bigr) %
	},\\
	M^4_1 \;&\isdef\; \mset[\big]{ %
		 \bigl( \set{ (1,4), (2,5) }, \hat{c}_0 \bigr), %
		 \bigl( \set{ (1,2) }, c_0 \bigr), %
		 \bigl( \set{ (1,1), (2,2) }, c_0 \bigr), %
		 \bigl( \set{ (2,1) }, c_0 \bigr) %
	},\\
	M^5_1 \;&\isdef\; \mset[\big]{ %
		 \bigl( \set{ (1,5), (2,3) }, \hat{c}_0 \bigr), %
		 \bigl( \set{ (1,2) }, c_0 \bigr), %
		 \bigl( \set{ (1,1), (2,2) }, c_0 \bigr), %
		 \bigl( \set{ (2,1) }, c_0 \bigr) %
	},\\
	M^6_1 \;&\isdef\; \mset[\big]{ %
		 \bigl( \set{ (1,3), (2,1) }, \hat{c}_0 \bigr), %
		 \bigl( \set{ (1,2) }, c_0 \bigr), %
		 \bigl( \set{ (1,1), (2,2) }, c_0 \bigr), %
		 \bigl( \set{ (2,1) }, c_0 \bigr) %
	}.
\end{align*}
Since every tuple has its own color after this round, the refinement is already stable.

Since for example the color $(\hat{c}_0, \widehat{N}_1) \in \cols{1}{\SpExa}$ is not present in $\cols{1}{\SpExb}$, Relational Color Refinement distinguishes $\SpExa$, $\SpExb$. 
\subsection{More Sophisticated Examples}\label{app:advanced-color-refinement}
\newcommand*{\enGF}[1]{\widetilde{G}_{#1}}
\newcommand*{\enI}[1]{\tilde{\str{I}}_{#1}}
\begin{figure}
	\begin{subfigure}{0.495\textwidth}
		\centering
\begin{tikzpicture}[
	edge from parent/.style={draw=none},
	every node/.style={mynode},
	grow=right]
	\node (1) at (0,0) {$1_{\text{\tiny \textcolor{myblue}{(1)}}}$}
	child {
		node (2) {$2_{\text{\tiny \textcolor{myblue}{(2)}}}$}
		child {
			node (w) {$w_{\text{\tiny \textcolor{myblue}{(3)}}}$}
			child[missing]
			child {
				node (u) {$u_{\text{\tiny \textcolor{myblue}{(1)}}}$}
			}
		}
	}
	child {
		node (3) {$3_{\text{\tiny \textcolor{myblue}{(3)}}}$}
		child {
			node (v) {$v_{\text{\tiny \textcolor{myblue}{(2)}}}$}
		}
	};
	\draw[RelE] (1) edge (2);
	\draw[RelE] (3) edge (1);
	
	\draw[RelE] (u) edge (v);
	\draw[RelE] (w) edge (u);
	
	\draw[RelE] (2) edge (3);
	\draw[RelE] (v) edge (w);
	\draw[RelF] (2) edge (w);
	\draw[RelF] (v) edge (3);

	\draw[RelR] \convexpath{1,3,2}{15pt};
	\draw[RelR] \convexpath{u,w,v}{15pt};

	\draw[RelR] \convexpath{w,2,u}{11pt};
	\draw[RelR] \convexpath{3,v,1}{11pt};
\end{tikzpicture} 		\caption{Representation of $\AdExa$.}%
		\label{fig:ex:advanced-hypergraph-a}
	\end{subfigure}
	\hfill
	\begin{subfigure}{0.495\textwidth}
		\centering
\begin{tikzpicture}[
	edge from parent/.style={draw=none},
	every node/.style={mynode},
	grow=right,
	]
	\node (1) at (0,0) {$1_{\text{\tiny \textcolor{myblue}{(1)}}}$}
	child {
		node (2) {$2_{\text{\tiny \textcolor{myblue}{(2)}}}$}
		child {
			node (w) {$w_{\text{\tiny \textcolor{myblue}{(3)}}}$}
			child[missing]
			child {
				node (u) {$u_{\text{\tiny \textcolor{myblue}{(1)}}}$}
			}
		}
	}
	child {
		node (3) {$3_{\text{\tiny \textcolor{myblue}{(3)}}}$}
		child {
			node (v) {$v_{\text{\tiny \textcolor{myblue}{(2)}}}$}
		}
	};
	\draw[RelE] (1) edge (2);
	\draw[RelE] (3) edge (1);
	
	\draw[RelE] (u) edge (v);
	\draw[RelE] (w) edge (u);

	\draw[RelF] (2) edge (3);
	\draw[RelF] (v) edge (w);
	\draw[RelE] (2) edge (w);
	\draw[RelE] (v) edge (3);

	\draw[RelR] \convexpath{1,3,2}{15pt};
	\draw[RelR] \convexpath{u,w,v}{15pt};

	\draw[RelR] \convexpath{w,2,u}{11pt};
	\draw[RelR] \convexpath{3,v,1}{11pt};
\end{tikzpicture} 		\caption{Representation of $\AdExb$.}%
		\label{fig:ex:advanced-hypergraph-b}
	\end{subfigure}
	\medskip\vspace{8ex}

	\begin{subfigure}{0.495\textwidth}
		\centering
\begin{tikzpicture}[
	edge from parent/.style={draw=none},
	every node/.style={mynode},
	level/.style={level distance=2.5cm, sibling distance=2.5cm},
	scale=1,
	grow=right,
	]
	\clip (-0.25,-2) rectangle (7.75,2);
	\node (1) at (0,0) {$1$}
	child {
		node (2) {$2$}
		child {
			node (w) {$w$}
			child[missing]
			child {
				node (u) {$u$}
			}
		}
	}
	child {
		node (3) {3}
		child {
			node (v) {$v$}
		}
	};
	\draw (1) edge 
		node[pos=0.3, sloped, above] {{\tiny \RelSymE(1,2), \RelSymR(1,2)}}
		node[near end, sloped, below] {{\tiny \RelSymE(2,1), \RelSymR(2,1)}}
	(2);
	\draw (1) edge 
		node[near start, sloped, above] {{\tiny \RelSymE(2,1), \RelSymR(1,3)}}
		node[near end, sloped, below] {{\tiny \RelSymE(1,2), \RelSymR(3,1)}}
	(3);
	\draw (1) edge[out=90, in=135, distance=40pt] 
		node[pos=0.15, sloped, above] {{\tiny \RelSymR(1,2)}}
		node[pos=0.85, sloped, below] {{\tiny \RelSymR(2,1)}}
	(v);

	\draw (2) edge 
		node[pos=0.3, sloped, above] {{\tiny \RelSymE(1,2), \RelSymR(2,3)}}
		node[pos=0.7, sloped, below] {{\tiny \RelSymE(2,1), \RelSymR(3,2)}}
	(3);
	\draw (2) edge[out=-45, in=-90, distance=40pt]
		node[pos=0.15, sloped, above] {{\tiny \RelSymR(2,1)}}
		node[pos=0.85, sloped, below] {{\tiny \RelSymR(1,2)}}
	(u);
	\draw (2) edge
		node[near start, sloped, above] {{\tiny \RelSymF(1,2), \RelSymR(2,3)}}
		node[near end, sloped, below] {{\tiny \RelSymF(2,1), \RelSymR(3,2)}}
	(w);

	\draw (3) edge
		node[near start, sloped, above] {{\tiny \RelSymF(2,1), \RelSymR(3,2)}}
		node[near end, sloped, below] {{\tiny \RelSymF(1,2), \RelSymR(2,3)}}
	(v);

	\draw (u) edge 
		node[pos=0.3, sloped, below] {{\tiny \RelSymE(1,2), \RelSymR(1,2)}}
		node[near end, sloped, above] {{\tiny \RelSymE(2,1), \RelSymR(2,1)}}
	(v);
	\draw (u) edge
		node[near start, sloped, below] {{\tiny \RelSymE(2,1), \RelSymR(1,3)}}
		node[near end, sloped, above] {{\tiny \RelSymE(1,2), \RelSymR(3,1)}}
	(w);

	\draw (v) edge
		node[pos=0.3, sloped, above] {{\tiny \RelSymE(1,2), \RelSymR(2,3)}}
		node[pos=0.7, sloped, below] {{\tiny \RelSymE(2,1), \RelSymR(3,2)}}
	(w);

\end{tikzpicture} 		\caption{Enriched Gaifman graph of $\enGF{\AdExa}$ $\AdExa$.}%
		\label{fig:ex:advanced-gaifman-a}
	\end{subfigure}
	\hfill
	\begin{subfigure}{0.495\textwidth}
		\centering
\begin{tikzpicture}[
	edge from parent/.style={draw=none},
	every node/.style={mynode},
	level/.style={level distance=2.5cm, sibling distance=2.5cm},
	scale=1,
	grow=right,
	]
	\clip (-0.25,-2) rectangle (7.75,2);
	\node (1) at (0,0) {$1$}
	child {
		node (2) {$2$}
		child {
			node (w) {$w$}
			child[missing]
			child {
				node (u) {$u$}
			}
		}
	}
	child {
		node (3) {3}
		child {
			node (v) {$v$}
		}
	};
	\draw (1) edge 
		node[pos=0.3, sloped, above] {{\tiny \RelSymE(1,2), \RelSymR(1,2)}}
		node[near end, sloped, below] {{\tiny \RelSymE(2,1), \RelSymR(2,1)}}
	(2);
	\draw (1) edge 
		node[near start, sloped, above] {{\tiny \RelSymE(2,1), \RelSymR(1,3)}}
		node[near end, sloped, below] {{\tiny \RelSymE(1,2), \RelSymR(3,1)}}
	(3);
	\draw (1) edge[out=90, in=135, distance=40pt] 
		node[pos=0.15, sloped, above] {{\tiny \RelSymR(1,2)}}
		node[pos=0.85, sloped, below] {{\tiny \RelSymR(2,1)}}
	(v);

	\draw (2) edge 
		node[pos=0.3, sloped, above] {{\tiny \RelSymF(1,2), \RelSymR(2,3)}}
		node[pos=0.7, sloped, below] {{\tiny \RelSymF(2,1), \RelSymR(3,2)}}
	(3);
	\draw (2) edge[out=-45, in=-90, distance=40pt]
		node[pos=0.15, sloped, above] {{\tiny \RelSymR(2,1)}}
		node[pos=0.85, sloped, below] {{\tiny \RelSymR(1,2)}}
	(u);
	\draw (2) edge
		node[near start, sloped, above] {{\tiny \RelSymE(1,2), \RelSymR(2,3)}}
		node[near end, sloped, below] {{\tiny \RelSymE(2,1), \RelSymR(3,2)}}
	(w);

	\draw (3) edge
		node[near start, sloped, above] {{\tiny \RelSymE(2,1), \RelSymR(3,2)}}
		node[near end, sloped, below] {{\tiny \RelSymE(1,2), \RelSymR(2,3)}}
	(v);

	\draw (u) edge 
		node[pos=0.3, sloped, below] {{\tiny \RelSymE(1,2), \RelSymR(1,2)}}
		node[near end, sloped, above] {{\tiny \RelSymE(2,1), \RelSymR(2,1)}}
	(v);
	\draw (u) edge
		node[near start, sloped, below] {{\tiny \RelSymE(2,1), \RelSymR(1,3)}}
		node[near end, sloped, above] {{\tiny \RelSymE(1,2), \RelSymR(3,1)}}
	(w);

	\draw (v) edge
		node[pos=0.3, sloped, above] {{\tiny \RelSymF(1,2), \RelSymR(2,3)}}
		node[pos=0.7, sloped, below] {{\tiny \RelSymF(2,1), \RelSymR(3,2)}}
	(w);

\end{tikzpicture} 		\caption{Enriched Gaifman graph $\enGF{\AdExb}$ of $\AdExb$.}%
		\label{fig:ex:advanced-gaifman-b}
	\end{subfigure}
	\medskip\vspace{8ex}

	\begin{subfigure}{0.495\textwidth}
		\centering
\begin{tikzpicture}[
	bn/.style={circle, fill=black},
	re/.style={draw=black, regular polygon, regular polygon sides=4, inner sep=3.25pt, fill=myred},
	be/.style={draw=black, regular polygon, regular polygon sides=5, fill=myblue},
	fe/.style={draw=black, regular polygon, regular polygon sides=6, fill=mygreen},
	xscale=0.8,
	yscale=1.5]
	\node[bn] (1) at (-1.25,0) {};
	\node[bn] (2) at (0.25,0) {};
	\node[bn] (3) at (1.75,0) {};
	\node[bn] (w) at (3.25,0) {};
	\node[bn] (v) at (4.75,0) {};
	\node[bn] (u) at (6.25,0) {};

	\node[re] (12) at (-0.5,0.75) {};
	\node[re] (23) at (0.7,0.75) {};
	\node[re] (31) at (1.9,0.75) {};
	\node[fe] (2w) at (1.9,-1) {};
	\node[re] (wu) at (3.1,.75) {};
	\node[re] (vw) at (4.3,.75) {};
	\node[re] (uv) at (5.5,.75) {};
	\node[fe] (v3) at (3.1,-1) {};
	\node[be] (123) at (-0.5,-1) {};
	\node[be] (1v3) at (0.7,-1) {};
	\node[be] (u2w) at (4.3,-1) {};
	\node[be] (uvw) at (5.5,-1) {};

	\draw[-latex] (1) -- node[sloped, pos=0.1, above=-2pt] {{\tiny 1}} (12);
	\draw[-latex] (1) -- node[sloped, pos=0.2, above=-2pt] {{\tiny 2}} (31);
	\draw[-latex] (1) -- node[sloped, pos=0.1, above=-2pt] {{\tiny 1}} (1v3);
	\draw[-latex] (1) -- node[sloped, pos=0.1, below=-2pt] {{\tiny 1}} (123);

	\draw[-latex] (2) -- node[sloped, pos=0.1, below=-2pt] {{\tiny 2}} (12);
	\draw[-latex] (2) -- node[sloped, pos=0.1, below=-2pt] {{\tiny 1}} (23);
	\draw[-latex] (2) -- node[sloped, pos=0.1, below=-2pt] {{\tiny 1}} (2w);
	\draw[-latex] (2) -- node[sloped, pos=0.1, above=-2pt] {{\tiny 2}} (u2w);
	\draw[-latex] (2) -- node[sloped, pos=0.1, above=-2pt] {{\tiny 2}} (123);

	\draw[-latex] (3) -- node[sloped, pos=0.1, below=-2pt] {{\tiny 2}} (23);
	\draw[-latex] (3) -- node[sloped, pos=0.1, below=-2pt] {{\tiny 1}} (31);
	\draw[-latex] (3) -- node[sloped, pos=0.1, above=-2pt] {{\tiny 2}} (v3);
	\draw[-latex] (3) -- node[sloped, pos=0.1, below=-2pt] {{\tiny 3}} (1v3);
	\draw[-latex] (3) -- node[sloped, pos=0.1, above=-2pt] {{\tiny 3}} (123);

	\draw[-latex] (u) -- node[sloped, pos=0.1, above=-2pt] {{\tiny 1}} (uv);
	\draw[-latex] (u) -- node[sloped, pos=0.2, above=-2pt] {{\tiny 2}} (wu);
	\draw[-latex] (u) -- node[sloped, pos=0.1, above=-2pt] {{\tiny 1}} (u2w);
	\draw[-latex] (u) -- node[sloped, pos=0.1, below=-2pt] {{\tiny 1}} (uvw);

	\draw[-latex] (v) -- node[sloped, pos=0.1, below=-2pt] {{\tiny 2}} (uv);
	\draw[-latex] (v) -- node[sloped, pos=0.1, below=-2pt] {{\tiny 1}} (vw);
	\draw[-latex] (v) -- node[sloped, pos=0.1, below=-2pt] {{\tiny 1}} (v3);
	\draw[-latex] (v) -- node[sloped, pos=0.1, above=-2pt] {{\tiny 2}} (1v3);
	\draw[-latex] (v) -- node[sloped, pos=0.1, above=-2pt] {{\tiny 2}} (uvw);

	\draw[-latex] (w) -- node[sloped, pos=0.1, below=-2pt] {{\tiny 2}} (vw);
	\draw[-latex] (w) -- node[sloped, pos=0.1, below=-2pt] {{\tiny 1}} (wu);
	\draw[-latex] (w) -- node[sloped, pos=0.1, above=-2pt] {{\tiny 2}} (2w);
	\draw[-latex] (w) -- node[sloped, pos=0.1, below=-2pt] {{\tiny 3}} (u2w);
	\draw[-latex] (w) -- node[sloped, pos=0.1, above=-2pt] {{\tiny 3}} (uvw);

\end{tikzpicture} 		\caption{Enriched incidence graph $\enI{\AdExa}$ of $\AdExa$.}%
		\label{fig:ex:advanced-incidence-a}
	\end{subfigure}
	\hfill
	\begin{subfigure}{0.495\textwidth}
		\centering
\begin{tikzpicture}[
	bn/.style={circle, fill=black},
	re/.style={draw=black, regular polygon, regular polygon sides=4, inner sep=3.25pt, fill=myred},
	be/.style={draw=black, regular polygon, regular polygon sides=5, fill=myblue},
	fe/.style={draw=black, regular polygon, regular polygon sides=6, fill=mygreen},
	xscale=0.8,
	yscale=1.5]
	\node[bn] (1) at (-1.25,0) {};
	\node[bn] (2) at (0.25,0) {};
	\node[bn] (3) at (1.75,0) {};
	\node[bn] (w) at (3.25,0) {};
	\node[bn] (v) at (4.75,0) {};
	\node[bn] (u) at (6.25,0) {};

	\node[re] (12) at (-0.5,0.75) {};
	\node[fe] (23) at (0.7,0.75) {};
	\node[re] (31) at (1.9,0.75) {};
	\node[re] (2w) at (1.9,-1) {};
	\node[re] (wu) at (3.1,.75) {};
	\node[fe] (vw) at (4.3,.75) {};
	\node[re] (uv) at (5.5,.75) {};
	\node[re] (v3) at (3.1,-1) {};
	\node[be] (123) at (-0.5,-1) {};
	\node[be] (1v3) at (0.7,-1) {};
	\node[be] (u2w) at (4.3,-1) {};
	\node[be] (uvw) at (5.5,-1) {};

	\draw[-latex] (1) -- node[sloped, pos=0.1, above=-2pt] {{\tiny 1}} (12);
	\draw[-latex] (1) -- node[sloped, pos=0.2, above=-2pt] {{\tiny 2}} (31);
	\draw[-latex] (1) -- node[sloped, pos=0.1, above=-2pt] {{\tiny 1}} (1v3);
	\draw[-latex] (1) -- node[sloped, pos=0.1, below=-2pt] {{\tiny 1}} (123);

	\draw[-latex] (2) -- node[sloped, pos=0.1, below=-2pt] {{\tiny 2}} (12);
	\draw[-latex] (2) -- node[sloped, pos=0.1, below=-2pt] {{\tiny 1}} (23);
	\draw[-latex] (2) -- node[sloped, pos=0.1, below=-2pt] {{\tiny 1}} (2w);
	\draw[-latex] (2) -- node[sloped, pos=0.1, above=-2pt] {{\tiny 2}} (u2w);
	\draw[-latex] (2) -- node[sloped, pos=0.1, above=-2pt] {{\tiny 2}} (123);

	\draw[-latex] (3) -- node[sloped, pos=0.1, below=-2pt] {{\tiny 2}} (23);
	\draw[-latex] (3) -- node[sloped, pos=0.1, below=-2pt] {{\tiny 1}} (31);
	\draw[-latex] (3) -- node[sloped, pos=0.1, above=-2pt] {{\tiny 2}} (v3);
	\draw[-latex] (3) -- node[sloped, pos=0.1, below=-2pt] {{\tiny 3}} (1v3);
	\draw[-latex] (3) -- node[sloped, pos=0.1, above=-2pt] {{\tiny 3}} (123);

	\draw[-latex] (u) -- node[sloped, pos=0.1, above=-2pt] {{\tiny 1}} (uv);
	\draw[-latex] (u) -- node[sloped, pos=0.2, above=-2pt] {{\tiny 2}} (wu);
	\draw[-latex] (u) -- node[sloped, pos=0.1, above=-2pt] {{\tiny 1}} (u2w);
	\draw[-latex] (u) -- node[sloped, pos=0.1, below=-2pt] {{\tiny 1}} (uvw);

	\draw[-latex] (v) -- node[sloped, pos=0.1, below=-2pt] {{\tiny 2}} (uv);
	\draw[-latex] (v) -- node[sloped, pos=0.1, below=-2pt] {{\tiny 1}} (vw);
	\draw[-latex] (v) -- node[sloped, pos=0.1, below=-2pt] {{\tiny 1}} (v3);
	\draw[-latex] (v) -- node[sloped, pos=0.1, above=-2pt] {{\tiny 2}} (1v3);
	\draw[-latex] (v) -- node[sloped, pos=0.1, above=-2pt] {{\tiny 2}} (uvw);

	\draw[-latex] (w) -- node[sloped, pos=0.1, below=-2pt] {{\tiny 2}} (vw);
	\draw[-latex] (w) -- node[sloped, pos=0.1, below=-2pt] {{\tiny 1}} (wu);
	\draw[-latex] (w) -- node[sloped, pos=0.1, above=-2pt] {{\tiny 2}} (2w);
	\draw[-latex] (w) -- node[sloped, pos=0.1, below=-2pt] {{\tiny 3}} (u2w);
	\draw[-latex] (w) -- node[sloped, pos=0.1, above=-2pt] {{\tiny 3}} (uvw);

\end{tikzpicture} 		\caption{Enriched incidence graph $\enI{\AdExb}$ of $\AdExb$.}%
		\label{fig:ex:advanced-incidence-b}
	\end{subfigure}
	\caption{Sophisticated example of two structures $\AdExa$, $\AdExb$ distinguished by Relational Color Refinement.}
\end{figure}

We already saw that running Color Refinement on the Gaifman graph (or on the incidence graph) is not very powerful.
In this part of the paper, we will explore more sophisticated ways of adapting Color Refinement based on the Gaifman or on the incidence graph:
To increase the power of the Color Refinement algorithm on Relational Structures, one could try to encode as much information about the original structure into the Gaifman graph (the incidence graph) as possible.
Here are two naive attempts.

We let the \emph{enriched Gaifman graph} $\enGF{\A}$ of $\A$ be defined by $\V(\enGF{\A}) \isdef \V(\A)$ and for all $R \in \sig$ and all $i, j \in [\ar(R)]$ with $i \neq j$ we let
\begin{equation*}
	{(E_{R(i,j)})}^{\enGF{\A}} \ \isdef \  \set{ (u,v) \mid R \in \sig, \at \in R^{\A}, u = a_i, v = a_j } \;.
\end{equation*}
I.e., we enrich the edge $\set{u,v}$ in the Gaifman graph with the information, in which relation, and on which position, $u$ and $v$ appear in the same tuple.

Similarly, we let the \emph{enriched incidence graph} $\enI{\A}$ of $\A$ be defined by $\V(\enI{\A})  \isdef \V(\A) \union \set{ w_{R(\at)} \!\mid\! \at \in R^{\A} }[]$, and for all $i \in [\ar(\sig)]$ we let
\begin{equation*}
	{(E_{i})}^{\enI{\A}} \ \isdef \ \set{ (a_i, w_{R(\at)}) \mid R \in \sig, \at \in R^{\A} } \;.
\end{equation*}
I.e., we enrich the edge $(u, \at)$ in the incidence graph with the information, on which position $u$ is in $\at$.

We will now present a pair of relational structures that is distinguished by RCR, but where Color Refinement distinguishes neither their enriched Gaifman graphs, nor their enriched incidence graphs.
Let $\AdsigEx \isdef \set{ E/2, F/2, R/3 }$ and consider the
structures $\AdExa$, $\AdExb$ over the universe $\set{ 1,2,3,u,v,w }$
with\\
$E^{\AdExa} \isdef \set{ (1,2), (2,3), (3,1), (u,v), (v,w), (w,u) }$,
\ $E^{\AdExb} \isdef \set{ (1,2), (2,w), (w,u), (u,v), (v,3) (3,1) }$,
\\
$F^{\AdExa} \isdef \set{ (2,w), (v,3) }$, \ $F^{\AdExb} \isdef \set{
  (2,3), (v,w) }$ \ and \
$R^{\AdExa} \isdef \set{ (1,2,3), (u,v,w), (1,v,3), (u,2,w) } \defis R^{\AdExb}$.
See \cref{fig:ex:advanced-hypergraph-a,fig:ex:advanced-hypergraph-b} for a graphical representation of $\AdExa$, $\AdExb$ and see \cref{fig:ex:advanced-gaifman-a,fig:ex:advanced-gaifman-b} for the enriched Gaifman graphs $\enGF{\AdExa}$, $\enGF{\AdExb}$ and \cref{fig:ex:advanced-incidence-a,fig:ex:advanced-incidence-b} for the enriched incidence graphs $\enI{\AdExa}$, $\enI{\AdExb}$.

It follows from \cref{thm:main-logic}, that Relational Color Refinement distinguishes $\AdExa$, $\AdExb$, since $\AdExa \models \phi$ and $\AdExb \not\models \phi$ for
	$\phi \isdef \existsgeq{1}(\var_1, \var_2, \var_3) \qsep \bigl(
		R(\var_1, \var_2, \var_3) 
		\land \bigl(
			E(\var_1, \var_2) 
			\land (\,E(\var_2, \var_3) 
			\land E(\var_3, \var_1)\,)
		\bigr)
	\bigr)$.
However, we will now see that Color Refinement does not distinguish $\enGF{\AdExa}$, $\enGF{\AdExb}$, and it does not distinguish $\enI{\AdExa}$, $\enI{\AdExb}$.

\paragraph*{Color Refinement does not distinguish $\enGF{\AdExa}$, $\enGF{\AdExb}$.}

The first round yields a coloring equivalent to the following:\\
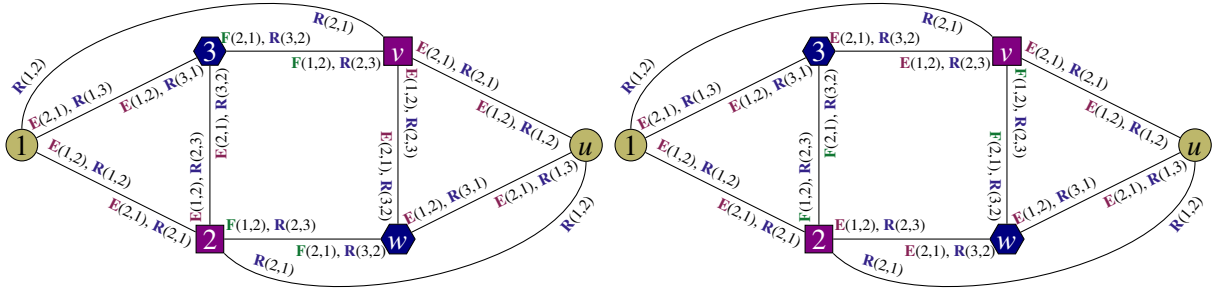
\begin{tikzpicture}[
	edge from parent/.style={draw=none},
	every node/.style={mynode},
	level/.style={level distance=2.5cm, sibling distance=2.5cm},
	scale=1,
	grow=right,
	]
	\clip (-0.25,-2) rectangle (7.75,2);
	\node[colourC] (1) at (0,0) {$1$}
	child {
		node[colourB, inner sep=0pt] (2) {$2$}
		child {
			node[colourA, inner sep=0pt] (w) {$\vphantom{1}w$}
			child[missing]
			child {
				node[colourC] (u) {$\vphantom{1}u$}
			}
		}
	}
	child {
		node[colourA, inner sep=0pt] (3) {3}
		child {
			node[colourB, inner sep=0pt] (v) {$\vphantom{1}v$}
		}
	};
	\draw (1) edge 
		node[pos=0.3, sloped, above] {{\tiny \RelSymE(1,2), \RelSymR(1,2)}}
		node[near end, sloped, below] {{\tiny \RelSymE(2,1), \RelSymR(2,1)}}
	(2);
	\draw (1) edge 
		node[near start, sloped, above] {{\tiny \RelSymE(2,1), \RelSymR(1,3)}}
		node[near end, sloped, below] {{\tiny \RelSymE(1,2), \RelSymR(3,1)}}
	(3);
	\draw (1) edge[out=90, in=135, distance=40pt] 
		node[pos=0.15, sloped, above] {{\tiny \RelSymR(1,2)}}
		node[pos=0.85, sloped, below] {{\tiny \RelSymR(2,1)}}
	(v);

	\draw (2) edge 
		node[pos=0.3, sloped, above] {{\tiny \RelSymE(1,2), \RelSymR(2,3)}}
		node[pos=0.7, sloped, below] {{\tiny \RelSymE(2,1), \RelSymR(3,2)}}
	(3);
	\draw (2) edge[out=-45, in=-90, distance=40pt]
		node[pos=0.15, sloped, above] {{\tiny \RelSymR(2,1)}}
		node[pos=0.85, sloped, below] {{\tiny \RelSymR(1,2)}}
	(u);
	\draw (2) edge
		node[pos=0.3, sloped, above] {{\tiny \RelSymF(1,2), \RelSymR(2,3)}}
		node[near end, sloped, below] {{\tiny \RelSymF(2,1), \RelSymR(3,2)}}
	(w);

	\draw (3) edge
		node[near start, sloped, above] {{\tiny \RelSymF(2,1), \RelSymR(3,2)}}
		node[pos=0.7, sloped, below] {{\tiny \RelSymF(1,2), \RelSymR(2,3)}}
	(v);

	\draw (u) edge 
		node[pos=0.3, sloped, below] {{\tiny \RelSymE(1,2), \RelSymR(1,2)}}
		node[near end, sloped, above] {{\tiny \RelSymE(2,1), \RelSymR(2,1)}}
	(v);
	\draw (u) edge
		node[near start, sloped, below] {{\tiny \RelSymE(2,1), \RelSymR(1,3)}}
		node[near end, sloped, above] {{\tiny \RelSymE(1,2), \RelSymR(3,1)}}
	(w);

	\draw (v) edge
		node[pos=0.3, sloped, above] {{\tiny \RelSymE(1,2), \RelSymR(2,3)}}
		node[pos=0.7, sloped, below] {{\tiny \RelSymE(2,1), \RelSymR(3,2)}}
	(w);

\end{tikzpicture} %
\begin{tikzpicture}[
	edge from parent/.style={draw=none},
	every node/.style={mynode},
	level/.style={level distance=2.5cm, sibling distance=2.5cm},
	scale=1,
	grow=right,
	]
	\clip (-0.25,-2) rectangle (7.75,2);
	\node[colourC] (1) at (0,0) {$1$}
	child {
		node[colourB, inner sep=0pt] (2) {$2$}
		child {
			node[colourA, inner sep=0pt] (w) {$\vphantom{1}w$}
			child[missing]
			child {
				node[colourC] (u) {$\vphantom{1}u$}
			}
		}
	}
	child {
		node[colourA, inner sep=0pt] (3) {3}
		child {
			node[colourB, inner sep=0pt] (v) {$\vphantom{1}v$}
		}
	};
	\draw (1) edge 
		node[pos=0.3, sloped, above] {{\tiny \RelSymE(1,2), \RelSymR(1,2)}}
		node[near end, sloped, below] {{\tiny \RelSymE(2,1), \RelSymR(2,1)}}
	(2);
	\draw (1) edge 
		node[near start, sloped, above] {{\tiny \RelSymE(2,1), \RelSymR(1,3)}}
		node[near end, sloped, below] {{\tiny \RelSymE(1,2), \RelSymR(3,1)}}
	(3);
	\draw (1) edge[out=90, in=135, distance=40pt] 
		node[pos=0.15, sloped, above] {{\tiny \RelSymR(1,2)}}
		node[pos=0.85, sloped, below] {{\tiny \RelSymR(2,1)}}
	(v);

	\draw (2) edge 
		node[pos=0.3, sloped, above] {{\tiny \RelSymF(1,2), \RelSymR(2,3)}}
		node[pos=0.7, sloped, below] {{\tiny \RelSymF(2,1), \RelSymR(3,2)}}
	(3);
	\draw (2) edge[out=-45, in=-90, distance=40pt]
		node[pos=0.15, sloped, above] {{\tiny \RelSymR(2,1)}}
		node[pos=0.85, sloped, below] {{\tiny \RelSymR(1,2)}}
	(u);
	\draw (2) edge
		node[pos=0.3, sloped, above] {{\tiny \RelSymE(1,2), \RelSymR(2,3)}}
		node[near end, sloped, below] {{\tiny \RelSymE(2,1), \RelSymR(3,2)}}
	(w);

	\draw (3) edge
		node[near start, sloped, above] {{\tiny \RelSymE(2,1), \RelSymR(3,2)}}
		node[pos=0.7, sloped, below] {{\tiny \RelSymE(1,2), \RelSymR(2,3)}}
	(v);

	\draw (u) edge 
		node[pos=0.3, sloped, below] {{\tiny \RelSymE(1,2), \RelSymR(1,2)}}
		node[near end, sloped, above] {{\tiny \RelSymE(2,1), \RelSymR(2,1)}}
	(v);
	\draw (u) edge
		node[near start, sloped, below] {{\tiny \RelSymE(2,1), \RelSymR(1,3)}}
		node[near end, sloped, above] {{\tiny \RelSymE(1,2), \RelSymR(3,1)}}
	(w);

	\draw (v) edge
		node[pos=0.3, sloped, above] {{\tiny \RelSymF(1,2), \RelSymR(2,3)}}
		node[pos=0.7, sloped, below] {{\tiny \RelSymF(2,1), \RelSymR(3,2)}}
	(w);

\end{tikzpicture} \medskip

\noindent The second round yields a coloring equivalent to the following:\\
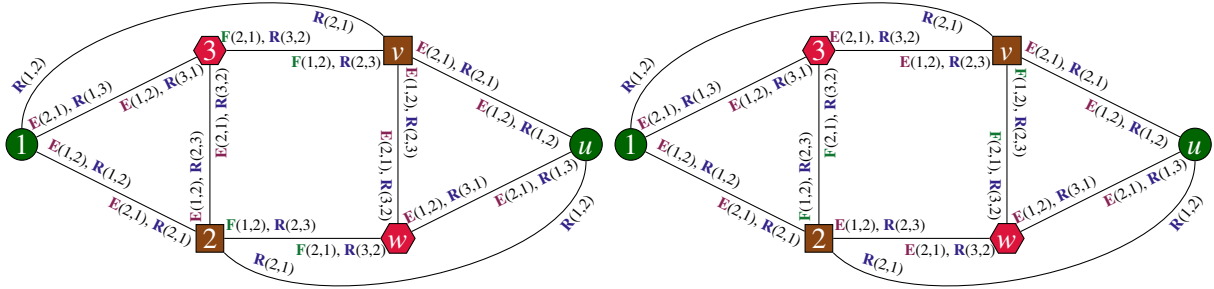
\begin{tikzpicture}[
	edge from parent/.style={draw=none},
	every node/.style={mynode},
	level/.style={level distance=2.5cm, sibling distance=2.5cm},
	scale=1,
	grow=right,
	]
	\clip (-0.25,-2) rectangle (7.75,2);
	\node[colourJ] (1) at (0,0) {$1$}
	child {
		node[colourK, inner sep=0pt] (2) {$2$}
		child {
			node[colourG, inner sep=0pt] (w) {$\vphantom{1}w$}
			child[missing]
			child {
				node[colourJ] (u) {$\vphantom{1}u$}
			}
		}
	}
	child {
		node[colourG, inner sep=0pt] (3) {3}
		child {
			node[colourK, inner sep=0pt] (v) {$\vphantom{1}v$}
		}
	};
	\draw (1) edge 
		node[pos=0.3, sloped, above] {{\tiny \RelSymE(1,2), \RelSymR(1,2)}}
		node[near end, sloped, below] {{\tiny \RelSymE(2,1), \RelSymR(2,1)}}
	(2);
	\draw (1) edge 
		node[near start, sloped, above] {{\tiny \RelSymE(2,1), \RelSymR(1,3)}}
		node[near end, sloped, below] {{\tiny \RelSymE(1,2), \RelSymR(3,1)}}
	(3);
	\draw (1) edge[out=90, in=135, distance=40pt] 
		node[pos=0.15, sloped, above] {{\tiny \RelSymR(1,2)}}
		node[pos=0.85, sloped, below] {{\tiny \RelSymR(2,1)}}
	(v);

	\draw (2) edge 
		node[pos=0.3, sloped, above] {{\tiny \RelSymE(1,2), \RelSymR(2,3)}}
		node[pos=0.7, sloped, below] {{\tiny \RelSymE(2,1), \RelSymR(3,2)}}
	(3);
	\draw (2) edge[out=-45, in=-90, distance=40pt]
		node[pos=0.15, sloped, above] {{\tiny \RelSymR(2,1)}}
		node[pos=0.85, sloped, below] {{\tiny \RelSymR(1,2)}}
	(u);
	\draw (2) edge
		node[pos=0.3, sloped, above] {{\tiny \RelSymF(1,2), \RelSymR(2,3)}}
		node[near end, sloped, below] {{\tiny \RelSymF(2,1), \RelSymR(3,2)}}
	(w);

	\draw (3) edge
		node[near start, sloped, above] {{\tiny \RelSymF(2,1), \RelSymR(3,2)}}
		node[pos=0.7, sloped, below] {{\tiny \RelSymF(1,2), \RelSymR(2,3)}}
	(v);

	\draw (u) edge 
		node[pos=0.3, sloped, below] {{\tiny \RelSymE(1,2), \RelSymR(1,2)}}
		node[near end, sloped, above] {{\tiny \RelSymE(2,1), \RelSymR(2,1)}}
	(v);
	\draw (u) edge
		node[near start, sloped, below] {{\tiny \RelSymE(2,1), \RelSymR(1,3)}}
		node[near end, sloped, above] {{\tiny \RelSymE(1,2), \RelSymR(3,1)}}
	(w);

	\draw (v) edge
		node[pos=0.3, sloped, above] {{\tiny \RelSymE(1,2), \RelSymR(2,3)}}
		node[pos=0.7, sloped, below] {{\tiny \RelSymE(2,1), \RelSymR(3,2)}}
	(w);

\end{tikzpicture} %
\begin{tikzpicture}[
	edge from parent/.style={draw=none},
	every node/.style={mynode},
	level/.style={level distance=2.5cm, sibling distance=2.5cm},
	scale=1,
	grow=right,
	]
	\clip (-0.25,-2) rectangle (7.75,2);
	\node[colourJ] (1) at (0,0) {$1$}
	child {
		node[colourK, inner sep=0pt] (2) {$2$}
		child {
			node[colourG, inner sep=0pt] (w) {$\vphantom{1}w$}
			child[missing]
			child {
				node[colourJ] (u) {$\vphantom{1}u$}
			}
		}
	}
	child {
		node[colourG, inner sep=0pt] (3) {3}
		child {
			node[colourK, inner sep=0pt] (v) {$\vphantom{1}v$}
		}
	};
	\draw (1) edge 
		node[pos=0.3, sloped, above] {{\tiny \RelSymE(1,2), \RelSymR(1,2)}}
		node[near end, sloped, below] {{\tiny \RelSymE(2,1), \RelSymR(2,1)}}
	(2);
	\draw (1) edge 
		node[near start, sloped, above] {{\tiny \RelSymE(2,1), \RelSymR(1,3)}}
		node[near end, sloped, below] {{\tiny \RelSymE(1,2), \RelSymR(3,1)}}
	(3);
	\draw (1) edge[out=90, in=135, distance=40pt] 
		node[pos=0.15, sloped, above] {{\tiny \RelSymR(1,2)}}
		node[pos=0.85, sloped, below] {{\tiny \RelSymR(2,1)}}
	(v);

	\draw (2) edge 
		node[pos=0.3, sloped, above] {{\tiny \RelSymF(1,2), \RelSymR(2,3)}}
		node[pos=0.7, sloped, below] {{\tiny \RelSymF(2,1), \RelSymR(3,2)}}
	(3);
	\draw (2) edge[out=-45, in=-90, distance=40pt]
		node[pos=0.15, sloped, above] {{\tiny \RelSymR(2,1)}}
		node[pos=0.85, sloped, below] {{\tiny \RelSymR(1,2)}}
	(u);
	\draw (2) edge
		node[pos=0.3, sloped, above] {{\tiny \RelSymE(1,2), \RelSymR(2,3)}}
		node[near end, sloped, below] {{\tiny \RelSymE(2,1), \RelSymR(3,2)}}
	(w);

	\draw (3) edge
		node[near start, sloped, above] {{\tiny \RelSymE(2,1), \RelSymR(3,2)}}
		node[pos=0.7, sloped, below] {{\tiny \RelSymE(1,2), \RelSymR(2,3)}}
	(v);

	\draw (u) edge 
		node[pos=0.3, sloped, below] {{\tiny \RelSymE(1,2), \RelSymR(1,2)}}
		node[near end, sloped, above] {{\tiny \RelSymE(2,1), \RelSymR(2,1)}}
	(v);
	\draw (u) edge
		node[near start, sloped, below] {{\tiny \RelSymE(2,1), \RelSymR(1,3)}}
		node[near end, sloped, above] {{\tiny \RelSymE(1,2), \RelSymR(3,1)}}
	(w);

	\draw (v) edge
		node[pos=0.3, sloped, above] {{\tiny \RelSymF(1,2), \RelSymR(2,3)}}
		node[pos=0.7, sloped, below] {{\tiny \RelSymF(2,1), \RelSymR(3,2)}}
	(w);

\end{tikzpicture} \medskip

\noindent Since the coloring is stable, Color Refinement stops.
It does not distinguish $\enGF{\AdExa}$, $\enGF{\AdExb}$.

\paragraph*{Color Refinement does not distinguish $\enI{\AdExa}$, $\enI{\AdExb}$.}

The first round yields a coloring equivalent to the following:
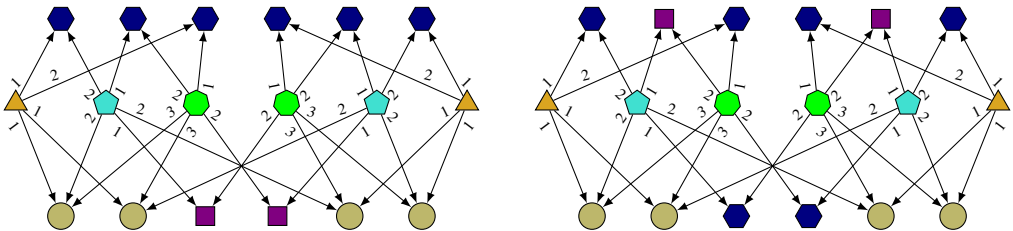
\begin{center}
\begin{tikzpicture}[
	bn/.style={circle, fill=black},
	re/.style={circle, fill=myred},
	be/.style={circle, fill=myblue},
	fe/.style={circle, fill=mygreen},
	xscale=0.8,
	yscale=1.5]
	\node[colourD] (1) at (-1.25,0) {};
	\node[colourE] (2) at (0.25,0) {};
	\node[colourF] (3) at (1.75,0) {};
	\node[colourF] (w) at (3.25,0) {};
	\node[colourE] (v) at (4.75,0) {};
	\node[colourD] (u) at (6.25,0) {};

	\node[colourA] (12) at (-0.5,0.75) {};
	\node[colourA] (23) at (0.7,0.75) {};
	\node[colourA] (31) at (1.9,0.75) {};
	\node[colourB] (2w) at (1.9,-1) {};
	\node[colourA] (wu) at (3.1,.75) {};
	\node[colourA] (vw) at (4.3,.75) {};
	\node[colourA] (uv) at (5.5,.75) {};
	\node[colourB] (v3) at (3.1,-1) {};
	\node[colourC] (123) at (-0.5,-1) {};
	\node[colourC] (1v3) at (0.7,-1) {};
	\node[colourC] (u2w) at (4.3,-1) {};
	\node[colourC] (uvw) at (5.5,-1) {};

	\draw[-latex] (1) -- node[sloped, pos=0.1, above=-2pt] {{\tiny 1}} (12);
	\draw[-latex] (1) -- node[sloped, pos=0.2, above=-2pt] {{\tiny 2}} (31);
	\draw[-latex] (1) -- node[sloped, pos=0.1, above=-2pt] {{\tiny 1}} (1v3);
	\draw[-latex] (1) -- node[sloped, pos=0.1, below=-2pt] {{\tiny 1}} (123);

	\draw[-latex] (2) -- node[sloped, pos=0.1, below=-2pt] {{\tiny 2}} (12);
	\draw[-latex] (2) -- node[sloped, pos=0.1, below=-2pt] {{\tiny 1}} (23);
	\draw[-latex] (2) -- node[sloped, pos=0.1, below=-2pt] {{\tiny 1}} (2w);
	\draw[-latex] (2) -- node[sloped, pos=0.1, above=-2pt] {{\tiny 2}} (u2w);
	\draw[-latex] (2) -- node[sloped, pos=0.1, above=-2pt] {{\tiny 2}} (123);

	\draw[-latex] (3) -- node[sloped, pos=0.1, below=-2pt] {{\tiny 2}} (23);
	\draw[-latex] (3) -- node[sloped, pos=0.1, below=-2pt] {{\tiny 1}} (31);
	\draw[-latex] (3) -- node[sloped, pos=0.1, above=-2pt] {{\tiny 2}} (v3);
	\draw[-latex] (3) -- node[sloped, pos=0.1, below=-2pt] {{\tiny 3}} (1v3);
	\draw[-latex] (3) -- node[sloped, pos=0.1, above=-2pt] {{\tiny 3}} (123);

	\draw[-latex] (u) -- node[sloped, pos=0.1, above=-2pt] {{\tiny 1}} (uv);
	\draw[-latex] (u) -- node[sloped, pos=0.2, above=-2pt] {{\tiny 2}} (wu);
	\draw[-latex] (u) -- node[sloped, pos=0.1, above=-2pt] {{\tiny 1}} (u2w);
	\draw[-latex] (u) -- node[sloped, pos=0.1, below=-2pt] {{\tiny 1}} (uvw);

	\draw[-latex] (v) -- node[sloped, pos=0.1, below=-2pt] {{\tiny 2}} (uv);
	\draw[-latex] (v) -- node[sloped, pos=0.1, below=-2pt] {{\tiny 1}} (vw);
	\draw[-latex] (v) -- node[sloped, pos=0.1, below=-2pt] {{\tiny 1}} (v3);
	\draw[-latex] (v) -- node[sloped, pos=0.1, above=-2pt] {{\tiny 2}} (1v3);
	\draw[-latex] (v) -- node[sloped, pos=0.1, above=-2pt] {{\tiny 2}} (uvw);

	\draw[-latex] (w) -- node[sloped, pos=0.1, below=-2pt] {{\tiny 2}} (vw);
	\draw[-latex] (w) -- node[sloped, pos=0.1, below=-2pt] {{\tiny 1}} (wu);
	\draw[-latex] (w) -- node[sloped, pos=0.1, above=-2pt] {{\tiny 2}} (2w);
	\draw[-latex] (w) -- node[sloped, pos=0.1, below=-2pt] {{\tiny 3}} (u2w);
	\draw[-latex] (w) -- node[sloped, pos=0.1, above=-2pt] {{\tiny 3}} (uvw);

\end{tikzpicture} 	\quad
\begin{tikzpicture}[
	bn/.style={circle, fill=black},
	re/.style={circle, fill=myred},
	be/.style={circle, fill=myblue},
	fe/.style={circle, fill=mygreen},
	xscale=0.8,
	yscale=1.5]
	\node[colourD] (1) at (-1.25,0) {};
	\node[colourE] (2) at (0.25,0) {};
	\node[colourF] (3) at (1.75,0) {};
	\node[colourF] (w) at (3.25,0) {};
	\node[colourE] (v) at (4.75,0) {};
	\node[colourD] (u) at (6.25,0) {};

	\node[colourA] (12) at (-0.5,0.75) {};
	\node[colourB] (23) at (0.7,0.75) {};
	\node[colourA] (31) at (1.9,0.75) {};
	\node[colourA] (2w) at (1.9,-1) {};
	\node[colourA] (wu) at (3.1,.75) {};
	\node[colourB] (vw) at (4.3,.75) {};
	\node[colourA] (uv) at (5.5,.75) {};
	\node[colourA] (v3) at (3.1,-1) {};
	\node[colourC] (123) at (-0.5,-1) {};
	\node[colourC] (1v3) at (0.7,-1) {};
	\node[colourC] (u2w) at (4.3,-1) {};
	\node[colourC] (uvw) at (5.5,-1) {};

	\draw[-latex] (1) -- node[sloped, pos=0.1, above=-2pt] {{\tiny 1}} (12);
	\draw[-latex] (1) -- node[sloped, pos=0.2, above=-2pt] {{\tiny 2}} (31);
	\draw[-latex] (1) -- node[sloped, pos=0.1, above=-2pt] {{\tiny 1}} (1v3);
	\draw[-latex] (1) -- node[sloped, pos=0.1, below=-2pt] {{\tiny 1}} (123);

	\draw[-latex] (2) -- node[sloped, pos=0.1, below=-2pt] {{\tiny 2}} (12);
	\draw[-latex] (2) -- node[sloped, pos=0.1, below=-2pt] {{\tiny 1}} (23);
	\draw[-latex] (2) -- node[sloped, pos=0.1, below=-2pt] {{\tiny 1}} (2w);
	\draw[-latex] (2) -- node[sloped, pos=0.1, above=-2pt] {{\tiny 2}} (u2w);
	\draw[-latex] (2) -- node[sloped, pos=0.1, above=-2pt] {{\tiny 2}} (123);

	\draw[-latex] (3) -- node[sloped, pos=0.1, below=-2pt] {{\tiny 2}} (23);
	\draw[-latex] (3) -- node[sloped, pos=0.1, below=-2pt] {{\tiny 1}} (31);
	\draw[-latex] (3) -- node[sloped, pos=0.1, above=-2pt] {{\tiny 2}} (v3);
	\draw[-latex] (3) -- node[sloped, pos=0.1, below=-2pt] {{\tiny 3}} (1v3);
	\draw[-latex] (3) -- node[sloped, pos=0.1, above=-2pt] {{\tiny 3}} (123);

	\draw[-latex] (u) -- node[sloped, pos=0.1, above=-2pt] {{\tiny 1}} (uv);
	\draw[-latex] (u) -- node[sloped, pos=0.2, above=-2pt] {{\tiny 2}} (wu);
	\draw[-latex] (u) -- node[sloped, pos=0.1, above=-2pt] {{\tiny 1}} (u2w);
	\draw[-latex] (u) -- node[sloped, pos=0.1, below=-2pt] {{\tiny 1}} (uvw);

	\draw[-latex] (v) -- node[sloped, pos=0.1, below=-2pt] {{\tiny 2}} (uv);
	\draw[-latex] (v) -- node[sloped, pos=0.1, below=-2pt] {{\tiny 1}} (vw);
	\draw[-latex] (v) -- node[sloped, pos=0.1, below=-2pt] {{\tiny 1}} (v3);
	\draw[-latex] (v) -- node[sloped, pos=0.1, above=-2pt] {{\tiny 2}} (1v3);
	\draw[-latex] (v) -- node[sloped, pos=0.1, above=-2pt] {{\tiny 2}} (uvw);

	\draw[-latex] (w) -- node[sloped, pos=0.1, below=-2pt] {{\tiny 2}} (vw);
	\draw[-latex] (w) -- node[sloped, pos=0.1, below=-2pt] {{\tiny 1}} (wu);
	\draw[-latex] (w) -- node[sloped, pos=0.1, above=-2pt] {{\tiny 2}} (2w);
	\draw[-latex] (w) -- node[sloped, pos=0.1, below=-2pt] {{\tiny 3}} (u2w);
	\draw[-latex] (w) -- node[sloped, pos=0.1, above=-2pt] {{\tiny 3}} (uvw);

\end{tikzpicture} \end{center}
\medskip

\noindent The second round yields a coloring equivalent to the following:
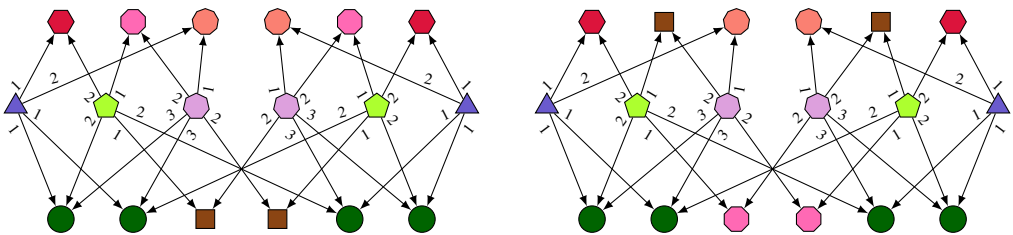
\begin{center}
\begin{tikzpicture}[
	bn/.style={circle, fill=black},
	re/.style={circle, fill=myred},
	be/.style={circle, fill=myblue},
	fe/.style={circle, fill=mygreen},
	xscale=0.8,
	yscale=1.5]
	\node[colourL] (1) at (-1.25,0) {};
	\node[colourM] (2) at (0.25,0) {};
	\node[colourN] (3) at (1.75,0) {};
	\node[colourN] (w) at (3.25,0) {};
	\node[colourM] (v) at (4.75,0) {};
	\node[colourL] (u) at (6.25,0) {};

	\node[colourG] (12) at (-0.5,0.75) {};
	\node[colourH] (23) at (0.7,0.75) {};
	\node[colourI] (31) at (1.9,0.75) {};
	\node[colourK] (2w) at (1.9,-1) {};
	\node[colourI] (wu) at (3.1,.75) {};
	\node[colourH] (vw) at (4.3,.75) {};
	\node[colourG] (uv) at (5.5,.75) {};
	\node[colourK] (v3) at (3.1,-1) {};
	\node[colourJ] (123) at (-0.5,-1) {};
	\node[colourJ] (1v3) at (0.7,-1) {};
	\node[colourJ] (u2w) at (4.3,-1) {};
	\node[colourJ] (uvw) at (5.5,-1) {};

	\draw[-latex] (1) -- node[sloped, pos=0.1, above=-2pt] {{\tiny 1}} (12);
	\draw[-latex] (1) -- node[sloped, pos=0.2, above=-2pt] {{\tiny 2}} (31);
	\draw[-latex] (1) -- node[sloped, pos=0.1, above=-2pt] {{\tiny 1}} (1v3);
	\draw[-latex] (1) -- node[sloped, pos=0.1, below=-2pt] {{\tiny 1}} (123);

	\draw[-latex] (2) -- node[sloped, pos=0.1, below=-2pt] {{\tiny 2}} (12);
	\draw[-latex] (2) -- node[sloped, pos=0.1, below=-2pt] {{\tiny 1}} (23);
	\draw[-latex] (2) -- node[sloped, pos=0.1, below=-2pt] {{\tiny 1}} (2w);
	\draw[-latex] (2) -- node[sloped, pos=0.1, above=-2pt] {{\tiny 2}} (u2w);
	\draw[-latex] (2) -- node[sloped, pos=0.1, above=-2pt] {{\tiny 2}} (123);

	\draw[-latex] (3) -- node[sloped, pos=0.1, below=-2pt] {{\tiny 2}} (23);
	\draw[-latex] (3) -- node[sloped, pos=0.1, below=-2pt] {{\tiny 1}} (31);
	\draw[-latex] (3) -- node[sloped, pos=0.1, above=-2pt] {{\tiny 2}} (v3);
	\draw[-latex] (3) -- node[sloped, pos=0.1, below=-2pt] {{\tiny 3}} (1v3);
	\draw[-latex] (3) -- node[sloped, pos=0.1, above=-2pt] {{\tiny 3}} (123);

	\draw[-latex] (u) -- node[sloped, pos=0.1, above=-2pt] {{\tiny 1}} (uv);
	\draw[-latex] (u) -- node[sloped, pos=0.2, above=-2pt] {{\tiny 2}} (wu);
	\draw[-latex] (u) -- node[sloped, pos=0.1, above=-2pt] {{\tiny 1}} (u2w);
	\draw[-latex] (u) -- node[sloped, pos=0.1, below=-2pt] {{\tiny 1}} (uvw);

	\draw[-latex] (v) -- node[sloped, pos=0.1, below=-2pt] {{\tiny 2}} (uv);
	\draw[-latex] (v) -- node[sloped, pos=0.1, below=-2pt] {{\tiny 1}} (vw);
	\draw[-latex] (v) -- node[sloped, pos=0.1, below=-2pt] {{\tiny 1}} (v3);
	\draw[-latex] (v) -- node[sloped, pos=0.1, above=-2pt] {{\tiny 2}} (1v3);
	\draw[-latex] (v) -- node[sloped, pos=0.1, above=-2pt] {{\tiny 2}} (uvw);

	\draw[-latex] (w) -- node[sloped, pos=0.1, below=-2pt] {{\tiny 2}} (vw);
	\draw[-latex] (w) -- node[sloped, pos=0.1, below=-2pt] {{\tiny 1}} (wu);
	\draw[-latex] (w) -- node[sloped, pos=0.1, above=-2pt] {{\tiny 2}} (2w);
	\draw[-latex] (w) -- node[sloped, pos=0.1, below=-2pt] {{\tiny 3}} (u2w);
	\draw[-latex] (w) -- node[sloped, pos=0.1, above=-2pt] {{\tiny 3}} (uvw);

\end{tikzpicture} 	\quad
\begin{tikzpicture}[
	bn/.style={circle, fill=black},
	re/.style={circle, fill=myred},
	be/.style={circle, fill=myblue},
	fe/.style={circle, fill=mygreen},
	xscale=0.8,
	yscale=1.5]
	\node[colourL] (1) at (-1.25,0) {};
	\node[colourM] (2) at (0.25,0) {};
	\node[colourN] (3) at (1.75,0) {};
	\node[colourN] (w) at (3.25,0) {};
	\node[colourM] (v) at (4.75,0) {};
	\node[colourL] (u) at (6.25,0) {};

	\node[colourG] (12) at (-0.5,0.75) {};
	\node[colourK] (23) at (0.7,0.75) {};
	\node[colourI] (31) at (1.9,0.75) {};
	\node[colourH] (2w) at (1.9,-1) {};
	\node[colourI] (wu) at (3.1,.75) {};
	\node[colourK] (vw) at (4.3,.75) {};
	\node[colourG] (uv) at (5.5,.75) {};
	\node[colourH] (v3) at (3.1,-1) {};
	\node[colourJ] (123) at (-0.5,-1) {};
	\node[colourJ] (1v3) at (0.7,-1) {};
	\node[colourJ] (u2w) at (4.3,-1) {};
	\node[colourJ] (uvw) at (5.5,-1) {};

	\draw[-latex] (1) -- node[sloped, pos=0.1, above=-2pt] {{\tiny 1}} (12);
	\draw[-latex] (1) -- node[sloped, pos=0.2, above=-2pt] {{\tiny 2}} (31);
	\draw[-latex] (1) -- node[sloped, pos=0.1, above=-2pt] {{\tiny 1}} (1v3);
	\draw[-latex] (1) -- node[sloped, pos=0.1, below=-2pt] {{\tiny 1}} (123);

	\draw[-latex] (2) -- node[sloped, pos=0.1, below=-2pt] {{\tiny 2}} (12);
	\draw[-latex] (2) -- node[sloped, pos=0.1, below=-2pt] {{\tiny 1}} (23);
	\draw[-latex] (2) -- node[sloped, pos=0.1, below=-2pt] {{\tiny 1}} (2w);
	\draw[-latex] (2) -- node[sloped, pos=0.1, above=-2pt] {{\tiny 2}} (u2w);
	\draw[-latex] (2) -- node[sloped, pos=0.1, above=-2pt] {{\tiny 2}} (123);

	\draw[-latex] (3) -- node[sloped, pos=0.1, below=-2pt] {{\tiny 2}} (23);
	\draw[-latex] (3) -- node[sloped, pos=0.1, below=-2pt] {{\tiny 1}} (31);
	\draw[-latex] (3) -- node[sloped, pos=0.1, above=-2pt] {{\tiny 2}} (v3);
	\draw[-latex] (3) -- node[sloped, pos=0.1, below=-2pt] {{\tiny 3}} (1v3);
	\draw[-latex] (3) -- node[sloped, pos=0.1, above=-2pt] {{\tiny 3}} (123);

	\draw[-latex] (u) -- node[sloped, pos=0.1, above=-2pt] {{\tiny 1}} (uv);
	\draw[-latex] (u) -- node[sloped, pos=0.2, above=-2pt] {{\tiny 2}} (wu);
	\draw[-latex] (u) -- node[sloped, pos=0.1, above=-2pt] {{\tiny 1}} (u2w);
	\draw[-latex] (u) -- node[sloped, pos=0.1, below=-2pt] {{\tiny 1}} (uvw);

	\draw[-latex] (v) -- node[sloped, pos=0.1, below=-2pt] {{\tiny 2}} (uv);
	\draw[-latex] (v) -- node[sloped, pos=0.1, below=-2pt] {{\tiny 1}} (vw);
	\draw[-latex] (v) -- node[sloped, pos=0.1, below=-2pt] {{\tiny 1}} (v3);
	\draw[-latex] (v) -- node[sloped, pos=0.1, above=-2pt] {{\tiny 2}} (1v3);
	\draw[-latex] (v) -- node[sloped, pos=0.1, above=-2pt] {{\tiny 2}} (uvw);

	\draw[-latex] (w) -- node[sloped, pos=0.1, below=-2pt] {{\tiny 2}} (vw);
	\draw[-latex] (w) -- node[sloped, pos=0.1, below=-2pt] {{\tiny 1}} (wu);
	\draw[-latex] (w) -- node[sloped, pos=0.1, above=-2pt] {{\tiny 2}} (2w);
	\draw[-latex] (w) -- node[sloped, pos=0.1, below=-2pt] {{\tiny 3}} (u2w);
	\draw[-latex] (w) -- node[sloped, pos=0.1, above=-2pt] {{\tiny 3}} (uvw);

\end{tikzpicture} \end{center}
\medskip

\noindent The third round yields a coloring equivalent to the following:
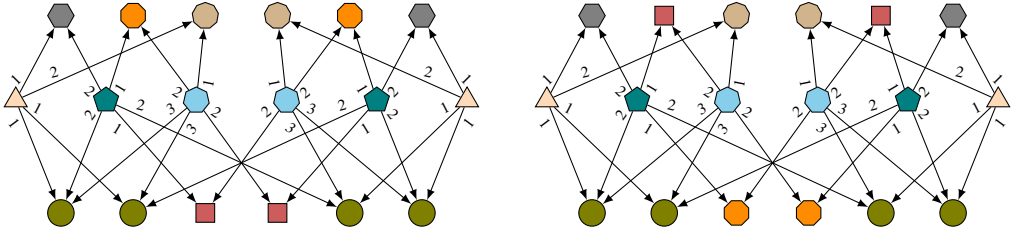
\begin{center}
\begin{tikzpicture}[
	bn/.style={circle, fill=black},
	re/.style={circle, fill=myred},
	be/.style={circle, fill=myblue},
	fe/.style={circle, fill=mygreen},
	xscale=0.8,
	yscale=1.5]
	\node[colourO] (1) at (-1.25,0) {};
	\node[colourP] (2) at (0.25,0) {};
	\node[colourQ] (3) at (1.75,0) {};
	\node[colourQ] (w) at (3.25,0) {};
	\node[colourP] (v) at (4.75,0) {};
	\node[colourO] (u) at (6.25,0) {};

	\node[colourR] (12) at (-0.5,0.75) {};
	\node[colourS] (23) at (0.7,0.75) {};
	\node[colourT] (31) at (1.9,0.75) {};
	\node[colourV] (2w) at (1.9,-1) {};
	\node[colourT] (wu) at (3.1,.75) {};
	\node[colourS] (vw) at (4.3,.75) {};
	\node[colourR] (uv) at (5.5,.75) {};
	\node[colourV] (v3) at (3.1,-1) {};
	\node[colourU] (123) at (-0.5,-1) {};
	\node[colourU] (1v3) at (0.7,-1) {};
	\node[colourU] (u2w) at (4.3,-1) {};
	\node[colourU] (uvw) at (5.5,-1) {};

	\draw[-latex] (1) -- node[sloped, pos=0.1, above=-2pt] {{\tiny 1}} (12);
	\draw[-latex] (1) -- node[sloped, pos=0.2, above=-2pt] {{\tiny 2}} (31);
	\draw[-latex] (1) -- node[sloped, pos=0.1, above=-2pt] {{\tiny 1}} (1v3);
	\draw[-latex] (1) -- node[sloped, pos=0.1, below=-2pt] {{\tiny 1}} (123);

	\draw[-latex] (2) -- node[sloped, pos=0.1, below=-2pt] {{\tiny 2}} (12);
	\draw[-latex] (2) -- node[sloped, pos=0.1, below=-2pt] {{\tiny 1}} (23);
	\draw[-latex] (2) -- node[sloped, pos=0.1, below=-2pt] {{\tiny 1}} (2w);
	\draw[-latex] (2) -- node[sloped, pos=0.1, above=-2pt] {{\tiny 2}} (u2w);
	\draw[-latex] (2) -- node[sloped, pos=0.1, above=-2pt] {{\tiny 2}} (123);

	\draw[-latex] (3) -- node[sloped, pos=0.1, below=-2pt] {{\tiny 2}} (23);
	\draw[-latex] (3) -- node[sloped, pos=0.1, below=-2pt] {{\tiny 1}} (31);
	\draw[-latex] (3) -- node[sloped, pos=0.1, above=-2pt] {{\tiny 2}} (v3);
	\draw[-latex] (3) -- node[sloped, pos=0.1, below=-2pt] {{\tiny 3}} (1v3);
	\draw[-latex] (3) -- node[sloped, pos=0.1, above=-2pt] {{\tiny 3}} (123);

	\draw[-latex] (u) -- node[sloped, pos=0.1, above=-2pt] {{\tiny 1}} (uv);
	\draw[-latex] (u) -- node[sloped, pos=0.2, above=-2pt] {{\tiny 2}} (wu);
	\draw[-latex] (u) -- node[sloped, pos=0.1, above=-2pt] {{\tiny 1}} (u2w);
	\draw[-latex] (u) -- node[sloped, pos=0.1, below=-2pt] {{\tiny 1}} (uvw);

	\draw[-latex] (v) -- node[sloped, pos=0.1, below=-2pt] {{\tiny 2}} (uv);
	\draw[-latex] (v) -- node[sloped, pos=0.1, below=-2pt] {{\tiny 1}} (vw);
	\draw[-latex] (v) -- node[sloped, pos=0.1, below=-2pt] {{\tiny 1}} (v3);
	\draw[-latex] (v) -- node[sloped, pos=0.1, above=-2pt] {{\tiny 2}} (1v3);
	\draw[-latex] (v) -- node[sloped, pos=0.1, above=-2pt] {{\tiny 2}} (uvw);

	\draw[-latex] (w) -- node[sloped, pos=0.1, below=-2pt] {{\tiny 2}} (vw);
	\draw[-latex] (w) -- node[sloped, pos=0.1, below=-2pt] {{\tiny 1}} (wu);
	\draw[-latex] (w) -- node[sloped, pos=0.1, above=-2pt] {{\tiny 2}} (2w);
	\draw[-latex] (w) -- node[sloped, pos=0.1, below=-2pt] {{\tiny 3}} (u2w);
	\draw[-latex] (w) -- node[sloped, pos=0.1, above=-2pt] {{\tiny 3}} (uvw);

\end{tikzpicture} 	\quad
\begin{tikzpicture}[
	bn/.style={circle, fill=black},
	re/.style={circle, fill=myred},
	be/.style={circle, fill=myblue},
	fe/.style={circle, fill=mygreen},
	xscale=0.8,
	yscale=1.5]
	\node[colourO] (1) at (-1.25,0) {};
	\node[colourP] (2) at (0.25,0) {};
	\node[colourQ] (3) at (1.75,0) {};
	\node[colourQ] (w) at (3.25,0) {};
	\node[colourP] (v) at (4.75,0) {};
	\node[colourO] (u) at (6.25,0) {};

	\node[colourR] (12) at (-0.5,0.75) {};
	\node[colourV] (23) at (0.7,0.75) {};
	\node[colourT] (31) at (1.9,0.75) {};
	\node[colourS] (2w) at (1.9,-1) {};
	\node[colourT] (wu) at (3.1,.75) {};
	\node[colourV] (vw) at (4.3,.75) {};
	\node[colourR] (uv) at (5.5,.75) {};
	\node[colourS] (v3) at (3.1,-1) {};
	\node[colourU] (123) at (-0.5,-1) {};
	\node[colourU] (1v3) at (0.7,-1) {};
	\node[colourU] (u2w) at (4.3,-1) {};
	\node[colourU] (uvw) at (5.5,-1) {};

	\draw[-latex] (1) -- node[sloped, pos=0.1, above=-2pt] {{\tiny 1}} (12);
	\draw[-latex] (1) -- node[sloped, pos=0.2, above=-2pt] {{\tiny 2}} (31);
	\draw[-latex] (1) -- node[sloped, pos=0.1, above=-2pt] {{\tiny 1}} (1v3);
	\draw[-latex] (1) -- node[sloped, pos=0.1, below=-2pt] {{\tiny 1}} (123);

	\draw[-latex] (2) -- node[sloped, pos=0.1, below=-2pt] {{\tiny 2}} (12);
	\draw[-latex] (2) -- node[sloped, pos=0.1, below=-2pt] {{\tiny 1}} (23);
	\draw[-latex] (2) -- node[sloped, pos=0.1, below=-2pt] {{\tiny 1}} (2w);
	\draw[-latex] (2) -- node[sloped, pos=0.1, above=-2pt] {{\tiny 2}} (u2w);
	\draw[-latex] (2) -- node[sloped, pos=0.1, above=-2pt] {{\tiny 2}} (123);

	\draw[-latex] (3) -- node[sloped, pos=0.1, below=-2pt] {{\tiny 2}} (23);
	\draw[-latex] (3) -- node[sloped, pos=0.1, below=-2pt] {{\tiny 1}} (31);
	\draw[-latex] (3) -- node[sloped, pos=0.1, above=-2pt] {{\tiny 2}} (v3);
	\draw[-latex] (3) -- node[sloped, pos=0.1, below=-2pt] {{\tiny 3}} (1v3);
	\draw[-latex] (3) -- node[sloped, pos=0.1, above=-2pt] {{\tiny 3}} (123);

	\draw[-latex] (u) -- node[sloped, pos=0.1, above=-2pt] {{\tiny 1}} (uv);
	\draw[-latex] (u) -- node[sloped, pos=0.2, above=-2pt] {{\tiny 2}} (wu);
	\draw[-latex] (u) -- node[sloped, pos=0.1, above=-2pt] {{\tiny 1}} (u2w);
	\draw[-latex] (u) -- node[sloped, pos=0.1, below=-2pt] {{\tiny 1}} (uvw);

	\draw[-latex] (v) -- node[sloped, pos=0.1, below=-2pt] {{\tiny 2}} (uv);
	\draw[-latex] (v) -- node[sloped, pos=0.1, below=-2pt] {{\tiny 1}} (vw);
	\draw[-latex] (v) -- node[sloped, pos=0.1, below=-2pt] {{\tiny 1}} (v3);
	\draw[-latex] (v) -- node[sloped, pos=0.1, above=-2pt] {{\tiny 2}} (1v3);
	\draw[-latex] (v) -- node[sloped, pos=0.1, above=-2pt] {{\tiny 2}} (uvw);

	\draw[-latex] (w) -- node[sloped, pos=0.1, below=-2pt] {{\tiny 2}} (vw);
	\draw[-latex] (w) -- node[sloped, pos=0.1, below=-2pt] {{\tiny 1}} (wu);
	\draw[-latex] (w) -- node[sloped, pos=0.1, above=-2pt] {{\tiny 2}} (2w);
	\draw[-latex] (w) -- node[sloped, pos=0.1, below=-2pt] {{\tiny 3}} (u2w);
	\draw[-latex] (w) -- node[sloped, pos=0.1, above=-2pt] {{\tiny 3}} (uvw);

\end{tikzpicture} \end{center}
\medskip

\noindent Since the coloring is stable, Color Refinement stops.
It does not distinguish $\enI{\AdExa}$, $\enI{\AdExb}$. 
\subsection{RCR on \texorpdfstring{$\A$}{A} is equivalent to CR on \texorpdfstring{$\grep{\A}$}{G\_A}}\label{app:rcr-is-cr-on-rep}
It is easy to see, that CR on $\grep{\A}$ produces a coloring on $\V(\grep{\A})$ that is equivalent to the one produced by RCR on $\A$ in the sense that for all $\at, \bt \in \tA$, $\col{i}{\at} = \col{i}{\bt}$ iff $\gamma_i(w_{\at}) = \gamma_i(w_{\bt})$.
This can be shown via a simple induction.
\medskip

\myparagraph{Base Case.} Let $\at, \bt \in \tA$.
Then $\col{0}{\at} = \col{0}{\bt}$ iff $\atp(\at) = \atp(\bt)$ and $\stp(\at) = \stp(\bt)$.
This holds iff for all $R \in \sig$ we have $\at \in R^\A \iff \bt \in R^\A$ and for all $i,j \in [\ar(\sig)]$ we have $(i,j) \in \stp(\at) \iff (i,j) \in \stp(\bt)$.
By definition of $\grep{\A}$, this holds iff for all $R \in \sig$ we have $w_{\at} \in {(U_R)}^{\grep{\A}} \iff w_{\bt} \in {(U_R)}^{\grep{\A}}$ and for all $i,j \in [\ar(\sig)]$ we have $(w_{\at}, w_{\at}) \in {(E_{i,j})}^{\grep{\A}} \iff (w_{\bt}, w_{\bt}) \in {(E_{i,j})}^{\grep{\A}}$.
Since $\gamma_0(w) = (\set{ U_R \mid w \in {(U_R)}^{\grep{\A}} }, \set{ E_{i,j} \mid (w,w) \in {(E_{i,j})}^{\grep{\A}} })$, it follows that $\col{0}{\at} = \col{0}{\bt}$ iff $\gamma_0(w_{\at}) = \gamma_0(w_{\bt})$.
\medskip

\myparagraph{Inductive Step.}
Consider $i \in \nat$.
If $\col{i}{\at} \neq \col{i}{\bt}$, then $\gamma_i(w_{\at}) \neq \gamma_i(w_{\bt})$ by induction hypothesis, and $\col{i+1}{\at} \neq \col{i+1}{\bt}$ and $\gamma_{i+1}(w_{\at}) \neq \gamma_{i+1}(w_{\bt})$ by definition.
Thus, consider the case that $\col{i}{\at} = \col{i}{\bt}$.
By the induction hypothesis we have $\gamma_i(\at) = \gamma_i(\bt)$.
Then, $\col{i+1}{\at} = \col{i+1}{\bt}$ iff $\Nset{i+1}{\A}(\at) = \Nset{i+1}{\A}(\bt)$, i.e., iff
\begin{align*}
	\mset[\big]{ \bigl( \stp(\at,\ct),\, \col{i}{\ct} \bigr) \mid \ct \in \tA,\, \stp(\at, \ct) \neq \emptyset }
	\;\;=\;\;
	\mset[\big]{ \bigl( \stp(\bt,\ct),\, \col{i}{\ct} \bigr) \mid \ct \in \tA,\, \stp(\bt, \ct) \neq \emptyset }\; .
\end{align*}
Since it holds that $(i,j) \in \stp(\at, \ct) \iff (j,i) \in \stp(\ct, \at)$, the above holds iff
\begin{align*}
	&\mset[\big]{ \bigl( \stp(\at,\ct),\stp(\ct, \at),\, \col{i}{\ct} \bigr) \mid \ct \in \tA,\, \stp(\at, \ct) \neq \emptyset }\\
	=\;\;
	&\mset[\big]{ \bigl( \stp(\bt,\ct),\stp(\ct, \bt),\, \col{i}{\ct} \bigr) \mid \ct \in \tA,\, \stp(\bt, \ct) \neq \emptyset }\; .
\end{align*}
For every $\ct \in \tA$ it holds that $(w_{\at}, w_{\ct}) \in E_{i,j}^{\grep{\A}}$ iff $(i,j) \in \stp(\at, \ct)$; and analogously, $(w_{\bt}, w_{\ct}) \in E_{i,j}^{\grep{\A}}$ iff $(i,j) \in \stp(\bt, \ct)$.
Thus, the above holds iff the following holds
\begin{align*}
	&\mset[\big]{ \bigl( \set{ E_{i,j} \mid (w_{\at}, w_{\ct}) \in E_{i,j}^{\grep{\A}} },\set{ E_{i,j} \mid (w_{\ct}, w_{\at}) \in E_{i,j}^{\grep{\A}}},\, \col{i}{\ct} \bigr) \mid \ct \in \tA,\, \stp(\at, \ct) \neq \emptyset }\\
	=\;\;
	&\mset[\big]{ \bigl( \set{ E_{i,j} \mid (w_{\bt}, w_{\ct}) \in E_{i,j}^{\grep{\A}} },\set{ E_{i,j} \mid (w_{\ct}, w_{\bt}) \in E_{i,j}^{\grep{\A}}},\, \col{i}{\ct} \bigr) \mid \ct \in \tA,\, \stp(\bt, \ct) \neq \emptyset }\; .
\end{align*}
And since this is equivalent to $\lambda(w_{\at}, w_{\ct})$ and $\lambda(w_{\bt}, w_{\ct})$, this holds iff
\begin{align*}
	\!\mset[\big]{ \bigl( \lambda(w_{\at}, w_{\ct}),\, \col{i}{\ct} \bigr) \mid \ct \in \tA,\, \stp(\at, \ct) \neq \emptyset }
	\ =\
	\mset[\big]{ \bigl( \lambda(w_{\bt}, w_{\ct}),\, \col{i}{\ct} \bigr) \mid \ct \in \tA,\, \stp(\bt, \ct) \neq \emptyset }.
\end{align*}
Plugging in the fact that $\set{ w_{\at}, w_{\ct} } \in \E(G)$ iff $\stp(\at, \ct) \neq \emptyset$ and $\at\neq\ct$, where $G$ is the Gaifman graph of $\grep{\A}$, and the induction hypothesis, this is the case iff
\begin{align*}
	\mset[\big]{ \bigl( \lambda(w_{\at}, w_{\ct}),\, \gamma_{i}(\ct) \bigr) \mid \set{ \at, \ct } \in \E(G) }
	\;=\;
	\mset[\big]{ \bigl( \lambda(w_{\bt}, w_{\ct}),\, \gamma_{i}(\ct) \bigr) \mid \set{ \bt, \ct } \in \E(G) }\; .
\end{align*}
Thus, $\col{i+1}{\at} = \col{i+1}{\bt}$ iff $\gamma_{i+1}(w_{\at}) = \gamma_{i+1}(w_{\bt})$.\qed% 
\subsection{Relational Color Refinement Generalizes Color Refinement}\label{app:rel-cr-generalizes-cr}
By induction, we show that (a) for all $u,v \in \V(G)$ we have: $\col{i}{u} = \col{i}{v}$ iff $\gamma_i(u) = \gamma_i(v)$, where $\gamma_i$ is defined as in the introduction, i.e., $\gamma_0(u) = 0$ for all $u \in \V(G)$ and $\gamma_{i+1}(u) = (\gamma_i(u), \mset{ \gamma_i(w) \mid \set{ u,w } \in \E(G) })$.
Additionally, we show that (b) for all $\set{ u,v }[], \set{ x,y }[] \in \E(G)$ it holds that $\col{i}{u,v} = \col{i}{x,y}$ iff $(\gamma_i(u), \gamma_i(v)) = (\gamma_i(x), \gamma_i(y))$.

\myparagraph{Base Case.} Initially, $\gamma_0(u) = 0$ and $\col{0}{u} = (\set{ U }, \set{ (1,1) })$ for every $u \in \V(G)$ and $\col{0}{u,v} = (\set{ E }, \set{ (1,1), (2,2) })$ for every $\set{ u,v } \in \E(G)$.
Thus, our claim holds in the base case.
\medskip

\myparagraph{Induction Hypothesis.} 
Let $i \in \nat$.
We assume that $\gamma_i(u) = \gamma_i(v)$ iff $\col{i}{u} = \col{i}{v}$ for all $u,v \in \V(G)$ and $\col{i}{u,v} = \col{i}{x,y}$ iff $(\gamma_i(u), \gamma_i(v)) = (\gamma_i(x), \gamma_i(y))$ for all $\set{ u,v }, \set{ x,y } \in \E(G)$.
\medskip

\myparagraph{Inductive Step (a).}
Let $u,v \in \V(G)$.
If $\gamma_i(u) \neq \gamma_i(v)$, then $\gamma_{i+1}(u) \neq \gamma_{i+1}(v)$.
By induction hypothesis, $\col{i}{u} \neq \col{i}{v}$ and hence, by definition, $\col{i+1}{u} \neq \col{i+1}{v}$.

If $\gamma_i(u) = \gamma_i(v)$, then by induction hypothesis, $\col{i}{u} = \col{i}{v}$.
Hence, it suffices to show that $\Nset{i+1}{\A_G}(u) = \Nset{i+1}{\A_G}(v)$ iff $\mset{ \gamma_i(w) \mid \set{ u,w } \in \E(G) } = \mset{ \gamma_i(w) \mid \set{ v,w } \in \E(G) }$.
It is easy to see that 
\begin{align*}
	\Nset{i+1}{\A_G}(u) &= \mset[\big]{ \bigl( \set{ (1,1) }, \col{i}{u,w} \bigr), \bigl( \set{ (1,2) }, \col{i}{w,u} \bigr) \mid \set{ u,w } \in \E(G) } \union \mset{ (\set{ (1,1) }, \col{i}{u}) }\\
	\Nset{i+1}{\A_G}(v) &= \mset[\big]{ \bigl( \set{ (1,1) }, \col{i}{v,w} \bigr), \bigl( \set{ (1,2) }, \col{i}{w,v} \bigr) \mid \set{ v,w } \in \E(G) } \union \mset{ (\set{ (1,1) }, \col{i}{v}) }
\end{align*}
Hence, $\Nset{i+1}{\A_G}(u) = \Nset{i+1}{\A_G}(v)$ if, and only if, $\mset{ \col{i}{u,w} \mid \set{ u,w }[] \in \E(G) } = \mset{ \col{i}{v,w} \mid \set{ v,w }[] \in \E(G) }$ and $\mset{ \col{i}{w,u} \mid \set{ u,w }[] \in \E(G) } = \mset{ \col{i}{w,v} \mid \set{ v,w }[] \in \E(G) }$.
By induction hypothesis, this is the case if, and only if, $\mset{ (\gamma_i(u), \gamma_i(w)) \mid \set{ u,w } \in \E(G) } = \mset{ (\gamma_i(v), \gamma_i(w)) \mid \set{ v,w } \in \E(G) }$ and $\mset{ (\gamma_i(w), \gamma_i(u)) \mid \set{ u,w } \in \E(G) } = \mset{ (\gamma_i(w), \gamma_i(v)) \mid \set{ v,w } \in \E(G) }$; and because $\gamma_i(u) = \gamma_i(v)$, this is the case iff $\mset{ \gamma_i(w) \mid \set{ u,w } \in \E(G) } = \mset{ \gamma_i(w) \mid \set{ v,w } \in \E(G) }$.
Thus, $\gamma_{i+1}(u) = \gamma_{i+1}(v)$ iff $\col{i+1}{u} = \col{i+1}{v}$.
\medskip

\myparagraph{Inductive Step (b).}
Let $\set{ u,v }, \set{ x,y } \in \E(G)$.
If $(\gamma_i(u), \gamma_i(v)) \neq (\gamma_i(x), \gamma_i(y))$, then $(\gamma_{i+1}(u), \gamma_{i+1}(v)) \neq (\gamma_{i+1}(x), \gamma_{i+1}(y))$.
By induction hypothesis, we also obtain that $\col{i}{u,v} \neq \col{i}{x,y}$ and hence, $\col{i+1}{u,v} \neq \col{i+1}{v,w}$.

If $(\gamma_i(u), \gamma_i(v)) = (\gamma_i(x), \gamma_i(y))$ then, by the induction hypothesis, $\col{i}{u,v} = \col{i}{x,y}$.
We show that $\Nset{i+1}{\A_G}(u,v) = \Nset{i+1}{\A_G}(x,y)$ if, and only if, $\mset{ \gamma_i(w) \mid \set{ u,w } \in \E(G) } = \mset{ \gamma_i(w) \mid \set{ v,w } \in \E(G) }$ and $\mset{ \gamma_i(w) \mid \set{ x,w } \in \E(G) } = \mset{ \gamma_i(w) \mid \set{ y,w } \in \E(G) }$.
Note that
\begin{align*}
	\Nset{i+1}{\A_G}(u,v)\;\; &= \;\;
		\mset[\big]{
			\bigl( \set{ (1,1) }, \col{i}{u,w} \bigr), \bigl( \set{ (1,2) }, \col{i}{w,u} \bigr) \mid \set{ u,w } \in \E(G)
		}\\
		&\union\;\;
		\mset[\big]{
			\bigl( \set{ (2,1) }, \col{i}{v,w} \bigr), \bigl( \set{ (2,2) }, \col{i}{w,v} \bigr) \mid \set{ v,w } \in \E(G)
		}\\
		&\union\;\;
		\bigl\{\!\!\bigl\{\begin{aligned}[t]
			&\bigl( \set{ (1,1), (2,2) }, \col{i}{u,v} \bigr), 
			\bigl( \set{ (1,2), (2,1) }, \col{i}{v,u} \bigr),\\
			&\bigl( \set{ (1,1) }, \col{i}{u} \bigr), 
			\bigl( \set{ (2,1) }, \col{i}{v} \bigr)
		\bigr\}\!\!\bigr\},\quad\text{and}
		\end{aligned}\\
	\Nset{i+1}{\A_G}(x,y)\;\; &= \;\;
		\mset[\big]{
			\bigl( \set{ (1,1) }, \col{i}{x,w} \bigr), \bigl( \set{ (1,2) }, \col{i}{w,x} \bigr) \mid \set{ x,w } \in \E(G)
		}\\
		&\union\;\;
		\mset[\big]{
			\bigl( \set{ (2,1) }, \col{i}{y,w} \bigr), \bigl( \set{ (2,2) }, \col{i}{w,y} \bigr) \mid \set{ y,w } \in \E(G)
		}\\
		&\union\;\;
		\bigl\{\!\!\bigl\{\begin{aligned}[t]
			&\bigl( \set{ (1,1), (2,2) }, \col{i}{x,y} \bigr),
			\bigl( \set{ (1,2), (2,1) }, \col{i}{y,x} \bigr),\\
			&\bigl( \set{ (1,1) }, \col{i}{x} \bigr),
			\bigl( \set{ (2,1) }, \col{i}{y} \bigr)
		\bigr\}\!\!\bigr\}.
		\end{aligned}
\end{align*}
Since $\col{i}{u} = \col{i}{x}$, $\col{i}{v} = \col{i}{y}$, $\col{i}{u,v} = \col{i}{x,y}$ and $\col{i}{v,u} = \col{i}{y,x}$ by induction hypothesis, it follows that $\Nset{i+1}{\A_G}(u,v) = \Nset{i+1}{\A_G}(x,y)$ iff
\begin{align*}
	&\phantom{=}\;\;\;\ 
	\mset[\big]{
		\bigl( \set{ (1,1) }, \col{i}{u,w} \bigr), \bigl( \set{ (1,2) }, \col{i}{w,u} \bigr) \mid \set{ u,w } \in \E(G)
	}\\
	&=\;\;
	\mset[\big]{
		\bigl( \set{ (1,1) }, \col{i}{x,w} \bigr), \bigl( \set{ (1,2) }, \col{i}{w,x} \bigr) \mid \set{ x,w } \in \E(G)
	}\quad\text{and}\\
	&\phantom{=}\;\;\;\ 
	\mset[\big]{
		\bigl( \set{ (2,1) }, \col{i}{v,w} \bigr), \bigl( \set{ (2,2) }, \col{i}{w,v} \bigr) \mid \set{ v,w } \in \E(G)
	}\\
	&=\;\;
	\mset[\big]{
		\bigl( \set{ (2,1) }, \col{i}{y,w} \bigr), \bigl( \set{ (2,2) }, \col{i}{w,y} \bigr) \mid \set{ y,w } \in \E(G)
	}.
\end{align*}
And this is the case iff
\begin{align*}
	&\mset[\big]{ \col{i}{u,w} \mid \set{ u,w } \in \E(G) }
	\; = \;
	\mset[\big]{ \col{i}{x,w} \mid \set{ x,w } \in \E(G) },\\
	&\mset[\big]{ \col{i}{w,u} \mid \set{ u,w } \in \E(G) }
	\; = \;
	\mset[\big]{ \col{i}{w,x} \mid \set{ x,w } \in \E(G) },\\
	&\mset[\big]{ \col{i}{v,w} \mid \set{ v,w } \in \E(G) }
	\; = \;
	\mset[\big]{ \col{i}{y,w} \mid \set{ y,w } \in \E(G) },\quad\text{and}\\
	&\mset[\big]{ \col{i}{w,v} \mid \set{ v,w } \in \E(G) }
	\; = \;
	\mset[\big]{ \col{i}{w,y} \mid \set{ y,w } \in \E(G) }.
\end{align*}
By induction hypothesis, this is the case iff
\begin{align*}
	&\mset[\big]{ (\gamma_i(u), \gamma_i(w)) \mid \set{ u,w } \in \E(G) }
	\; = \;
	\mset[\big]{ (\gamma_i(x), \gamma_i(w)) \mid \set{ x,w } \in \E(G) },\\
	&\mset[\big]{ (\gamma_i(w), \gamma_i(u)) \mid \set{ u,w } \in \E(G) }
	\; = \;
	\mset[\big]{ (\gamma_i(w), \gamma_i(x)) \mid \set{ x,w } \in \E(G) },\\
	&\mset[\big]{ (\gamma_i(v), \gamma_i(w)) \mid \set{ v,w } \in \E(G) }
	\; = \;
	\mset[\big]{ (\gamma_i(y), \gamma_i(w)) \mid \set{ y,w } \in \E(G) },\quad\text{and}\\
	&\mset[\big]{ (\gamma_i(w), \gamma_i(v)) \mid \set{ v,w } \in \E(G) }
	\; = \;
	\mset[\big]{ (\gamma_i(w), \gamma_i(y)) \mid \set{ y,w } \in \E(G) }.
\end{align*}
Since $\gamma_i(u) = \gamma_i(x)$ and $\gamma_i(v) = \gamma_i(y)$, this holds iff
\begin{align*}
	&\mset[\big]{ \gamma_i(w) \mid \set{ u,w } \in \E(G) }
	\; = \;
	\mset[\big]{ \gamma_i(w) \mid \set{ x,w } \in \E(G) },\quad\text{and}\\
	&\mset[\big]{ \gamma_i(w) \mid \set{ v,w } \in \E(G) }
	\; = \;
	\mset[\big]{ \gamma_i(w) \mid \set{ y,w } \in \E(G) }.
\end{align*}
This is the case iff $\gamma_{i+1}(u) = \gamma_{i+1}(x)$ and $\gamma_{i+1}(v) = \gamma_{i+1}(y)$.
Thus, in total we get $\col{i+1}{u,v} = \col{i+1}{v,w}$ iff $(\gamma_{i+1}(u), \gamma_{i+1}(v)) = (\gamma_{i+1}(x), \gamma_{i+1}(y))$.\qed% 

\clearpage
\section{Details omitted in the proof of
 Lemma~\ref{lem:homs-tree-to-acyclic-structure}}\label{sec:appendix:homomorphisms:SecondLemma}
\begin{claim}\label{claim:homomorphisms:SecondLemma:FirstClaim}
	$\hom(\str{T}, \grep{\A}) = \sum_{P \in P_\A} \#(P, \A)$ \ \ and \ \ $\hom(\str{T}, \grep{\B}) = \sum_{P \in P_\B} \#(P, \B)$.
\end{claim}
\begin{proof}[Proof of \cref{claim:homomorphisms:SecondLemma:FirstClaim}]
We only prove the first equality; the second equality follows analogously.
By definition,
\[
	\Hom(\str{T}, \grep{\A})
	\ \  = \ \
	\bigcup_{P \in P_{\A}} \setc{h \in \Hom(\str{T}, \grep{\A})}{P_h=P}\;.
\]
Clearly, this is a union of pairwise disjoint sets.
Thus,
\[
	\hom(\str{T}, \grep{\A})
	\ \  = \ \
	\card{\Hom(\str{T}, \grep{\A})}
	\ \ = \ \
	\sum_{P \in P_{\A}} \card{\setc{h \in \Hom(\str{T}, \grep{\A})}{P_h=P}}
	\ \ = \ \
	\sum_{P \in P_{\A}}\#(P, \A).
\]
This completes the proof of \cref{claim:homomorphisms:SecondLemma:FirstClaim}.
\end{proof}
\medskip

\begin{claim}\label{claim:homomorphisms:SecondLemma:SecondClaim}
	$\hom(P, \grep{\A}) = \sum_{P': P \preceq P'} \#(P', \A)$ \ \ and \ \ $\hom(P, \grep{\B}) = \sum_{P':P \preceq P'} \#(P', \B)$.
\end{claim}  
\begin{proof}[Proof of \cref{claim:homomorphisms:SecondLemma:SecondClaim}]
We only prove the first equality; the second equality follows analogously.
Let
\[
	X \ \ \deff \ \ \Hom(P, \grep{A}).
	\qquad\quad
	\text{Hence, }\quad |X| \ \ = \ \ \hom(P, \grep{A}).
\]
For every $P'$ with $P \preceq P'$ let
\[
	Y_{P'} \ \ \deff \ \ \setc{h \in \Hom(\str{T}, \grep{A})}{P_h=P'}.
	\qquad\quad
	\text{Hence, }\quad |Y_{P'}| \ \ = \ \ \#(P', \A).
\]
In the following, we will show that $X = \bigcup_{P': P \preceq P'}
Y_{P'}$.
Since this is a union of pairwise disjoint sets,
it implies that $\card{X} = \sum_{P' : P \preceq P'}|Y_{P'}|$,
i.e., $\hom(P, \grep{A}) = \sum_{P' : P \preceq P'}\#(P', \A)$,
i.e., it completes the proof of the first statement of \cref{claim:homomorphisms:SecondLemma:SecondClaim}.
\\
We will prove the following two claims.
\bigskip

\noindent
\emph{Claim (i):} For all $h \in \Hom(\str{T}, \grep{\A})$ with $P \preceq P_h$ we have: \ $h \in \Hom(P, \grep{\A})$.
\smallskip

\noindent
\emph{Claim (ii):} For all $f \in \Hom(P, \grep{\A})$ we have: \ $f \in \Hom(\str{T}, \grep{\A})$ \ and \ $P \preceq P_f$.

\bigskip

\noindent
Note that Claim~(i) implies that
\[
  \bigcup_{P' : P \preceq P'}Y_{P'}
  \ \ \subseteq \ \ 
  X
\]
--- to see this, recall that $P' \in P_{\A}$, and hence $P'=P_h$ for some $h \in \Hom(\str{T}, \grep{\A})$.
\smallskip

\noindent
Furthermore, note that
\[
	\bigcup_{P' : P \preceq P'}Y_{P'}
	\ \ = \ \
	\bigcup_{P' : P \preceq P'}   \setc{h \in \Hom(\str{T}, \grep{A})}{P_h=P'}
	\ \ = \ \
	\setc{h \in \Hom(\str{T}, \grep{A})}{P \preceq P_h}.
\]
Moreover, Claim~(ii) implies that
\[
	\setc{h \in \Hom(\str{T}, \grep{A})}{P \preceq P_h}
	\ \ \supseteq \ \
	\Hom(P, \grep{\A}).
\]
Hence, we have:
\[
	\bigcup_{P' : P \preceq P'}Y_{P'}
	\ \ \supseteq \ \
	X.
\]
In summary, we obtain: $X = \bigcup_{P' : P \preceq P'}Y_{P'}$.
Therefore, all that remains to be done in order to complete the proof of the first statement of \cref{claim:homomorphisms:SecondLemma:SecondClaim} is to prove Claim~(i) and Claim~(ii).

\begin{proof}[Proof of Claim~(i):]
Consider an arbitrary $h \in \Hom(\str{T}, \grep{\A})$ with $P \preceq P_h$.
We have to show that  $h \in \Hom(P, \grep{\A})$.

Since $h \in \Hom(\str{T}, \grep{\A})$, it is a mapping $h \colon V(\str{T})\to V(\grep{\A})$.
Since $P \in P_{\A}$, by definition we have: $V(P)=V(T) = V(\str{T})$;
thus, $h$ also is a mapping $h \colon V(P)\to V(\grep{\A})$.
Let us show that this mapping is a homomorphism from $P$ to $\grep{\A}$.

First, consider an arbitrary $R \in \sig$ and an arbitrary $v \in {(U_R)}^{P}$.
We have to show that $h(v)\in {(U_R)}^{\grep{A}}$.
We know that $P \preceq P_h$.
Thus, $v \in {(U_R)}^{P}$ implies that $v \in {(U_R)}^{P_h}$.
From the definition of $P_h$ we obtain: $h(v)\in {(U_R)}^{\grep{\A}}$.

Next, consider arbitrary $i,j \in [\ar(\sig)]$ and an arbitrary tuple $(u,v) \in {(E_{i,j})}^{P}$.
We have to show that $(h(u),h(v))\in {(E_{i,j})}^{\grep{\A}}$.
We know that $P \preceq P_h$.
Thus, $(u,v)\in {(E_{i,j})}^{P}$ implies that $(u,v)\in {(E_{i,j})}^{P_h}$.
From the definition of $P_h$ we obtain: $(h(u),h(v))\in {(E_{i,j})}^{\grep{\A}}$.

In summary, we obtain that $h$ is a homomorphism from $P$ to $\grep{\A}$.
This completes the proof of Claim~(i).
\end{proof}

\begin{proof}[Proof of Claim~(ii):]
Consider an arbitrary $f \in \Hom(P, \grep{\A})$.
We have to show that $f \in \Hom(\str{T}, \grep{\A})$ \ and \ $P \preceq P_f$.

First, recall that $P \in P_{\A}$.
Thus, $P=P_{h_0}$ for some $h_0 \in \Hom(\str{T}, \grep{\A})$.
Throughout this proof, let us fix such an $h_0$.
I.e., we have: \ $h_0 \in \Hom(\str{T}, \grep{\A})$ and $P=P_{h_0}$.

By assumption, $f \in \Hom(P, \grep{\A})$.
Thus, $f$ is a mapping $f \colon V(P)\to V(\grep{\A})$.
Since $P \in P_{\A}$, by definition we have: $V(P)=V(T) = V(\str{T})$;
thus, $f$ also is a mapping $f \colon V(\str{T})\to V(\grep{\A})$.
\smallskip

\noindent
\emph{Step~1:} Show that $f$ is a homomorphism from $\str{T}$ to $\grep{\A}$.
\smallskip

\noindent
First, consider an arbitrary $R \in \sig$ and an arbitrary $v \in {(U_R)}^{\str{T}}$.
We have to show that $f(v)\in {(U_R)}^{\grep{\A}}$.

We know that $P = P_{h_0}$.
Since $h_0$ is a homomorphism from $\str{T}$ to $\grep{\A}$, $v \in {(U_R)}^{\str{T}}$ implies that $h_0(v)\in {(U_R)}^{\grep{\A}}$.
By the definition of $P_{h_0}$, we obtain that $v \in {(U_R)}^{P_{h_0}}$, i.e., $v \in {(U_R)}^P$.
Since $f$ is a homomorphism from $P$ to $\grep{\A}$, we obtain: $f(v)\in {(U_R)}^{\grep{\A}}$.

Next, consider arbitrary $i,j \in[\ar(\sig)]$ and an arbitrary tuple $(u,v)\in {(E_{i,j})}^{\str{T}}$.
We have to show that $(f(u),f(v))\in {(E_{i,j})}^{\grep{\A}}$.
Since $h_0$ is a homomorphism from $\str{T}$ to $\grep{\A}$, $(u,v)\in {(E_{i,j})}^{\str{T}}$ implies that $(h_0(u),h_0(v))\in {(E_{i,j})}^{\grep{\A}}$.
Furthermore, from $(u,v)\in {(E_{i,j})}^{\str{T}}$, by the definition of the Gaifman graph $T$ of $\str{T}$ we obtain that either $\set{u,v}\in \E(T)$ or $u=v$.
Thus, according to the definition of $P_{h_0}$ we obtain that $(u,v)\in {(E_{i,j})}^{P_{h_0}}$, i.e., $(u,v)\in {(E_{i,j})}^P$.
Since $f$ is a homomorphism from $P$ to $\grep{\A}$, we obtain: $(f(u),f(v))\in {(E_{i,j})}^{\grep{\A}}$.

In summary, we obtain that $f$ is a homomorphism from $\str{T}$ to $\grep{\A}$.
This completes \emph{Step~1}.
\smallskip

\noindent
\emph{Step~2:} Show that $P \preceq P_f$.
\smallskip

\noindent
First, consider an arbitrary $R \in \sig$ and an arbitrary $v \in {(U_R)}^{P}$.
We have to show that $v \in {(U_R)}^{P_f}$.

According to the definition of $P_f$ it suffices to show that $f(v)\in {(U_R)}^{\grep{\A}}$.
Since $f$ is a homomorphism from $P$ to $\grep{\A}$ and $v \in {(U_R)}^{P}$, we obtain: $f(v) \in {(U_R)}^{\grep{\A}}$.

Next, consider arbitrary $i,j \in[\ar(\sig)]$ and an arbitrary tuple $(u,v)\in {(E_{i,j})}^{P}$.
We have to show that $(u,v)\in {(E_{i,j})}^{P_f}$.
According to the definition of $P_f$ it suffices to show the following:
\[
	(f(u),f(v))\in {(E_{i,j})}^{\grep{\A}} \quad \text{and} \quad
	\big( \ \ 
		u=v \ \ \text{ or } \ \ \set{u,v}\in \E(T)
	\ \ \big).
\] 
Since $f$ is a homomorphism from $P$ to $\grep{\A}$ and $(u,v)\in {(E_{i,j})}^P$, we obtain that $(f(u),f(v))\in {(E_{i,j})}^{\grep{\A}}$.
\\
In case that $u=v$, we are done.
\\
In case that $u \neq v$,  we have to show that $\set{u,v}\in \E(T)$.
The reason why $\set{u,v}\in \E(T)$ is as follows:
By assumption, we have $(u,v)\in {(E_{i,j})}^P$.
Recall that $P=P_{h_0}$.
According to the definition of $P_{h_0}$, from $(u,v)\in {(E_{i,j})}^P$ and $u \neq v$ we obtain: $\set{u,v}\in \E(T)$.

In summary, we obtain that $P \preceq P_f$.
This completes \emph{Step~2} and the proof of \emph{Claim~(ii)}.
\end{proof}  

\noindent
This also completes the proof of \cref{claim:homomorphisms:SecondLemma:SecondClaim}.
\end{proof}
\medskip 

\clearpage
\section{Details omitted  in Section~\ref{sec:main-logic}}\label{app:main-logic}

\subsection{Proof of Lemma~\ref{lem:describing-color-in-gfc}}
\describingColorInGFC*
\begin{proof}[Proof of~\cref{lem:describing-color-in-gfc}]
Let $m \isdef \ar(\sig)= \max\set{\ar(R) \mid R \in \sigma}$.

Note that, by the definition of RCR, for every $\sig$-structure $\B$, for every $i \in \NN$, and for every color $c \in \cols{i}{\B}$, there is a number $k \in[m]$, an atomic type $\myatype$ of arity $k$, and a similarity type $\mystype$ of arity $k$ such that every tuple that gets assigned color $c$ by RCR has arity $k$, atomic type $\myatype$, and similarity type $\mystype$.
We will write $\ar(c)$ to denote $k$, and $\stp(c)$ to denote $\mystype$, and $\atp(c)$ to denote $\myatype$.

The following notation allows us to substitute variables in a variable tuple and will be handy throughout the proof.

Let $X = \set{ y_1, \dots, y_m, y_1', \dots, y_m' } \subseteq \VAR$ be a set of $2m$ variables and let \,$\compl{\mbox{\,\vphantom{y}$\cdot$\,}}\colon X \to X$ with $\compl{y_j} = y_j'$ and $\compl{y_j'} = y_j$, for all $j \in [m]$.
I.e., $y_j, y_j'$ serve as complementary pairs, and \,$\compl{\mbox{\,\vphantom{y}$\cdot$\,}}$\, allows us to switch from one to the other.
In the following, w.l.o.g.\ we will only consider $\ell$-variable tuples $\xt = (x_1,\dots, x_\ell)$ where $1 \leq \ell \leq m$ and $x_j$ is either $y_j$ or $y_j'$ (for all $j \in [\ell]$).

For proving the lemma, let us consider a fixed $\sig$-structure $\A$.
We proceed by induction on the number $i$ of iterations performed with RCR, and we prove a statement slightly stronger than the lemma's statement.
For every $i \in \NN$ let $\myC_i$ be the union of the sets $\cols{i}{\B}$ for all $\sig$-structures $\B$ of size strictly equal to $\A$.
Note that $\myC_i$ is a \emph{finite}, non-empty set.
By induction on $i$ we show that for every $i \in \NN$ and every $c \in \myC_i$ of arity $k \deff \ar(c)$, there exists a formula $\phi^{i}_{c}(\xt) \in \GFC$ with $\xt = (x_1, \dots, x_k)$ such that for every $\sig$-structure $\B$ of size strictly equal to $\A$ and every $\bt \in \tB$ of arity $k$ we have: $\B, \bt \models \phi^{i}_{c}(\xt)$ $\iff$ $\col{i}{\bt} = c$.

\paragraph*{Base Cases:}
For every $k \in \natpos$, every atomic type $\myatype$ of arity $k$, every similarity type $\mystype$ of arity $k$, and every $k$-variable tuple $\xt = (x_1, \dots, x_k)$ we define the formulas
\[
	\phi_{\myatype}(\xt) \;\isdef\;
	\smashoperator{\bigland_{R \in \myatype}} R(\xt)
	\;\;\land\;\;
	\smashoperator{\bigland_{R \not\in \myatype}} \lnot R(\xt)
	\qquad\text{and}\qquad
	\phi_{\mystype}(\xt) \;\isdef\;
	\smashoperator{\bigland_{(j,j') \in \mystype}}
	\xvar_j \lequal \xvar_{j'}
	\;\;\land\;\;
	\smashoperator{\bigland_{(j,j') \not\in \mystype}}
	\lnot\, \xvar_j \lequal \xvar_{j'}
	\;.
\]
Obviously, $\phi_{\myatype}, \phi_{\mystype} \in \GFC$ and $\gd(\phi_{\myatype}) = \gd(\phi_{\mystype}) = 0$, and it is easy to see that for all $\sig$-structures $\str{B}$ and all $k$-tuples $\bt \in {\V(\str{B})}^k$ the following holds:
\begin{align}\label{char:phi-typen}%
	\str{B}, \tup{b}\ \models \ \phi_{\myatype}(\xt)
		\;\iff\;
		\atp(\tup{b}) = \myatype
		\;,
	\qquad\text{and}\qquad
	\str{B}, \tup{b}\ \models \ \phi_{\mystype}(\xt)
		\;\iff\;
		\stp(\tup{b}) = \mystype
	\;.
\end{align}
Now consider an arbitrary color $c \in \myC_0$.
By definition of RCR, $c=(\atp(c), \stp(c))$.
Thus, we are done by choosing $\phi^0_{c}(\xt) \isdef \bigl(\phi_{\atp(c)}(\xt) \land \phi_{\stp(c)}(\xt)\bigr)$.

\paragraph*{Inductive Step:}
Consider $c \in \myC_{i+1}$ of arity $k$.
Let $\xt$ be a $k$-variable tuple (as described at the beginning of this proof).
By definition of RCR, $c$ is of the form $(c', N)$, where $c' \in \myC_i$ also has arity $k$, and $N$ is a multiset consisting of tuples of the form $(\mystype, d)$ with $\mystype \in \sTypes$ and $d \in \myC_i$.
\\
By the induction hypothesis we already have available a formula $\phi^{i}_{c'}(\xt)$ with the desired properties.
\\
We let $\phi^{i+1}_{c}(\xt) \isdef \bigl(\phi^{i}_{c'}(\xt) \land \phi_{N}(\xt)\bigr)$, where
\begin{equation*}
	\phi_{N}(\xt) \;\;=\;\; \smashoperator{\bigland_{\substack{d \in \myC_i,\\ \mystype \in \sTypes}}}\;
		\phi^{N}_{\mystype, d}(\xt) \;,\quad
		\text{with $\phi_{\mystype, d}^{N}(\xt)$ ideally describing how many times $(\mystype, d)$ is contained in $N$.}
\end{equation*}
Obviously, if all formulas $\phi_{\mystype, d}^{N}(\xt)$ are in $\GFC$, then also $\phi_{c}^{i+1}(\xt)$ is in $\GFC$.
Furthermore, for any $\sig$-structure $\B$ and all tuples $\bt \in \tB$,
\begin{quote}
	\enquote{$\B, \bt \models \phi^{i+1}_{c}(\xt) \iff \col{i+1}{\bt} = c$}
\end{quote}
holds if, and only if,
\begin{quote}
	\enquote{$\B, \bt \models \phi_N(\xt) \iff \Nset{i+1}{\B}(\bt) = N$}
\end{quote}
holds.
However, $\phi^{N}_{\mystype, d}(\xt)$ will not describe the number $\mult{N}((\mystype, d))$; instead it will describe the sum over all similarity types that \emph{contain} $\mystype$, i.e., $\sum_{\mystype' \supseteq \mystype} \mult{N}((\mystype', d))$.
This is fine though, due to the following claim.
\medskip

\noindent
\emph{Claim~a:} For all $\sig$-structures $\B$ of size strictly equal to $\A$, and for all tuples $\bt \in \tB$, the following statements are equivalent:
\begin{statements}[topsep=0pt, noitemsep]
\item\label{claima:one} $\Nset{i+1}{\B}(\bt)=N$.
\item\label{claima:two} For all $\mystype \in \sTypes$ and all $d \in \myC_i$ we have \ \ $\sum_{\mystype' \supseteq \mystype} \mult{\Nset{i+1}{\B}(\bt)}((\mystype', d)) \ = \ \sum_{\mystype' \supseteq \mystype} \mult{N}((\mystype', d))$.
\end{statements}
\medskip

\noindent
\emph{Proof:}
By definition, $\Nset{i+1}{\B}(\bt)=N$ iff for all $\mystype \in \sTypes$ and all $d \in \myC_i$ we have \ \ $\mult{\Nset{i+1}{\B}(\bt)}((\mystype, d)) \ = \ \mult{N}((\mystype, d))$.
Thus, the direction \enquote{\ref{claima:one}~$\Rightarrow$~\ref{claima:two}} holds trivially.

For the direction \enquote{\ref{claima:two}~$\Rightarrow$~\ref{claima:one}}, we prove the contraposition.
Let us assume that~\cref{claima:one} does not hold, i.e., there exist $\mystype \in \sTypes$ and $d \in \myC_i$ such that \ \ $\mult{\Nset{i+1}{\B}(\bt)}((\mystype, d)) \ \neq \ \mult{N}((\mystype, d))$.
We pick such $\mystype$ and $d$ so that $\mystype$ is maximal w.r.t.\ $\subseteq$.
I.e., $\mult{\Nset{i+1}{\B}(\bt)}((\mystype, d)) \ \neq \ \mult{N}((\mystype, d))$, but for all $\mystype'$ with $\mystype' \varsupsetneqq \mystype$ we have $\mult{\Nset{i+1}{\B}(\bt)}((\mystype', d)) \ = \ \mult{N}((\mystype', d))$.
Obviously, then $\sum_{\mystype' \supseteq \mystype} \mult{\Nset{i+1}{\B}(\bt)}((\mystype', d)) \ \neq \ \sum_{\mystype' \supseteq \mystype} \mult{N}((\mystype', d))$, and hence~\cref{claima:two} does not hold.
This completes the proof of Claim~a.
\bigskip

\noindent
In order to complete the induction step, it therefore suffices to prove the following claim.
\medskip

\noindent
\emph{Claim~b:} For every $\mystype \in \sTypes$ and every $d \in \myC_i$ there exists a formula $\phi^{N}_{\mystype, d}(\xt)$ in $\GFC$ such that for all $\sig$-structures $\B$ of size strictly equal to $\A$ and all tuples $\bt \in \tB$ we have:\\
$\B,\bt \models \phi^{N}_{\mystype, d}(\xt)$ \ \ $\iff$ \ \ $\sum_{\mystype' \supseteq \mystype} \mult{\Nset{i+1}{\B}(\bt)}((\mystype', d)) \ = \ \sum_{\mystype' \supseteq \mystype} \mult{N}((\mystype', d))$.
\medskip

\noindent
\emph{Proof:}
Let $\mystype \in \sTypes$ and $d \in \myC_i$.
Let $k$, $\ell$ be the arity of $c$ and $d$, respectively.
The tricky part is to encode $\mystype$ into the formula in such a way that it is a syntactically correct formula in $\GFC$.
To do this, we define \emph{left} ($\equiv_{\links}$) and \emph{right} ($\equiv_\rechts$) equivalence w.r.t.\ $\mystype$ as follows: $j \equiv_{\links} j'$ if there is an $n$ such that $(j,n), (j',n) \in \mystype$, and $j \equiv_{\rechts} j'$ if there is an $n$ such that $(n,j), (n,j') \in \mystype$.
By ${[j]}_{\links}$ and ${[j]}_{\rechts}$ we denote the corresponding equivalence classes.
We use these equivalences to decompose $\mystype$ into its \emph{left} and \emph{right} and \emph{center} parts defined as
\begin{eqnarray*}
	\lpart{\mystype} & \isdef & \set{ (j,j') \mid j \equiv_{\links} j' },
	\\
	\rpart{\mystype} & \isdef & \set{ (j,j') \mid j \equiv_{\rechts} j' },
	\quad\text{and}\quad \\
	\cpart{\mystype} & \isdef & \set{ (\min {[j]}_{\links}, \min {[j']}_{\rechts})\mid (j,j') \in \mystype }\;.
\end{eqnarray*}
\noindent
Notice that $\cpart{\mystype} \subseteq \mystype$; and if $\stp(\bt, \bt') = \mystype$, then $\stp(\bt) \supseteq \lpart{\mystype}$ and $\stp(\bt') \supseteq \rpart{\mystype}$.

We let $\xt' = (x_1', \dots, x_\ell')$ be the $\ell$-variable tuple where for all $j' \in [\ell]$, $x'_{j'} \isdef x_j$ if $(j,j') \in \cpart{\mystype}$ and $x'_{j'} \isdef \compl{x_{j'}}$ otherwise.
Let $\widehat{\xt}$ be the ordered variable tuple with $\vset(\widehat{\xt}) = \vset(\xt') \setminus \vset(\xt)$.

By the induction hypothesis, we already have available a formula $\phi^i_{d}(\xt')$ with the desired properties.
Choose an arbitrary $R \in \atp(d)$ (this is possible because $\atp(d) \neq \emptyset$), and let
\begin{equation*}
	\phi^N_{\mystype,d}(\xt) \ \ \isdef \ \ \existseq{n} \,\widehat{\xt} \qsep
	\bigl(
		R(\xt')
		\land \phi^i_{d}(\xt')
	\bigr),
	\qquad\text{where}\quad
	n \isdef \sum_{\mystype' \supseteq \mystype} \mult{N}((\mystype', d))\;.
\end{equation*}
It is easy to verify that $\phi^N_{\mystype,d} \in \GFC$ and $\gd(\phi^N_{\mystype,d}) = \gd(\phi^i_d)+1$.

Let $\B$ be an arbitrary $\sig$-structure of size strictly equal to $\A$, and let $\bt$ be an arbitrary tuple in $\tB$ of arity $k$.
It now suffices to show that $\B, \bt \models \phi^N_{\mystype,d}(\xt) \iff \sum_{\mystype' \supseteq \mystype} \mult{\Nset{i+1}{\B}(\bt)}((\mystype',d)) = n$.
By construction of the formula, we have:
\begin{align*}
	\B, \bt \models \phi_{\mystype,d}^{N}(\xt) \iff\; &\text{there exist exactly $n$ tuples $\bt'$ such that } b_j = b'_{j'} \text{ for all } (j,j') \in \cpart{\mystype}\\
	&\text{and } \B, \bt' \models ( R(\xt') \land \phi_{d}^{i}(\xt') ) \\
	\iff\; &\text{there exist exactly $n$ tuples $\bt'$ such that } b_j = b'_{j'} \text{ for all } (j,j') \in \cpart{\mystype} \\
	&\text{and } \col{i}{\bt'} = d \text{ and } \B, \bt' \models R(\xt') \\
	\iff\; &\text{there exist exactly $n$ tuples $\bt'$ such that } b_j = b'_{j'} \text{ for all } (j,j') \in \cpart{\mystype} \\
	&\text{and } \col{i}{\bt'} = d.
\end{align*}
For all $\bt' \in \tB$, the inclusion $\mystype \subseteq \stp(\bt, \bt')$ trivially implies that $b_j = b'_{j'}$ holds for all $(j,j') \in \mystype$.
Thus, $\sum_{\mystype' \supseteq \mystype} \mult{\Nset{i+1}{\B}(\bt)}((\mystype',d)) \leq n$.
Hence, it remains to show that for all $\bt' \in \tB$ with $\col{i}{\bt'} = d$,
\begin{quote}
	\enquote{$b_j = b'_{j'}$ for all $(j,j') \in \cpart{\mystype}$} \qquad implies that \qquad $\stp(\bt, \bt') \supseteq \mystype$.
\end{quote}
First, note that $\stp(\bt') = \stp(d)$.

Let $(j,j') \in \mystype$.
We have to show that $(j,j') \in \stp(\bt, \bt')$.
By construction, there exists a pair $(j_1, j'_1) \in \cpart{\mystype}$ with $j \equiv_{\links} j_1$ and $j' \equiv_{\rechts} j'_1$ and $(j_1, j'_1) \in \stp(\bt, \bt')$.
Since $\rpart{\mystype} \subseteq \stp(d)$ and $\stp(d) = \stp(\bt')$, we obtain: $(j', j'_1) \in \stp(\bt')$.
Similarly, $\lpart{\mystype} \subseteq \stp(\bt)$, thus $(j, j_1) \in \stp(\bt)$.
I.e., $b_{j} = b_{j_1} = b'_{j'_1} = b'_{j'}$, and thus, $(j,j') \in \stp(\bt, \bt')$.
Since $(j,j') \in \mystype$ was chosen arbitrarily, we have proven that $\mystype \subseteq \stp(\bt, \bt')$.
\\
In total, we obtain:
\begin{quote}
	$\B, \bt \models \phi_{\mystype,d}^{N}(\xt)$ \ \ $\iff$ \ \ $\sum_{\mystype' \supseteq \mystype} \mult{\Nset{i+1}{\B}(\bt)}((\mystype', d)) \ \ = \ n$.
\end{quote}
This completes the proof of Claim~b.
\medskip

By the construction of the formula $\phi_{N}(\xt)$ (at the beginning of the inductive step) and using Claim~a, we obtain that the following equivalence holds for all $\sig$-structures $\B$ of size strictly equal to $\A$ and for all tuples $\bt \in \tB$:
\begin{quote}
	$\B, \bt \models \phi_{N}(\xt)$ \ \ $\iff$ \ \ $\Nset{i+1}{\B}(\bt) = N$.
\end{quote}
This completes the inductive step, and hence it completes the proof of \cref{lem:describing-color-in-gfc}. \end{proof}

\subsection{Proof of Lemma~\ref{lem:formula-is-spoiler-strategy}}
\formulaIsSpoilerStrategy*
\begin{proof}[Proof of~\cref{lem:formula-is-spoiler-strategy}]
The proof works by induction over the definition of $\GFC$.
\paragraph{Base Cases.}
$\phi$ is of one of the following forms.

\medskip\noindent\emph{\Cref{def:logic-syntax:rules:relation}}:\;  $\phi(\xt)$ is of the form $R(x_{i_1}, \dots, x_{i_\ell})$, with $i_1,\ldots,i_\ell \in [k]$.
Since $\A, \at \models \phi(\xt) \iff \B, \bt \not\models \phi(\xt)$, either $(a_{i_1}, \dots, a_{i_\ell}) \in R^\A$ and $(b_{i_1}, \dots, b_{i_\ell}) \not\in R^\B$ or $(a_{i_1}, \dots, a_{i_\ell}) \not\in R^\A$ and $(b_{i_1}, \dots, b_{i_\ell}) \in R^\B$.
Hence, $\atp((a_{i_1}, \dots, a_{i_\ell})) \neq \atp((b_{i_1}, \dots, b_{i_\ell}))$.
I.e., the configuration is distinguishing and Spoiler has a $0$-round winning strategy on $(\A, \at)$, $(\B, \bt)$.

\medskip\noindent\emph{\Cref{def:logic-syntax:rules:equality}}:\;  $\phi(\xt)$ is of the form $x_i \lequal x_j$ with $i,j \in[k]$.
Since $\A, \at \models \phi(\xt) \iff \B, \bt \not\models \phi(\xt)$, either $a_i = a_j$ and $b_i \neq b_j$ or $a_i \neq a_j$ and $b_i = b_j$.
Hence, $\stp(\at) \neq \stp(\bt)$, i.e., Spoiler has a $0$-round winning strategy on $(\A, \at)$, $(\B, \bt)$.

\medskip\noindent This completes the induction base (note that, in both cases, $\gd(\phi) = 0$).

\paragraph*{Inductive Step.}~%
\medskip

\noindent\emph{\Cref{def:logic-syntax:rules:negation}}:\; $\phi(\xt)$ is of the form $\lnot \chi(\xt)$;
and by the induction hypothesis, the lemma's statement holds for the formula $\chi$.
In case that \enquote{$\A, \at \models \phi(\xt) {\iff} \B, \bt \not\models \phi(\xt)$} holds, also \enquote{$\A, \at \models \chi(\xt) {\iff} \B, \bt \not\models \chi(\xt)$} holds.
Hence, by the induction hypothesis Spoiler has a $\gd(\chi)$-round winning strategy on $(\A, \at)$, $(\B, \bt)$.
Since $\gd(\phi) = \gd(\chi)$, we are done.

\medskip\noindent\emph{\Cref{def:logic-syntax:rules:conjunction}}:\; $\phi(\xt)$ is of the form $(\psi_1(\xt) \land \psi_2(\xt))$; and by the induction hypothesis, the lemma's statement holds for each of the formulas $\psi_1$ and $\psi_2$.
In case that \enquote{$\A, \at \models \phi(\xt) {\iff} \B, \bt \not\models \phi(\xt)$} holds, there is an $i \in \set{1,2}$ such that \enquote{$\A, \at \models \psi_i(\xt) \iff \B, \bt \not\models \psi_i(\xt)$} holds.
Hence, by the induction hypothesis, Spoiler has a $\gd(\psi_i)$-round winning strategy on $(\A, \at)$, $(\B, \bt)$.
Since $\gd(\phi) \geq \gd(\psi_i)$, Spoiler also has a $\gd(\phi)$-round winning strategy on $(\A, \at), (\B, \bt)$.

\medskip\noindent\emph{\Cref{def:logic-syntax:rules:quantification}}:\; $\phi(\xt)$ is of the form $\existsgeq{n} \vartup \qsep (R(\xt') \land \psi(\xt'))$; and by the induction hypothesis, the lemma's statement holds for the formula $\psi$.
Let $\mystype \isdef \stp(\xt, \xt') = \set{ (i, j) \mid x_i = x_j' }$.
\smallskip

\noindent
\emph{Case 1:} Assume that $\A, \at \models \phi(\xt)$ and $\B, \bt \not\models \phi(\xt)$.
Then, there exist at least $n$ tuples $\at'$ with $\stp(\at, \at') \supseteq \mystype$ such that $\A, \at' \models (R(\xt') \land \psi(\xt'))$.

Spoiler picks the relation $R$.
For every bijection $\pi$ that Duplicator chooses, there must be at least one of these tuples $\at'$ such that $\A, \at' \models (R(\xt') \land \psi(\xt'))$, but $\B, \pi(\at') \not\models (R(\xt') \land \psi(\xt'))$ or $\stp(\bt, \pi(\at')) \not\supseteq \mystype$ --- otherwise, $\pi$ would witness that there are at least $n$ suitable $\bt'$ such that $\stp(\bt, \bt') \supseteq \mystype$ and $\B, \bt' \models (R(\xt') \land \psi(\xt'))$.
Let $\bt' \isdef \pi(\at')$.
Spoiler chooses $\at'$ and creates the configuration $(\A, \at')$, $(\B, \bt')$.
\\
If $\stp(\bt, \bt') \not\supseteq \mystype$, he has won the round, since this also implies that $\stp(\at, \at') \neq \stp(\bt, \bt')$.
\\
If $\B, \bt' \not\models R(\xt')$, then $(\A, \at')$, $(\B, \bt')$ is a distinguishing configuration, i.e., Spoiler wins this round as well.
Otherwise, $\B, \bt' \not\models \psi(\xt')$ must hold.
In this case, by the induction hypothesis Spoiler has a $\gd(\psi)$-round winning strategy on $(\A, \at')$, $(\B, \bt')$.
This means, that Spoiler has a $\gd(\psi){+}1$-round winning strategy on $(\A, \at)$, $(\B, \bt)$.
Since $\gd(\phi) = \gd(\psi) + 1$, we are done.
\smallskip

\noindent
\emph{Case 2:} Assume that $\A, \at \not\models \phi(\xt)$ and $\B, \bt \models \phi(\xt)$.
Then, there exist at least $n$ tuples $\bt'$ with $\stp(\bt, \bt') \supseteq \mystype$ such that $\B, \bt' \models (R(\xt') \land \psi(\xt'))$.

Again, Spoiler picks the relation $R$.
For every bijection $\pi$ that Duplicator chooses, there must be a tuple $\bt'$ such that $\B, \bt' \models (R(\xt') \land \psi(\xt'))$, but $\A, \pi^{-1}(\bt') \not\models (R(\xt') \land \psi(\xt'))$ or $\stp(\at, \pi^{-1}(\bt')) \not\supseteq \mystype$.
With the same argument as in \emph{Case~1}, Spoiler has a $\gd(\phi)$-round winning strategy on $(\A, \at)$, $(\B, \bt)$.

This completes the proof of \cref{lem:formula-is-spoiler-strategy}.
\qedhere \end{proof}

\subsection{Proof of Lemma~\ref{lem:same-color-is-dup-strategy}}
\sameColorIsDupStrategy*
\begin{proof}[Proof of~\cref{lem:same-color-is-dup-strategy}]
For the lemma's first claim, consider $\at$, $\bt$ such that $\col{1}{\at} = \col{1}{\bt}$.
Then $\stp(\at) = \stp(\bt)$.
Assume for contradiction that there exists an $\ell \in [\ar(\sig)]$ and indices $i_1, \ldots, i_\ell \in[k]$ such that $\atp(\at') \neq \atp(\bt')$, for $\at' = (a_{i_1}, \dots, a_{i_\ell})$ and $\bt' = (b_{i_1}, \dots, b_{i_\ell})$.
This implies that $\at' \in \tA$ or $\bt' \in \tB$ (because otherwise, $\atp(\at') = \emptyset = \atp(\bt')$).
\\
Furthermore, note that $\stp(\at) = \stp(\bt)$ implies that $\stp(\at') = \stp(\bt')$ and $\stp(\at,\at') = \stp(\bt,\bt')$.
\medskip

\noindent
\emph{Case~1}: \ $\at' \in \tA$.
In this case, RCR assigns colors to $\at'$.
Let $c \deff \col{0}{\at'}$.
Hence, by definition of RCR, $c=(\atp(\at'),\stp(\at'))$.
Let $\mystype \deff \stp(\at,\at')$, and note that $\mystype = \stp(\bt,\bt')$.

Furthermore, note that by definition, $\at'$ is the \emph{only} tuple $\hat{\at} \in {V(\A)}^\ell$ with $\stp(\at,\hat{\at}) = \mystype$.
Analogously, $\bt'$ is the \emph{only} tuple $\hat{\bt} \in {V(\B)}^\ell$ with $\stp(\bt,\hat{\bt}) = \mystype$.

By definition of RCR, the tuple $(\mystype,c)$ has multiplicity 1 in $\Nset{1}{\A}(\at)$ --- because $\at'$ witnesses that this tuple belongs to $\Nset{1}{\A}(\at)$, and we know that there don't exist further potential witnesses.
With a similar reasoning we obtain that the tuple $(\mystype,c)$ has multiplicity 0 in  $\Nset{1}{\B}(\bt)$ --- because $\bt'$ is the only potential witness for a multiplicity $>0$, but $\atp(\at') \neq \atp(\bt')$ implies that $\col{0}{\bt'} \neq \col{0}{\at'}$, and hence $\bt'$ cannot be a witness.

In summary, $\mult{\Nset{1}{\A}(\at)}((\mystype, c)) = 1 \neq 0 = \mult{\Nset{1}{\B}(\bt)}((\mystype, c))$.
Therefore, $\Nset{1}{\A}(\at) \neq \Nset{1}{\B}(\bt)$.
But, by definition of RCR, this implies that $\col{1}{\at} \neq \col{1}{\bt}$.
A contradiction!

\medskip

\noindent
\emph{Case~2}: \ $\bt' \in \tB$.
In this case we can argue analogously as in Case~1 (switching the roles of $\at'$ and $\bt'$).
\medskip

\noindent
This completes the proof of the lemma's first claim.
\bigskip

\noindent
We prove the lemma's second claim by induction on $i$.
\paragraph*{Base Case.}
Consider $i = 0$.
If $\col{1}{\at} = \col{1}{\bt}$ then, by the lemma's first claim, $(\A, \at)$, $(\B, \bt)$ is not distinguishing.
Hence, Duplicator has a $0$-round winning strategy on $(\A, \at)$, $(\B, \bt)$.

\paragraph*{Inductive Step.}
Consider $i \in \natpos$.
Let $\col{i+1}{\at} = \col{i+1}{\bt}$, and let $\ccount{\A}{c} = \ccount{\B}{c}$ for all $c \in \cols{i}{\A} \union \cols{i}{\B}$.
By definition of RCR, the former means that  $\col{i}{\at} = \col{i}{\bt}$ and $\Nset{i+1}{\A}(\at) = \Nset{i+1}{\B}(\bt)$.
Taking into account the definition of $\Nset{i+1}{\A}(\at)$ and $\Nset{i+1}{\B}(\bt)$ in RCR, we obtain that for any $R \in \sig$ which Spoiler might choose, Duplicator can give a bijection $\pi\colon R^\A \to R^\B$ such that for all $\at' \in R^\A$ we have $\col{i}{\at'} = \col{i}{\pi(\at')}$ and $\stp(\at, \at') = \stp(\bt, \pi(\at'))$; this can be seen as follows:
Let $X \deff \setc{\at' \in R^{\A}}{\stp(\at,\at') \neq \emptyset}$ and $Y \deff \setc{\bt' \in R^{\B}}{\stp(\bt,\bt') \neq \emptyset}$.
Because of $\Nset{i+1}{\A}(\at) = \Nset{i+1}{\B}(\bt)$, there exists a bijection $\pi\colon X \to Y$ such that $\stp(\at,\at') = \stp(\bt,\pi(\at'))$ and $\col{i}{\at'} = \col{i}{\pi(\at')}$ for all $\at' \in X$.
Let $\myC_R \deff \setc{c \in \cols{i}{\A} \union \cols{i}{\B}}{R \in \atp(c)}$.
Note that $R^{\A} = \setc{\at' \in \tA}{\col{i}{\at'} \in \myC_R}$ and $R^{\B} = \setc{\bt' \in \tB}{\col{i}{\bt'} \in \myC_R}$.
Furthermore, $R^{\A} \setminus X = \setc{\at'\in R^\A}{\stp(\at,\at') = \emptyset}$ and $R^{\B} \setminus Y = \setc{\bt' \in R^\B}{\stp(\bt,\bt') = \emptyset}$.
Since, by assumption, $\ccount{\A}{c} = \ccount{\B}{c}$ for all $c \in \myC_R$, we can extend $\pi$ to a bijection $\pi\colon R^\A \to R^\B$ such that for all $\at' \in R^\A \setminus X$ we have $\pi(\at') \in R^\B \setminus Y$ and $\col{i}{\at'} = \col{i}{\pi(\at')}$.
This completes the reasoning why Duplicator can choose a bijection $\pi$ with the claimed properties.

Since $i \geq 1$, and by definition of RCR, we obtain that $\col{i}{\at'} = \col{i}{\pi(\at')}$ implies that $\col{1}{\at'} = \col{1}{\pi(\at')}$.
Hence, by the lemma's first claim, the configuration $(\A,\at')$, $(\B,\pi(\at'))$ is not distinguishing.
I.e., Duplicator wins this round for any $\at'$ and $\bt' \deff \pi(\at')$ that Spoiler could pick.
\\
Note that $\col{i}{\at'} = \col{i}{\bt'}$.
Furthermore, by definition of RCR,
\begin{quote}
	\enquote{$\ccount{\A}{c} = \ccount{\B}{c}$ for all $c \in \cols{i}{\A} \union \cols{i}{\B}$}
\end{quote}
implies that also
\begin{quote}
	\enquote{$\ccount{\A}{c'} = \ccount{\B}{c'}$ for all $c' \in \cols{i-1}{\A} \union \cols{i-1}{\B}$}.
\end{quote}
Hence, by the induction hypothesis, Duplicator has an $(i{-}1)$-round winning strategy on $(\A, \at')$, $(\B, \bt')$.
Thus, Duplicator has an $i$-round winning strategy on $(\A, \at)$, $(\B, \bt)$.
This completes the proof of \cref{lem:same-color-is-dup-strategy}. \end{proof}

\clearpage
\section{Details omitted in  Section~\ref{sec:main-runtime}}%
\label{app:alt-rep}

\subsection{Proof of Lemma~\ref{lem:slice-observations}}

\sliceObservations*
\begin{proof}~%
\begin{enumerate}[label={(\alph*)}]
	\item By definition, $\stp(\at, \bt) \neq \emptyset$ holds iff there are $i\in[k]$ and $j\in[k']$ such that $a_i = b_j$. The latter holds iff $\tset(\at) \intersect \tset(\bt) \neq  \emptyset$.

	\item 
	If $\ar(\at) > \ar(\bt)$ then $(k,k) \in \stp(\at)$ but $(k,k) \not\in \stp(\bt)$, and hence $\stp(\at) \neq \stp(\bt)$.

	If $\ar(\at) < \ar(\bt)$ then $(k',k') \in \stp(\bt)$ but $(k',k') \not\in \stp(\at)$, and hence $\stp(\at) \neq \stp(\bt)$.

	It remains to consider the case where $\ar(\at) = \ar(\bt)$, i.e., $k=k'$.
	Note that in this case, the function $\beta$ is well-defined iff for all $i,j \in [k]$ with $a_i = a_j$ we have $b_i = b_j$.
	Thus, $\beta$ is well-defined iff $\stp(\at) \subseteq \stp(\bt)$.

	Hence, if $\beta$ is \emph{not} well-defined then $\stp(\at)\neq\stp(\bt)$.

	It remains to consider the case where $\beta$ \emph{is} well-defined.
	In this case we have $\stp(\at) \subseteq \stp(\bt)$.
	By definition, $\beta$ is surjective.
	Thus, $\beta$ is bijective iff $\beta$ is injective iff for all $i,j \in [k]$ with $a_i \neq a_j$ we have $b_i\neq b_j$.
	The latter is the case iff $\stp(\bt) \subseteq \stp(\at)$, and this is the case iff $\stp(\bt) = \stp(\at)$ (because we already know that $\stp(\at)\subseteq\stp(\bt)$).
	This completes the proof of \cref{lem:slice-observations:self-stp-equal-iff-bijection}.
	
	\item $\slice \in \slices(\at) \iff \tset(\slice) \subseteq \tset(\at)$, i.e., for every $i \in [\ell]$, there is a $j \in [k]$ such that $s_i = a_j$.
	Furthermore, $s_i = a_j \iff (i,j) \in \stp(\slice,\at) \iff (j, i) \in \stp(\at, \slice)$.

	\item
	\enquote{$\Longleftarrow$} is obvious.
	For \enquote{$\Longrightarrow$} let $(s'_1, \ldots, s'_{\ell'}) = \slice'$ and assume that $\stp(\at,\slice)=\stp(\at,\slice')$.
	From \cref{lem:slice-observations:every-j-has-i} we know that for every $i \in [\ell]$ there exists a $j \in [k]$ such that $(j,i) \in \stp(\at,\slice) = \stp(\at,\slice')$, and hence $s_i = a_j = s'_i$.
	Thus, $\ell' \geq \ell$ and $s'_i = s_i$ for all $i \leq \ell$.
	Furthermore, applying \cref{lem:slice-observations:every-j-has-i} to $\slice'$ we obtain that in particular for $i = \ell'$ there is a $j \in [k]$ such that $(\ell',j) \in \stp(\slice',\at) = \stp(\slice,\at)$.
	This implies that $\ell \geq \ell'$, and hence $\ell = \ell'$ and $\slice = \slice'$.
\end{enumerate}
This completes the proof of \cref{lem:slice-observations}.
\end{proof}

\subsection{Proof of Lemma~\ref{lem:self-stp-identifies-stp-of-slices}}

\selfStpIdentifiesSlices*
\begin{proof}
	Let $\at = (a_1, \dots, a_k)$ and $\bt = (b_1, \dots, b_\ell)$ for $k,\ell \in \natpos$.
	
	\enquote{$\Longleftarrow$}:\; By assumption there is a bijection $\pi_{\slices}\colon \slices(\at) \to \slices(\bt)$ such that $\stp(\at, \slice) = \stp(\bt, \pi_{\slices}(\slice))$ holds for all $\slice \in \slices(\at)$.

	Consider an arbitrary slice $\tup{t}'$ with $\tup{t}'\in\slices(\bt)$ and $\tset(\tup{t}') = \tset(\bt)$ (obviously, such a slice exists).
	In particular, there exists a $j \in [\ar(\tup{t}')]$ such that $b_{\ell} = t'_{j}$, and hence $(\ell,j) \in \stp(\bt,\tup{t}')$.
	Let $\slice' \isdef \pi_{\slices}^{-1}(\tup{t}')$ and note that, by the choice of $\pi_{\slices}$ we have $\stp(\at, \slice') = \stp(\bt, \tup{t}')$.
	Hence, from $(\ell,j)\in\stp(\bt,\tup{t}')$ we obtain that $(\ell,j) \in \stp(\at,\slice')$, and hence $k \geq \ell$. 

	Now, let us fix an arbitrary slice $\slice \in \slices(\at)$ with $\tset(\slice) = \tset(\at)$.
	Let $\tup{t} \isdef \pi_{\slices}(\slice)$.
	By the choice of $\pi_{\slices}$ we have $\stp(\at,\slice) = \stp(\bt,\tup{t})$.
	From $\tset(\slice) = \tset(\at)$ and $\slice \in \slices(\at)$, we obtain that for each $i \in [k]$ there exists exactly one $j_i \in [\ar(\slice)]$ with $(i,j_i) \in \stp(\at, \slice)$.
	From $\stp(\at, \slice) = \stp(\bt, \tup{t})$ we obtain that $(i,j_i) \in \stp(\bt,\tup{t})$ for every $i \in [k]$.
	For $i = k$ this in particular implies that $\ell \geq k$.
	In summary, we have shown that $k = \ell$.

	Finally, let us fix arbitrary $i, \hati \in [k]$.
	We have $(i,\hati) \in \stp(\at) \iff a_i = a_\hati \iff j_i=j_{\hati} \iff (i, j_i), (\hati, j_i) \in \stp(\at, \slice)$.
	From $\stp(\at, \slice) = \stp(\bt, \tup{t})$ we obtain that $(i, j_i), (\hati, j_i) \in \stp(\at, \slice) \iff (i, j_i), (\hati, j_i) \in \stp(\bt, \tup{t}) \iff$ $b_i = t_{j_i} = b_{\hati} \iff b_i = b_{\hati} \iff (i,\hati) \in \stp(\bt)$.

	In summary, we obtain that $\stp(\at)=\stp(\bt)$. This completes the proof of \enquote{$\Longleftarrow$}.
	\medskip

	\enquote{$\Longrightarrow$}:\;
	By assumption we have $\stp(\at) = \stp(\bt)$.
	From \cref{lem:slice-observations:self-stp-equal-iff-bijection} we obtain that $k = \ell$ and the function $\beta\colon \tset(\at) \to \tset(\bt)$ with $\beta(a_i) \isdef b_i$ for all $i \in [k]$ is well-defined and bijective.
	For any $\slice \in \slices(\at)$ of the form $(s_1, \ldots, s_n) = \slice$ (in particular, $1 \leq n = \ar(\slice)$), we let $\pi_\slices(\slice)\isdef (\beta(s_1), \ldots, \beta(s_n))$.

	We first show that $\pi_\slices(\slice) \in \slices(\bt)$:
	Let $\tup{t} \isdef \pi_\slices(\slice)$, i.e., $\tup{t} = (t_1, \ldots, t_n)$ and $t_i = \beta(s_i)$ for all $i \in [n]$.
	Since $\beta$ is a bijection from $\tset(\at)$ to $\tset(\bt)$ and the elements $s_1, \ldots, s_n$ are pairwise distinct, $t_1, \ldots, t_n$ are pairwise distinct elements in $\tset(\bt)$.
	Thus, $\tup{t} \in \slices(\bt)$.

	Next, we show that $\stp(\at,\slice) = \stp(\bt,\tup{t})$:
	For arbitrary $i \in [k]$ and $j \in [n]$ we have $(i,j) \in \stp(\at,\slice) \iff a_i = s_j \iff \beta(a_i) = \beta(s_j) \iff b_i = t_j \iff (i,j) \in \stp(\bt,\tup{t})$.
	Hence, we have $\stp(\at, \slice) = \stp(\bt, \tup{t})$.

	In summary, we have shown that $\pi_\slices$ is a mapping  $\pi_\slices\colon \slices(\at) \to \slices(\bt)$ that satisfies $\stp(\at,\slice) = \stp(\bt,\pi_\slices(\slice))$ for all $\slice \in \slices(\at)$.

	Next, we show that $\pi_\slices$ is \emph{injective}:
	Consider arbitrary $\slice, \slice' \in \slices(\at)$ of the form $(s_1, \dots, s_n)$ and $(s_1', \dots, s_m')$ such that $\pi_{\slices}(\slice) = \pi_{\slices}(\slice')$.
	By definition of $\pi_{\slices}$ we have $n = m$, and $\beta(s_i) = \beta(s'_i)$ for all $i \in [n]$.
	Since $\beta$ is injective, we obtain that $s_i = s'_i$ holds for all $i \in [n]$.
	Thus, $\slice = \slice'$.
	Hence, $\pi_\slices$ is injective.

	Next, we show that $\pi_\slices$ is \emph{surjective}:
	Consider an arbitrary $\tup{t} \in \slices(\bt)$ of the form $(t_1, \ldots, t_n)$.
	For $i \in [n]$ let $s_i \isdef \beta^{-1}(t_i)$, and let $\slice \isdef (s_1, \ldots, s_n)$.
	Clearly, $\tset(\slice) \subseteq \tset(\at)$.
	Furthermore, for any $i, j \in [n]$ with $i \neq j$ we have $t_i \neq t_j$, and hence (since $t_i = \beta(s_i)$ and $t_j = \beta(s_j)$) also $s_i \neq s_j$.
	Therefore, $\slice \in \slices(\at)$.
	We are done by noting that $\pi_\slice(\slice) = \tup{t}$.
	This completes the proof of \enquote{$\Longrightarrow$} and the proof of \cref{lem:self-stp-identifies-stp-of-slices}.
\end{proof}

\subsection{Proof of Lemma~\ref{lemma:easy-properties-of-pi-slices}}

\easyPropertiesOfPiSlices*
\begin{proof}
	We first show that $\ar(\slice) = \ar(\tup{t})$:
	By assumption, $\slice \in \slices(\at)$.
	Hence, in particular for $i \isdef \ar(\slice)$ there exists a $j_i$ such that $s_i = a_{j_i}$.
	I.e., $(j_i,i) \in \stp(\at,\slice) = \stp(\bt,\tup{t})$.
	Hence, $b_{j_i} = t_i$ and, in particular, $i \leq \ar(\tup{t})$, i.e., $\ar(\slice) \leq \ar(\tup{t})$.
	By a similar reasoning we obtain that $\ar(\tup{t}) \leq \ar(\slice)$:
	By assumption, $\tup{t} \in \slices(\bt)$.
	Hence, in particular for $i \isdef \ar(\tup{t})$ there exists a $j_i$ such that $t_i = b_{j_i}$.
	I.e., $(j_i,i) \in \stp(\bt,\tup{t}) = \stp(\at,\slice)$.
	Hence, $a_{j_i} = s_i$ and, in particular, $i \leq \ar(\slice)$, i.e., $\ar(\tup{t}) \leq \ar(\slice)$.
	This proves that the lemma's first statement is correct.
	\smallskip

	Note that the lemma's third statement follows immediately from the lemma's second statement.
	For proving the lemma's second statement, let $\ell \isdef \ar(\slice)$ and $\ell' \isdef \ar(\slice')$.
	From the lemma's first statement we know that $\ell = \ar(\tup{t})$ and $\ell' = \ar(\tup{t}')$.
	Let $(s_1, \ldots, s_\ell) = \slice$ and $(t_1, \ldots, t_\ell) = \tup{t}$ and $(s'_1, \ldots, s'_{\ell'}) = \slice'$ and $(t'_1, \ldots, t'_{\ell'}) = \tup{t}'$.

	For direction \enquote{$\Longrightarrow$} assume that $\tset(\slice) \subseteq \tset(\slice')$.
	This implies that $\ell \leq \ell'$ and there is an injection $\beta\colon[\ell] \to [\ell']$ such that $s_i = s'_{\beta(i)}$ for all $i \in [\ell]$.
	Furthermore, from $\slice \in \slices(\at)$ we know that for every $i \in [\ell]$ there is a $j_i$ such that $s_i = a_{j_i}$.
	Hence, $(j_i,i) \in \stp(\at,\slice) = \stp(\bt,\tup{t})$, and therefore we have $t_i = b_{j_i}$.
	On the other hand, we have $a_{j_i} = s_i = s'_{\beta(i)}$.
	Hence, $(j_i,\beta(i)) \in \stp(\at,\slice') = \stp(\bt,\tup{t}')$, and therefore we have $t'_{\beta(i)} = b_{j_i}$, i.e., $t_i = t'_{\beta(i)}$.
	This holds for every $i \in [\ell]$.
	Thus, $\tset(\tup{t}) \subseteq \tset(\tup{t}')$.
	This completes the proof for direction \enquote{$\Longrightarrow$}.

	For direction \enquote{$\Longleftarrow$} assume that $\tset(\tup{t}) \subseteq \tset(\tup{t}')$.
	This implies that $\ell \leq \ell'$ and there is an injection $\beta\colon[\ell] \to [\ell']$ such that $t_i = t'_{\beta(i)}$ for all $i \in [\ell]$.
	Furthermore, from $\tup{t} \in \slices(\bt)$ we know that for every $i \in [\ell]$ there is a $j_i$ such that $t_i = b_{j_i}$.
	Hence, $(j_i,i) \in \stp(\bt,\tup{t}) = \stp(\at,\slice)$, and therefore we have $s_i = a_{j_i}$.
	On the other hand, we have $b_{j_i} = t_i = t'_{\beta(i)}$.
	Hence, $(j_i,\beta(i)) \in \stp(\at,\tup{t}') = \stp(\at,\slice')$, and therefore we have $s'_{\beta(i)} = a_{j_i}$, i.e., $s_i = s'_{\beta(i)}$.
	This holds for every $i \in [\ell]$.
	Thus, $\tset(\slice) \subseteq \tset(\slice')$.
	This completes the proof for direction \enquote{$\Longleftarrow$} and the proof of \cref{lemma:easy-properties-of-pi-slices}.
\end{proof}

\subsection{Proof of Lemma~\ref{lem:vgrep-equiv-rcr}}\label[appendix]{appendix:lem:vgrep-equiv-rcr}

\vgrepEquivRCR*
\begin{proof}
	\enquote{\ref{statement1-equiv}~$\Longrightarrow$~\ref{statement2-equiv}}:\; By assumption we have
	\[
		\mset[\big]{ (\stp(\at, \ct), \myc(\ct)) \mid \ct \in N(\at) } = \mset[\big]{ (\stp(\bt, \ct), \myc(\ct)) \mid \ct \in N(\bt) }\;.
	\]
	Thus, there exists a bijection $\pi\colon N(\at) \to N(\bt)$ such that $\stp(\at, \ct) = \stp(\bt, \pi(\ct))$ and $\myc(\ct) = \myc(\pi(\ct))$ for all $\ct \in N(\at)$.
	Consider an arbitrary $\slice \in \slices(\at)$.
	Let $\tup{t} \isdef \pi_{\slices}(\slice)$.
	Note that the proof of \enquote{\ref{statement1-equiv}~$\Longrightarrow$~\ref{statement2-equiv}} is complete as soon as we have found a bijection $\pi'\colon \slices^{-1}(\slice) \to \slices^{-1}(\tup{t})$ such that $\stp(\slice,\ct) = \stp(\tup{t},\pi'(\ct))$ and $f(\ct) = f(\pi'(\ct))$ holds for all $\ct \in \slices^{-1}(\slice)$.
	In the following, we construct such a mapping $\pi'$.

	By the choice of $\pi_{\slices}$ we have $\tup{t} \in \slices(\bt)$ and $\stp(\at, \slice) = \stp(\bt, \tup{t})$, and hence in particular $\ar(\slice) = \ar(\tup{t})$.
	Note that for every $\ct \in \slices^{-1}(\slice)$ we have $\slice \in \slices(\ct)$ and $\slice \in \slices(\at)$, and hence $\emptyset \neq \tset(\slice) \subseteq \tset(\ct) \cap \tset(\at)$.
	From \cref{lem:slice-observations:stp-nonempty-iff-intersect} we obtain that $\stp(\at,\ct) \neq \emptyset$, and thus $\ct \in N(\at)$.
	Hence, $\slices^{-1}(\slice) \subseteq N(\at)$.

	Let $\pi'$ be the restriction of $\pi$ to the set $\slices^{-1}(\slice)$, i.e., $\pi'$ is the mapping $\pi'\colon \slices^{-1}(\slice) \to N(\bt)$ defined by $\pi'(\ct) \isdef \pi(\ct)$ for every $\ct \in \slices^{-1}(\slice)$.
	Clearly, $\pi'$ is injective (since $\pi$ is), and for all $\ct \in \slices^{-1}(\slice)$ we have $f(\ct) = f(\pi'(\ct))$ (since $\pi$ satisfies $f(\ct) = f(\pi(\ct))$).
	Thus, to complete the proof of \enquote{\ref{statement1-equiv}~$\Longrightarrow$~\ref{statement2-equiv}}, it suffices to show that $\img(\pi') = \slices^{-1}(\tup{t})$ and that $\stp(\slice,\ct) = \stp(\tup{t},\pi'(\ct))$ holds for all $\ct \in \slices^{-1}(\slice)$.
	\smallskip

	Let us consider an arbitrary $\ct \in \slices^{-1}(\slice)$, and let $\tup{d} \isdef \pi'(\ct) = \pi(\ct)$. 
	By the choice of $\pi$ we have $\stp(\at, \ct) = \stp(\bt, \tup{d})$.
	Recall that by the choice of $\pi_{\slices}$ we have $\tup{t} \in \slices(\bt)$ and $\stp(\at, \slice) = \stp(\bt, \tup{t})$, and hence in particular $\ar(\slice) = \ar(\tup{t})$.
	\medskip

	\begin{claim*}
		$\stp(\slice, \ct) = \stp(\tup{t}, \dt)$.
	\end{claim*}
	\begin{proof}[Proof of Claim]
	For \enquote{$\subseteq$} consider an arbitrary $(i,j) \in \stp(\slice, \ct)$.
	Hence, $s_i = c_j$.
	From $\slice \in \slices(\at)$ we obtain that there exists an $\hati$ such that $s_i = a_{\hati}$.
	I.e., $s_i = a_{\hati} = c_j$, and hence $(\hati, j) \in \stp(\at, \ct)$ and $(\hati,i) \in \stp(\at,\slice)$.
	From $\stp(\at,\ct) = \stp(\bt,\tup{d})$ we obtain that $(\hati,j) \in \stp(\bt,\tup{d})$, and from $\stp(\at, \slice) = \stp(\bt, \tup{t})$ we obtain that $(\hati,i) \in \stp(\bt,\tup{t})$.
	Hence, $t_i = b_{\hati} = d_j$, and thus $(i,j) \in \stp(\tup{t},\tup{d})$.
	This completes the proof of \enquote{$\subseteq$}.
	For proving \enquote{$\supseteq$}, consider an arbitrary $(i,j) \in \stp(\tup{t}, \tup{d})$.
	Hence, $t_i = d_j$.
	From $\tup{t} \in \slices(\bt)$ we obtain that there exists an $\hati$ such that $t_i = b_{\hati}$.
	I.e., $t_i = b_{\hati}=d_j$, and hence $(\hati, j) \in \stp(\bt, \tup{d})$ and $(\hati,i) \in \stp(\bt,\tup{t})$.
	From $\stp(\at,\ct) = \stp(\bt,\tup{d})$ we obtain that $(\hati,j) \in \stp(\at,\ct)$, and from $\stp(\at, \slice) = \stp(\bt, \tup{t})$ we obtain that $(\hati,i) \in \stp(\at,\slice)$.
	Hence, $s_i = a_{\hati} = c_j$, and thus $(i,j) \in \stp(\slice,\ct)$.
	This completes the proof of \enquote{$\supseteq$} and the proof of the claim.
\end{proof}

\begin{claim*}
	$\dt\in \slices^{-1}(\tup{t})$.
\end{claim*}
\begin{proof}[Proof of Claim]
	We need to show that $\dt \in \slices^{-1}(\tup{t})$, i.e., that $\tup{t} \in \slices(\dt)$.
	We know that $\tup{t} = \pi_{\slices}(\slice) \in \slices(\bt)$, and hence for $(t_1, \ldots, t_\ell) = \tup{t}$, the entries $t_1, \ldots, t_\ell$ are pairwise distinct.
	Furthermore, we already know that $\ell = \ar(\tup{t}) = \ar(\slice)$ and $\stp(\slice, \ct) = \stp(\tup{t}, \dt)$.
	Since $\slice \in \slices(\ct)$, we know that for every $i \in [\ell]$ there is a $j_i$ such that $(i,j_i) \in \stp(\slice, \ct) = \stp(\tup{t}, \dt)$, i.e., $s_i = c_{j_i}$ and $t_i = d_{j_i}$ for all $i \in [\ell]$.
	Hence, $\tup{t} \in \slices(\dt)$.
	This completes the proof of the claim.
\end{proof}
\smallskip

So far, we have shown that $\pi'$ is an injective mapping $\pi'\colon \slices^{-1}(\slice) \to \slices^{-1}(\tup{t})$ such that $f(\ct) = f(\pi'(\ct))$ and $\stp(\slice,\ct) = \stp(\tup{t},\pi'(\ct))$ holds for all $\ct \in \slices^{-1}(\slice)$.
Hence, the proof of \enquote{\ref{statement1-equiv}~$\Longrightarrow$~\ref{statement2-equiv}} is complete as soon as we have proved the following claim.

\begin{claim*}
	For every $\dt \in \slices^{-1}(\tup{t})$ there exists a $\ct \in \slices^{-1}(\slice)$ such that $\pi'(\ct) = \dt$.
\end{claim*}
\begin{proof}[Proof of Claim]
	Let us fix an arbitrary $\dt \in \slices^{-1}(\tup{t})$. 
	From $\dt \in \slices^{-1}(\tup{t})$ we obtain that $\tup{t} \in \slices(\tup{d})$, and hence $\emptyset \neq \tset(\tup{t}) \subseteq \tset(\tup{d})$.
	On the other hand, $\tup{t} = \pi_{\slices}(\slice) \in \slices(\bt)$, and hence $\emptyset \neq \tset(\tup{t}) \subseteq\tset(\bt)$.
	Thus, $\tset(\bt) \intersect \tset(\tup{d}) \neq \emptyset$.
	Hence, by \cref{lem:slice-observations:stp-nonempty-iff-intersect}, $\stp(\bt,\tup{d}) \neq \emptyset$, and therefore $\tup{d} \in N(\bt)$.
	Let $\ct \isdef \pi^{-1}(\tup{d})$.
	Note that to complete the proof of the claim, it suffices to show that $\ct \in \slices^{-1}(\slice)$, i.e., that $\slice \in \slices(\ct)$.

	Let $\ell \isdef \ar(\slice)$, and recall that $\ell = \ar(\slice) = \ar(\tup{t})$ and $\stp(\at,\slice) = \stp(\bt,\tup{t})$.
	Consider an arbitrary $i \in [\ell]$.
	We have to show that there exists a $j$ such that $s_i = c_j$.
	Since $\slice \in \slices(\at)$, there exists an $\hati$ such that $s_i=a_{\hati}$, i.e., $(\hati,i) \in \stp(\at,\slice)$.
	From $\stp(\at,\slice) = \stp(\bt,\tup{t})$ we obtain that $(\hati,i) \in \stp(\bt,\tup{t})$, i.e., $b_{\hati} = t_i$.
	Furthermore, by assumption, we have $\tup{t} \in \slices(\tup{d})$, and hence there exists a $j$ such that $t_i = d_j$.
	From $b_{\hati} = t_i = d_j$ we obtain that $(\hati,j) \in \stp(\bt,\dt)$.
	We know that $\dt = \pi(\ct)$, and hence by the choice of $\pi$ we have $\stp(\at,\ct) = \stp(\bt,\dt)$.
	This implies that $(\hati,j) \in \stp(\at,\ct)$, i.e., $a_{\hati} = c_j$.
	In summary, we have $s_i = a_{\hati} = c_j$, i.e., $s_i = c_j$. This proves that $\slice\in \slices(\ct)$, i.e., that $\ct \in \slices^{-1}(\slice)$.
	This completes the proof of the claim.
\end{proof}
\smallskip

In summary, the proof of the direction \enquote{\ref{statement1-equiv}~$\Longrightarrow$~\ref{statement2-equiv}} of \cref{lem:vgrep-equiv-rcr} is now complete.
\bigskip

\enquote{\ref{statement2-equiv}~$\Longrightarrow$~\ref{statement1-equiv}}:\;
By assumption,  for every $\slice \in \slices(\at)$
we have \[
	\mset[\big]{ (\stp(\slice, \ct), \myc({\ct}))
		\mid
		\ct \in \slices^{-1}(\slice) }
	=
	\mset[\big]{ (\stp(\pi_{\slices}(\slice), \dt), \myc({\dt}))
		\mid
		\dt \in \slices^{-1}(\pi_{\slices}(\slice))
	}.
\]
Note that the proof of \enquote{\ref{statement2-equiv}~$\Longrightarrow$~\ref{statement1-equiv}} is complete as soon as we have found a bijection $\pi\colon N(\at) \to N(\bt)$ such that $\stp(\at,\ct) = \stp(\bt,\pi(\ct))$ and $f(\ct) = f(\pi(\ct))$ for all $\ct \in N(\at)$.
For constructing such a mapping $\pi$, we need the following notation.

\begin{enumerate}[label={\arabic*.}]
\item
	For any set $X\subseteq V(\A)$ let $\Neighbors{X} \isdef \set{ \ct' \in \tA \mid \tset(\ct') \supseteq X}$.
	\smallskip

	Note that for any slice $\slice \in \slices(\tA)$ and any tuple $\ct \in \tA$ we have:\; $\slice\in\slices(\ct) \iff \ct \in \Neighbors{\tset(\slice)}$. Thus, for any $\slice \in \slices(\tA)$ and for $X \isdef \tset(\slice)$ we have:\; $\slices^{-1}(\slice) = \Neighbors{X}$.
	\medskip

	\item
	For  $n \geq 1$ and $ \ct \in \tA$ let $\PotEq{n}{\ct} \isdef \set{ X \subseteq \tset(\ct) \mid \card{X} = n }$ and $\PotGeq{n}{\ct} \isdef \set{X \subseteq \tset(\ct) \mid \card{X} \geq n}$.

	Let $\NeighborsEq{n}{\ct} \isdef \set{ \ct' \in \tA \mid \card{ \tset(\ct')\intersect \tset(\ct)} = n}$ and $\NeighborsGeq{n}{\ct} \isdef \set{ \ct' \in \tA \mid \card{\tset(\ct') \intersect \tset(\ct)} \geq n}$.

	For $X \subseteq \V(\A)$ let $\NeighborsEq{X}{\ct} \isdef \set{ \ct' \in \tA \mid \tset(\ct') \intersect \tset(\ct) = X }$.
	\smallskip

	Note that for any $\ct \in \tA$ and for $k \isdef \card{\tset(\ct)}$, the sets $\NeighborsEq{n}{\ct}$ for $n \in [k]$ are pairwise disjoint, and we have:
	$N(\ct) = \NeighborsGeq{1}{\ct}$ and $\NeighborsGeq{n}{\ct} = \NeighborsEq{n}{\ct} \disunion \NeighborsGeq{n+1}{\ct}$, for any $n \geq 1$.
	
	Furthermore, for any $n \in [k]$, the sets $\NeighborsEq{X}{\ct}$ for $X \in \PotEq{n}{\ct}$ are pairwise disjoint, and we have:
	$\NeighborsEq{n}{\ct} = \bigdisunion_{X \in \PotEq{n}{\ct}} \NeighborsEq{X}{\ct}$.
	Moreover, for any $\ct \in \tA$ and any $X \subseteq \V(\A)$ we have $\NeighborsEq{X}{\ct} \subseteq \Neighbors{X}$.
	\medskip

	\item
	For $\ct \in \tA$ and $X \subseteq \V(\A)$ we let $\ov{\NeighborsEq{X}{\ct}} \isdef \Neighbors{X} \setminus \NeighborsEq{X}{\ct}$.
\end{enumerate}
It is easy to see that for any $\ct \in \tA$ and $X \subseteq \tset(\ct)$ we have:\;
$\ov{\NeighborsEq{X}{\ct}} = \setc{\ct' \in \tA}{ \tset(\ct') \intersect \tset(\ct) \varsupsetneqq X }$, and hence, in particular, $\ov{\NeighborsEq{X}{\ct}}\subseteq \NeighborsGeq{\card{X}+1}{\ct}$.

Recall that our ultimate goal is to construct a bijection $\pi\colon N(\at) \to N(\bt)$ such that $\stp(\at,\ct) = \stp(\bt,\pi(\ct))$ and $f(\ct) = f(\pi(\ct))$ for all $\ct \in N(\at)$.
We will proceed recursively and construct, for every $n$ with $\card{\tset(\at)} \geq n \geq 1$, a bijection $\pi_n\colon \NeighborsGeq{n}{\at} \to \NeighborsGeq{n}{\bt}$ such that $\stp(\at,\ct) = \stp(\bt,\pi_n(\ct))$ and $f(\ct) = f(\pi_n(\ct))$ for all $\ct \in \NeighborsGeq{n}{\at}$.
Note that, once having accomplished this, we are done by choosing $\pi \isdef \pi_1$.
\medskip

We start by constructing $\pi_n$ for $n \isdef \card{\tset(\at)}$.
Note that for $n \isdef \card{\tset(\at)}$ and $X \isdef \tset(\at)$ we have:
$\NeighborsGeq{n}{\at} = \NeighborsEq{n}{\at} = \NeighborsEq{X}{\at} = \Neighbors{X} = \slices^{-1}(\slice)$ where $\slice$ is an arbitrary slice of $\at$ that satisfies $\tset(\slice) = \tset(\at)$.
Let us fix such a slice $\slice$, and let $\tup{t} \isdef \pi_\slices(\slice)$.
From $\tset(\slice) = \tset(\at)$ and $\stp(\at,\slice) = \stp(\bt,\tup{t})$ we obtain that $\tset(\tup{t}) = \tset(\bt)$ and $n = \ar(\slice) = \ar(\tup{t}) = \card{\tset(\bt)}$.
Thus, for $Y \isdef \tset(\bt)$ we have:
$\NeighborsGeq{n}{\bt} = \NeighborsEq{n}{\bt} = \NeighborsEq{Y}{\bt} = \Neighbors{Y} = \slices^{-1}(\tup{t})$.

By assumption, \cref{statement2-equiv} of \cref{lem:vgrep-equiv-rcr} holds, and hence for our particular choice of $\slice$ and for $\tup{t} = \pi_\slices(\slice)$ there is a bijection $\beta\colon \slices^{-1}(\slice) \to \slices^{-1}(\tup{t})$ such that $\stp(\at,\ct) = \stp(\bt,\beta(\ct))$ and $\myc(\ct) = \myc(\beta(\ct))$ holds for all $\ct \in \slices^{-1}(\slice)$.
Since $\slices^{-1}(\slice) = \NeighborsGeq{n}{\at}$ and $\slices^{-1}(\tup{t})=\NeighborsGeq{n}{\bt}$, we are done by choosing $\pi_n \isdef \beta$.
\medskip

Let us now consider an arbitrary integer $n$ with $\card{\tset(\at)} > n \geq 1$, and let us assume that we have already constructed a bijection $\pi_{n+1}\colon \NeighborsGeq{n+1}{\at} \to \NeighborsGeq{n+1}{\bt}$ with the
desired properties.
Our goal is to construct a bijection $\pi_n\colon \NeighborsGeq{n}{\at} \to \NeighborsGeq{n}{\bt}$ with the desired properties.
Note that $\NeighborsGeq{n}{\at} = \NeighborsEq{n}{\at} \disunion \NeighborsGeq{n+1}{\at}$ and $\NeighborsGeq{n}{\bt} = \NeighborsEq{n}{\bt} \disunion \NeighborsGeq{n+1}{\bt}$.

In the following, we will construct a bijection $\pi'_n\colon \NeighborsEq{n}{\at} \to \NeighborsEq{n}{\bt}$ such that $\stp(\at,\ct) = \stp(\bt,\pi'_n(\ct))$ and $\myc(\ct) = \myc(\pi'_n(\ct))$ holds for all $\ct \in \NeighborsEq{n}{\at}$.
Note that, once having accomplished this, we are done by letting $\pi_n(\ct) \isdef \pi'_n(\ct)$ for any $\ct \in \NeighborsEq{n}{\at}$, and letting
$\pi_n(\ct) \isdef \pi_{n+1}(\ct)$ for any $\ct \in \NeighborsGeq{n+1}{\at}$.

Recall from the notations introduced above that $\NeighborsEq{n}{\at} = \bigdisunion_{X \in \PotEq{n}{\at}} \NeighborsEq{X}{\at}$ and $\NeighborsEq{n}{\bt} = \bigdisunion_{Y \in \PotEq{n}{\bt}} \NeighborsEq{Y}{\bt}$.
Throughout the remainder of this proof let us fix, for every $X \in \PotEq{n}{\at}$ (i.e., $X \subseteq \tset(\at)$ with $\card{X} = n$), a slice $\slice_X \in \slices(\at)$ with $\tset(\slice_X) = X$, and let $\tup{t}_X \isdef \pi_\slices(\slice_X)$ and $Y_X \isdef \tset(\tup{t}_X)$.

From \cref{lemma:easy-properties-of-pi-slices} we obtain that $\ar(\tup{t}_X) = \ar(\slice_X)$ (and hence $\card{Y_X} = \card{X} = n$) and, moreover, that for any $X, X' \in \PotEq{n}{\at}$ with $X \neq X'$ we have $Y_{X} \neq Y_{X'}$.
Thus, $\PotEq{n}{\bt} = \set{ Y_X \mid X \in \PotEq{n}{\at} }$ (note that \enquote{$\supseteq$} holds because $Y_X \subseteq \tset(\bt)$ and $\card{Y_X} = n$, and \enquote{$\subseteq$} holds because of the following reasoning:
From $\stp(\at) = \stp(\bt)$ we obtain that $\card{\tset(\bt)} = \card{\tset(\at)}$.
Thus, $\card{\PotEq{n}{\bt}} = \card{\PotEq{n}{\at}} = \binom{\card{\tset(\at)}}{n} = \card{\set{ Y_X \mid X \in \PotEq{n}{\at}}}$.).

Recall from the notations introduced above that $\slices^{-1}(\slice_X) = \Neighbors{X} = \NeighborsEq{X}{\at} \disunion \ov{\NeighborsEq{X}{\at}}$, and
$\ov{\NeighborsEq{X}{\at}} \subseteq \NeighborsGeq{n+1}{\at}$.
Analogously, we have: $\slices^{-1}(\tup{t}_X) = \Neighbors{Y_X} = \NeighborsEq{Y_X}{\bt} \disunion \ov{\NeighborsEq{Y_X}{\bt}}$, and $\ov{\NeighborsEq{Y_X}{\bt}} \subseteq \NeighborsGeq{n+1}{\bt}$.

Let us consider an arbitrary $X \in \PotEq{n}{\at}$.
In the following we will construct a bijection
\begin{equation}\label{eq:remaininggoal}
	\beta_X\colon \NeighborsEq{X}{\at} \to \NeighborsEq{Y_X}{\bt}
\end{equation}
such that $\stp(\at,\ct) = \stp(\bt,\beta_X(\ct))$ and $\myc(\ct) = \myc(\beta_X(\ct))$ holds for all $\ct \in \NeighborsEq{X}{\at}$.
Note that, once having accomplished this, we are done by letting $\pi'_n(\ct) \isdef \beta_{X_\ct}(\ct)$ for every $\ct \in \NeighborsEq{n}{\at}$ and the (uniquely defined) set $X_{\ct} \in \PotEq{n}{\at}$ satisfying $\ct \in \NeighborsEq{X_{\ct}}{\at}$.
Thus, all that remains to be done in order to complete the proof of direction \enquote{\ref{statement2-equiv}~$\Longrightarrow$~\ref{statement1-equiv}} of \cref{lem:vgrep-equiv-rcr} is to construct the mapping $\beta_X$ for $X \in \PotEq{n}{\at}$.

By assumption, \cref{statement2-equiv} holds for $\slice \isdef \slice_X$, i.e., for $\tup{t}_X = \pi_\slices(\slice_X)$ we have:
\begin{equation}\label{eq:finalstep}
	\mset[\big]{ (\stp(\slice_X, \ct), \myc({\ct})) \mid \ct \in \slices^{-1}(\slice_X)}
	\ \ = \ \
	\mset[\big]{ (\stp(\tup{t}_X, \dt), \myc({\dt})) \mid \dt \in \slices^{-1}(\tup{t}_X) }.
\end{equation}
We already know that $\slices^{-1}(\slice_X) = \NeighborsEq{X}{\at} \disunion \ov{\NeighborsEq{X}{\at}}$ and $\slices^{-1}(\tup{t}_X) = \NeighborsEq{Y_X}{\bt}\disunion \ov{\NeighborsEq{Y_X}{\bt}}$.
Furthermore, we know that $\ov{\NeighborsEq{X}{\at}} \subseteq \NeighborsGeq{n+1}{\at}$ and $\ov{\NeighborsEq{Y_X}{\bt}} \subseteq \NeighborsGeq{n+1}{\bt}$.

By our recursive construction, we have already defined a bijection $\pi_{n+1}\colon \NeighborsGeq{n+1}{\at} \to \NeighborsGeq{n+1}{\bt}$ such that $\stp(\at,\ct) = \stp(\bt,\pi_{n+1}(\ct))$ and $f(\ct) = f(\pi_{n+1}(\ct))$ for all $\ct \in \NeighborsGeq{n+1}{\at}$.

\begin{claim}\label{claim:XvsYX}
	For every $\ct \in \ov{\NeighborsEq{X}{\at}}$ we have: \ $\pi_{n+1}(\ct) \in \ov{\NeighborsEq{Y_X}{\bt}}$.
\end{claim}
\begin{proof}[Proof of Claim]
	When introducing the above notations, we had already noted that
	\[
		\ov{\NeighborsEq{X}{\at}}
		=
		\set[\big]{ \ct' \in \tA \mid \tset(\ct') \intersect \tset(\at) \varsupsetneqq X }
		\quad \text{and} \quad
		\ov{\NeighborsEq{Y_X}{\bt}}
		=
		\set[\big]{ \dt' \in \tA \mid \tset(\dt') \intersect \tset(\bt) \varsupsetneqq Y_X }.
	\]
	Consider an arbitrary $\ct \in \ov{\NeighborsEq{X}{\at}}$, and let $\dt \isdef \pi_{n+1}(\ct)$.
	We have to show that $\dt \in \ov{\NeighborsEq{Y_X}{\bt}}$.

	Let $\hat{X} \isdef \tset(\ct) \intersect \tset(\at)$.
	Note that $\hat{X} \varsupsetneqq X$.
	Let $\hat{\slice}$ be an arbitrary slice in $\slices(\at)$ with $\tset(\hat{\slice}) = \hat{X}$ (such a slice exists because $\hat{X} \subseteq \tset(\at)$).
	Let $\hat{\tup{t}} \isdef \pi_\slices(\hat{\slice})$.

	We know that $\tset(\slice_X) = X \varsubsetneqq \hat{X} = \tset(\hat{\slice})$.
	Thus, from \cref{lemma:easy-properties-of-pi-slices} we obtain that $\tset(\tup{t}_X) \varsubsetneqq \tset(\hat{\tup{t}})$.
	Recall that $Y_X = \tset(\tup{t}_X)$.
	In order to prove that $\dt \in \ov{\NeighborsEq{Y_X}{\bt}}$, it therefore suffices to prove that $\tset(\dt) \intersect \tset(\bt) = \tset(\hat{\tup{t}})$.

	We know that $\hat{\tup{t}} = \pi_\slices(\hat{\slice})$, and thus, $\stp(\at,\hat{\slice}) = \stp(\bt,\hat{\tup{t}})$.
	Furthermore, we have $\dt = \pi_{n+1}(\ct)$, and thus, $\stp(\at,\ct) = \stp(\bt,\dt)$.

	Let us first show that $\tset(\dt) \intersect \tset(\bt) \subseteq \tset(\hat{\tup{t}})$.
	Consider an arbitrary $z \in \tset(\dt) \intersect \tset(\bt)$. There exist $i, j$ such that $z = b_i = d_j$.
	Thus, $(i,j) \in \stp(\bt,\dt) = \stp(\at,\ct)$, and hence $a_i = c_j \in \tset(\at) \cap \tset(\ct) = \hat{X}$.
	We have $\tset(\hat{\slice}) = \hat{X}$.
	Thus, there is a $k$ such that $\hat{s}_k = a_i = c_j$.
	Hence, $(i,k) \in \stp(\at,\hat{\slice}) = \stp(\bt,\hat{\tup{t}})$, i.e., $b_i = \hat{t}_k$.
	In summary, $z = b_i \in \tset(\hat{\tup{t}})$.
	We have thus shown that $\tset(\dt) \intersect \tset(\bt) \subseteq \tset(\hat{\tup{t}})$.

	Let us now show that $\tset(\dt) \cap \tset(\bt) \supseteq \tset(\hat{\tup{t}})$.
	Consider an arbitrary $k \in [\ar(\hat{\tup{t}})]$.
	We have to show that $\hat{t}_k \in \tset(\dt) \cap \tset(\bt)$.
	From \cref{lemma:easy-properties-of-pi-slices} we know that $\ar(\hat{\tup{t}}) = \ar(\hat{\slice})$.
	From $\tset(\at) \cap \tset(\ct) = \hat{X} = \tset(\hat{\slice})$ we obtain that there exist $i, j$ such that $\hat{s}_k = a_i = c_j$.
	Thus, $(i,k) \in \stp(\at,\hat{\slice}) = \stp(\bt,\hat{\tup{t}})$, i.e., $b_i = \hat{t}_k$.
	Furthermore, $(i,j) \in \stp(\at,\ct) = \stp(\bt,\dt)$, i.e., $b_i = d_j$.
	Hence, $\hat{t}_k = b_i = d_j$, and thus $\hat{t}_k \in \tset(\bt) \cap \tset(\dt)$.
	We have thus shown that $\tset(\dt) \cap \tset(\bt) \supseteq \tset(\hat{\tup{t}})$.
	This completes the proof of \cref{claim:XvsYX}.
\end{proof}  

By an analogous reasoning we obtain:

\begin{claim}\label{claim:YXvsX}
	For every $\dt \in \ov{\NeighborsEq{Y_X}{\bt}}$ we have: \ $\pi_{n+1}^{-1}(\dt) \in \ov{\NeighborsEq{X}{\at}}$.
\end{claim}
\begin{proof}[Proof of Claim]
	Recall that
	\[  
		\ov{\NeighborsEq{X}{\at}}
		=
		\set[\big]{ \ct' \in \tA \mid \tset(\ct') \cap \tset(\at) \varsupsetneqq X }
		\quad \text{and} \quad
		\ov{\NeighborsEq{Y_X}{\bt}}
		=
		\set[\big]{ \dt' \in \tA \mid \tset(\dt') \cap \tset(\bt) \varsupsetneqq Y_X}.
	\]
	Consider an arbitrary $\dt \in \ov{\NeighborsEq{Y_X}{\bt}}$, and let $\ct \isdef \pi_{n+1}^{-1}(\dt)$.
	We have to show that $\ct \in \ov{\NeighborsEq{X}{\at}}$.

	Let $\hat{Y} \isdef \tset(\dt) \cap \tset(\bt)$.
	Note that $\hat{Y} \varsupsetneqq Y_X$.
	Let $\hat{\tup{t}}$ be an arbitrary slice in $\slices(\bt)$ with $\tset(\hat{\tup{t}}) = \hat{Y}$ (such a slice exists because $\hat{Y} \subseteq \tset(\bt)$).
	Let $\hat{\slice} \isdef \pi_\slices^{-1}(\hat{\tup{t}})$.

	We know that $\tset(\tup{t}_X) = Y_X \varsubsetneqq \hat{Y} = \tset(\hat{\tup{t}})$.
	Thus, from \cref{lemma:easy-properties-of-pi-slices} we obtain that $\tset(\slice_X) \varsubsetneqq \tset(\hat{\slice})$.
	Recall that $X = \tset(\slice_X)$.
	In order to prove that $\ct \in \ov{\NeighborsEq{X}{\at}}$, it therefore suffices to prove that $\tset(\ct) \cap \tset(\at) = \tset(\hat{\slice})$.

	We know that $\hat{\tup{t}} = \pi_\slices(\hat{\slice})$, and thus, $\stp(\at,\hat{\slice}) = \stp(\bt,\hat{\tup{t}})$.
	Furthermore, we have $\dt = \pi_{n+1}(\ct)$, and thus, $\stp(\at,\ct) = \stp(\bt,\dt)$.

	Let us first show that $\tset(\ct) \cap \tset(\at) \subseteq \tset(\hat{\slice})$.
	Consider an arbitrary $z \in \tset(\ct) \cap \tset(\at)$.
	There exist $i,j$ such that $z = a_i = c_j$.
	Thus, $(i,j) \in \stp(\at,\ct) = \stp(\bt,\dt)$, and hence $b_i = d_j \in \tset(\bt) \cap \tset(\dt) = \hat{Y}$.
	We have $\tset(\hat{\tup{t}}) = \hat{Y}$.
	Thus, there is a $k$ such that $\hat{t}_k = b_i = d_j$.
	Hence, $(i,k) \in \stp(\bt,\hat{\tup{t}}) = \stp(\at,\hat{\slice})$, i.e., $a_i = \hat{s}_k$.
	In summary, $z = a_i \in \tset(\hat{\slice})$.
	We have thus shown that $\tset(\ct) \cap \tset(\at) \subseteq \tset(\hat{\slice})$.

	Let us now show that $\tset(\ct) \cap \tset(\at) \supseteq \tset(\hat{\slice})$.
	Consider an arbitrary $k \in [\ar(\hat{\slice})]$.
	We have to show that $\hat{s}_k \in \tset(\ct) \cap \tset(\at)$.
	From \cref{lemma:easy-properties-of-pi-slices} we know that $\ar(\hat{\tup{t}}) = \ar(\hat{\slice})$.
	From $\tset(\bt) \cap \tset(\dt) = \hat{Y} = \tset(\hat{\tup{t}})$ we obtain that there exist $i,j$ such that $\hat{t}_k = b_i = d_j$.
	Thus, $(i,k) \in \stp(\bt,\hat{\tup{t}}) = \stp(\at,\hat{\slice})$, i.e.,
	$a_i = \hat{s}_k$.
	Furthermore, $(i,j) \in \stp(\bt,\dt) = \stp(\at,\ct)$, i.e., $a_i = c_j$.
	Hence, $\hat{s}_k = a_i = c_j$, and thus $\hat{s}_k \in \tset(\at) \cap \tset(\ct)$.
	We have thus shown that $\tset(\at) \cap \tset(\ct) \supseteq \tset(\hat{\slice})$.
	This completes the proof of \cref{claim:YXvsX}.
\end{proof}  

\begin{claim}\label{claim:XvsYX-stp}
	For every $\ct \in \ov{\NeighborsEq{X}{\at}}$ we have: \ 
	$\stp(\slice_X,\ct) = \stp(\tup{t}_X,\pi_{n+1}(\ct))$.
\end{claim}
\begin{proof}[Proof of Claim]
	Let $\slice \isdef \slice_X$ and $\tup{t} \isdef \tup{t}_X$.
	We know that $\tup{t} = \pi_\slices(\slice)$, and thus $\stp(\at,\slice) = \stp(\bt,\tup{t})$.
	Furthermore, $X = \tset(\slice)$ and $Y_X = \tset(\tup{t})$.
	Consider an arbitrary $\ct \in \ov{\NeighborsEq{X}{\at}}$ and let $\dt \isdef \pi_{n+1}(\ct)$.
	We know that $\stp(\at,\ct) = \stp(\bt,\dt)$.
	We have to show that $\stp(\slice,\ct) = \stp(\tup{t},\dt)$.

	For proving \enquote{$\subseteq$}, consider arbitrary $(i,j) \in \stp(\slice,\ct)$, i.e., $s_i = c_j$.
	From $s_i \in X \subseteq \tset(\at)$ we obtain that there exists an $\hati$ such that $s_i = a_\hati$.
	Thus, $(\hati,i) \in \stp(\at,\slice) = \stp(\bt,\tup{t})$.
	Thus, $t_i = b_{\hati}$.
	Furthermore, from $c_j = s_i = a_\hati$ we obtain that $(\hati,j) \in \stp(\at,\ct) = \stp(\bt,\dt)$.
	Thus, $t_i = b_\hati = d_j$.
	Hence, $(i,j) \in \stp(\tup{t},\dt)$.
	This completes the proof of \enquote{$\subseteq$}.

	For proving \enquote{$\supseteq$}, consider arbitrary $(i,j) \in \stp(\tup{t},\dt)$, i.e., $t_i = d_j$.
	From $t_i \in Y_X \subseteq \tset(\bt)$ we obtain that there exists an $\hati$ such that $t_i = b_\hati$.
	Thus, $(\hati,i) \in \stp(\bt,\tup{t}) = \stp(\at,\slice)$.
	Thus, $s_i = a_{\hati}$.
	Furthermore, from $d_j = t_i = b_\hati$ we obtain that $(\hati,j) \in \stp(\bt,\dt) = \stp(\at,\ct)$.
	Thus, $s_i = a_\hati = c_j$.
	Hence, $(i,j) \in \stp(\slice,\ct)$.
	This completes the proof of \enquote{$\supseteq$} and the proof of \cref{claim:XvsYX-stp}.
\end{proof}

From \cref{claim:XvsYX}, \cref{claim:YXvsX}, and \cref{claim:XvsYX-stp} we obtain that restricting $\pi_{n+1}$ to $\ov{\NeighborsEq{X}{\at}}$ yields a bijection $\ov{\pi}\colon \ov{\NeighborsEq{X}{\at}} \to \ov{\NeighborsEq{Y_X}{\bt}}$ such that $\stp(\slice_X,\ct) = \stp(\tup{t}_X,\ov{\pi}(\ct))$ and $\myc(\ct) = \myc(\ov{\pi}(\ct))$ holds for all $\ct \in \ov{\NeighborsEq{X}{\at}}$
(to achieve this, we simply let $\ov{\pi}(\ct) \isdef \pi_{n+1}(\ct)$ for all $\ct \in \ov{\NeighborsEq{X}{\at}}$).

Recall from equation~\eqref{eq:finalstep} that $\slices^{-1}(\slice_X) = \NeighborsEq{X}{\at} \disunion \ov{\NeighborsEq{X}{\at}}$ and $\slices^{-1}(\tup{t}_X) = \NeighborsEq{Y_X}{\bt} \disunion \ov{\NeighborsEq{Y_X}{\bt}}$ and
\[
	\mset{ (\stp(\slice_X, \ct), \myc({\ct})) \mid \ct \in \slices^{-1}(\slice_X)}
	\ \ = \ \ 
	\mset{ (\stp(\tup{t}_X, \dt), \myc({\dt})) \mid \dt \in \slices^{-1}(\tup{t}_X) }.
\]
For every $\ct \in \ov{\NeighborsEq{X}{\at}}$ we remove the tuple $(\stp(\slice_X, \ct), \myc({\ct}))$ from the left multiset, and we remove the tuple $(\stp(\tup{t}_X,  \ov{\pi}(\ct)), \myc(\ov{\pi}(\ct)))$ from the right multiset.
The resulting multisets are equal, i.e., we have:
\[
	\mset{ (\stp(\slice_X, \ct), \myc({\ct})) \mid \ct \in \NeighborsEq{X}{\at}}
	\ \ = \ \ 
	\mset{ (\stp(\tup{t}_X, \dt), \myc({\dt})) \mid \dt \in \NeighborsEq{Y_X}{\bt} }.
\]
Hence, there is a bijection $\beta_X\colon \NeighborsEq{X}{\at} \to \NeighborsEq{Y_X}{\bt}$ such that $\stp(\slice_X, \ct) = \stp(\tup{t}_X, \beta_X(\ct))$ and
$\myc(\ct) = \myc(\beta_X(\ct))$ holds for all $\ct \in \NeighborsEq{X}{\at}$.

\begin{claim}\label{claim:finalstep}
	For every $\ct \in \NeighborsEq{X}{\at}$ we have: \
	$\stp(\at,\ct) = \stp(\bt,\beta_X(\ct))$.
\end{claim}
\begin{proof}[Proof of Claim]
	Let $\slice \isdef \slice_X$ and $\tup{t} \isdef \tup{t}_X$.
	We have $\tset(\slice) = X$ and $\tset(\tup{t}) = Y_X$.
	Furthermore, from $\tup{t} = \pi_\slices(\slice)$ we obtain that $\stp(\slice,\at) = \stp(\tup{t},\bt)$.
	Consider an arbitrary $\ct \in \NeighborsEq{X}{\at}$ and let $\dt \isdef \beta_X(\ct)$.
	We know that $\stp(\slice,\ct) = \stp(\tup{t},\dt)$.
	We have to show that $\stp(\at,\ct) = \stp(\bt,\dt)$.

	For proving \enquote{$\subseteq$} consider arbitrary $(i,j) \in \stp(\at,\ct)$.
	I.e., $a_i=c_j$.
	Hence, $a_i = c_j \in \tset(\at) \cap \tset(\ct) = X$.
	Since $X = \tset(\slice)$, there is an $\hati$ such that $a_i = c_j = s_\hati$.
	Hence, $(\hati,i) \in \stp(\slice,\at) = \stp(\tup{t},\bt)$, i.e., $t_\hati = b_i$.
	Furthermore, $(\hati,j) \in \stp(\slice,\ct) = \stp(\tup{t},\dt)$, i.e., $t_\hati = d_j$. Hence, $d_j = t_\hati = b_i$. Therefore, $(i,j) \in \stp(\bt,\dt)$.
	This completes the proof of \enquote{$\subseteq$}.

	For proving \enquote{$\supseteq$} consider arbitrary $(i,j) \in \stp(\bt,\dt)$.
	I.e., $b_i = d_j$.
	Hence, $b_i = d_j \in \tset(\bt) \cap \tset(\dt) = Y_X$.
	Since $Y_X = \tset(\tup{t})$, there is an $\hati$ such that $b_i = d_j = t_\hati$.
	Hence, $(\hati,i) \in \stp(\tup{t},\bt) = \stp(\slice,\at)$, i.e.,
	$s_\hati = a_i$.
	Furthermore, $(\hati,j) \in \stp(\tup{t},\dt) = \stp(\slice,\ct)$, i.e.,
	$s_\hati = c_j$.
	Hence, $c_j = s_\hati = a_i$.
	Therefore, $(i,j) \in \stp(\at,\ct)$.
	This completes the proof of \enquote{$\supseteq$} and the proof of \cref{claim:finalstep}.
\end{proof}  

From \cref{claim:finalstep} we obtain that $\beta_X$ indeed has all the properties stated in the text around equation~\eqref{eq:remaininggoal}, i.e., $\beta_X$ is a bijection from $\NeighborsEq{X}{\at}$ to $\NeighborsEq{Y_X}{\bt}$ such that $\stp(\at, \ct) = \stp(\bt, \beta_X(\ct))$ and $\myc(\ct) =\myc(\beta_X(\ct))$ holds for all $\ct\in \NeighborsEq{X}{\at}$.
This completes the proof of direction \enquote{\ref{statement2-equiv}~$\Longrightarrow$~\ref{statement1-equiv}} of \cref{lem:vgrep-equiv-rcr}.

In summary, the proof of \cref{lem:vgrep-equiv-rcr} is complete.
\end{proof} 

\subsection{Proof of Lemma~\ref{lem:slices-determine-cr}}\label[appendix]{appendix:slices-determine-cr}
\slicesDetermineCR*
\begin{proof}
	Let $N_{\at}$ and $N_{\bt}$ be the neighbors of $w_{\at}$ and $w_{\bt}$ in $\vgrep{\A}$. By definition of CR we have that $\gcr{i}{w_{\at}} = \gcr{i}{w_{\bt}} \iff$
	\begin{enumerate}
		\item $\gcr{i-1}{w_{\at}} = \gcr{i-1}{w_{\bt}}$\; and
		\item $\mset{ (\lambda(w_{\at}, w), \gcr{i-1}{w}) \mid w \in N_{\at} } = \mset{ (\lambda(w_{\bt}, w), \gcr{i-1}{w}) \mid w \in N_{\bt} }$.
	\end{enumerate}
	Using \cref{fact:easy-rcr-cr:neighbors-of-w} we get that the second statement is equivalent to
	\begin{equation*}
		\mset{ (\lambda(w_{\at}, v_{\slice}), \gcr{i-1}{v_{\slice}}) \mid \slice \in \slices(\at) }
		=
		\mset{ (\lambda(w_{\bt}, v_{\slice}), \gcr{i-1}{v_{\slice}}) \mid \slice \in \slices(\bt) }\;.
	\end{equation*}
	This is according to \cref{fact:easy-rcr-cr:lambda-tuple-stp} equivalent to
	\begin{equation*}
		\mset{ (\stp(\at, \slice), \gcr{i-1}{v_{\slice}}) \mid \slice \in \slices(\at) }
		=
		\mset{ (\stp(\bt, \slice), \gcr{i-1}{v_{\slice}}) \mid \slice \in \slices(\bt) }\;.
	\end{equation*}
	Consider the bijection $\pi_\slices\colon \slices(\at) \to \slices(\bt)$ satisfying $\stp(\at,\slice) = \stp(\bt,\pi_\slices(\slice))$ for all $\slice \in \slices(\at)$, that we obtain from \cref{lem:self-stp-identifies-stp-of-slices}.
	From \cref{remark:neighbors-in-vgrep} we know that $\stp(\at,\slice) \neq \stp(\at,\slice')$ for all $\slice, \slice' \in \slices(\at)$, and that $\stp(\bt,\tup{t}) \neq \stp(\bt,\tup{t}')$ for all $\tup{t}, \tup{t}' \in \slices(\bt)$.
	Thus, $\pi_\slices$ is the \emph{only} bijection $\pi\colon \slices(\at) \to \slices(\bt)$ satisfying $\stp(\at,\slice) = \stp(\bt,\pi(\slice))$ for all $\slice \in \slices(\at)$, and we get that $\gcr{i}{w_{\at}} = \gcr{i}{w_{\bt}} \iff$
	\begin{enumerate}
		\item $\gcr{i-1}{w_{\at}} = \gcr{i-1}{w_{\bt}}$\; and
		\item $\gcr{i-1}{v_{\slice}} = \gcr{i-1}{v_{\pi_\slices(\slice)}}$, for all $\slice \in \slices(\at)$.
	\end{enumerate}
	This finishes the proof.
\end{proof}

\subsection{Proof of Claims~\ref{claim:intermediate-slice-cr-implied}~and~\ref{claim:rcr-iff-slice-cr}}\label[appendix]{appendix:main-runtime-thm-claims}

\intermediateSliceCR*%
\begin{proof}
For a slice $\slice$, let $N_{\slice}$ denote the set of neighbors of $v_{\slice}$ in $\vgrep{\A}$ in the following.\\
Assume that $\col{i}{\at} = \col{i}{\bt}$. Then, the following must hold:
\begin{equation*}
	\mset{ (\stp(\at, \ct), \col{i-1}{\ct}) \mid \ct \in N(\at) }
	=
	\mset{ (\stp(\bt, \ct), \col{i-1}{\ct}) \mid \ct \in N(\bt) }.
\end{equation*}
Thus, \cref{lem:vgrep-equiv-rcr} yields that for all $\slice \in \slices(\at)$ and $\tup{t} \isdef \pi_{\slices}(\slice)$ we have that
\begin{equation}\label{claim:intermediate-slice:a}
	\mset{ (\stp(\slice, \ct), \col{i-1}{\ct}) \mid \ct \in \slices^{-1}(\slice) }
	=
	\mset{ (\stp(\tup{t}, \ct), \col{i-1}{\ct}) \mid \ct \in \slices^{-1}(\tup{t}) }.
\end{equation}
The induction hypothesis yields that for all $\ct, \tup{d} \in \tA$ we have $\col{i-1}{\ct} = \col{i-1}{\tup{d}}$ if, and only if, $\gcr{2(i-1)+1}{w_{\ct}} = \gcr{2(i-1)+1}{w_{\tup{d}}}$. Combined with the fact that $2(i-1)+1 = 2i-1$ it follows that~\eqref{claim:intermediate-slice:a} holds if, and only if
\begin{equation*}
	\mset{ (\stp(\slice, \ct), \gcr{2i-1}{w_{\ct}}) \mid \ct \in \slices^{-1}(\slice) }
	=
	\mset{ (\stp(\tup{t}, \ct), \gcr{2i-1}{w_{\ct}}) \mid \ct \in \slices^{-1}(\tup{t}) },
\end{equation*}
which, because of $N_{\slice} = \set{ w_{\ct} \mid \ct \in \slices^{-1}(\slice) }$ and $N_{\tup{t}} = \set{ w_{\ct} \mid \ct \in \slices^{-1}(\tup{t}) }$ (follows from \cref{fact:easy-rcr-cr:neighbors-of-w}), holds if, and only if
\begin{equation}\label{claim:intermediate-slice:b}
	\mset{ (\stp(\slice, \ct), \gcr{2i-1}{w_{\ct}}) \mid w_{\ct} \in N_{\slice} }
	=
	\mset{ (\stp(\tup{t}, \ct), \gcr{2i-1}{w_{\ct}}) \mid w_{\ct} \in N_{\tup{t}} }.
\end{equation}
By \cref{fact:easy-rcr-cr:lambda-slice-stp},~\eqref{claim:intermediate-slice:b} holds if, and only if,
\begin{equation}\label{claim:intermediate-slice:c}
	\mset{ (\lambda(v_{\slice}, w_{\ct}), \gcr{2i-1}{w_{\ct}}) \mid w_{\ct} \in N_{\slice} }
	=
	\mset{ (\lambda(v_{\tup{t}}, w_{\ct}), \gcr{2i-1}{w_{\ct}}) \mid w_{\ct} \in N_{\tup{t}} }.
\end{equation}
Statement~\eqref{claim:intermediate-slice:c} holds if, and only if,
\begin{equation}\label{claim:intermediate-slice:d}
	\mset{ (\lambda(v_{\slice}, w), \gcr{2i-1}{w}) \mid w \in N_{\slice} }
	=
	\mset{ (\lambda(v_{\tup{t}}, w), \gcr{2i-1}{w}) \mid w \in N_{\tup{t}} }.
\end{equation}

In total, we showed that $\col{i}{\at} = \col{i}{\bt}$ implies that~\eqref{claim:intermediate-slice:d} holds for all $\slice \in \slices(\at)$ and $\tup{t} \isdef \pi_{\slices}(\slice)$. By definition, we know that for all $\slice \in \slices(\at)$ and $\tup{t} \isdef \pi_{\slices}(\slice)$ it holds that $\gcr{2i}{v_{\slice}} = \gcr{2i}{v_{\tup{t}}}$ if, and only if, $\gcr{2i-1}{v_{\slice}} = \gcr{2i-1}{v_{\tup{t}}}$ and~\eqref{claim:intermediate-slice:d}.
Thus, for all $\slice \in \slices(\at)$ and $\tup{t} \isdef \pi_{\slices}(\slice)$ it holds that:\quad
	If\; $\col{i}{\at} = \col{i}{\bt}$\; and \;$\gcr{2i-1}{v_{\slice}} = \gcr{2i-1}{v_{\tup{t}}}$, then\; $\gcr{2i}{v_{\slice}} = \gcr{2i}{v_{\tup{t}}}$. \end{proof}
\medskip

\RCRiffSliceCR*%
\begin{proof}
For a slice $\slice$, let $N_{\slice}$ denote the set of neighbors of $v_{\slice}$ in $\vgrep{\A}$ in the following.\\
Consider $\slice \in \slices(\at)$ and let $\tup{t} \isdef \pi_{\slices}(\slice)$. By definition of CR, $\gcr{2i-1}{v_{\slice}} = \gcr{2i-1}{v_{\tup{t}}}$ holds if, and only if $\gcr{2i-2}{v_{\slice}} = \gcr{2i-2}{v_{\tup{t}}}$ and
\begin{equation}\label{slice-claim:stmt:a}
	\mset{ (\lambda(v_{\slice}, w), \gcr{2i-2}{w}) \mid w \in N_{\slice} }
	=
	\mset{ (\lambda(v_{\tup{t}}, w), \gcr{2i-2}{w}) \mid w \in N_{\tup{t}} }.
\end{equation}
Since $N_{\slice} = \set{ w_{\ct} \mid \ct \in \slices^{-1}(\slice) }$ and $N_{\tup{t}} = \set{ w_{\ct} \mid \ct \in \slices^{-1}(\tup{t}) }$ (follows from \cref{fact:easy-rcr-cr:neighbors-of-w}),~\eqref{slice-claim:stmt:a} holds if, and only if, the following holds:
\begin{equation}\label{slice-claim:stmt:b}
	\mset{ (\lambda(v_{\slice}, w_{\ct}), \gcr{2i-2}{w_{\ct}}) \mid \ct \in \slices^{-1}(\slice) }
	=
	\mset{ (\lambda(v_{\tup{t}}, w_{\ct}), \gcr{2i-2}{w_{\ct}}) \mid \ct \in \slices^{-1}(\tup{t}) }.
\end{equation}
By \cref{fact:easy-rcr-cr:lambda-slice-stp},~\eqref{slice-claim:stmt:b} holds if, and only if, the following holds:
\begin{equation}\label{slice-claim:stmt:c}
	\mset{ (\stp(\slice, \ct), \gcr{2i-2}{w_{\ct}}) \mid \ct \in \slices^{-1}(\slice) }
	=
	\mset{ (\stp(\tup{t}, \ct), \gcr{2i-2}{w_{\ct}}) \mid \ct \in \slices^{-1}(\tup{t}) }.
\end{equation}
Since $2i-2 = 2(i-1)$, the induction hypothesis yields that~\eqref{slice-claim:stmt:c} holds if, and only if, the following holds:
\begin{equation}\label{slice-claim:stmt:d}
	\mset{ (\stp(\slice, \ct), \col{i-1}{\ct}) \mid \ct \in \slices^{-1}(\slice) }
	=
	\mset{ (\stp(\tup{t}, \ct), \col{i-1}{\ct}) \mid \ct \in \slices^{-1}(\tup{t}) }.
\end{equation}
Thus, for every $\slice \in \slices(\at)$ we have that $\gcr{2i-1}{v_{\slice}} = \gcr{2i-1}{v_{\pi_{\slices}(\slice)}}$ if, and only if, $\gcr{2i-1}{v_{\slice}} = \gcr{2i-1}{v_{\pi_{\slices}(\slice)}}$ and
\begin{equation}\label{slice-claim:stmt:e}
	\mset{ (\stp(\slice, \ct), \col{i-1}{\ct}) \mid \ct \in \slices^{-1}(\slice) }
	=
	\mset{ (\stp(\pi_{\slices}(\slice), \ct), \col{i-1}{\ct}) \mid \ct \in \slices^{-1}(\pi_{\slices}(\slice)) }.
\end{equation}
\Cref{lem:vgrep-equiv-rcr} yields that~\eqref{slice-claim:stmt:e} holds for all $\slice \in \slices(\at)$ if, and only if, the following holds:
\begin{equation}\label{slice-claim:stmt:f}
	\mset{ (\stp(\at, \ct), \col{i-1}{\ct}) \mid \ct \in N(\at) }
	=
	\mset{ (\stp(\bt, \ct), \col{i-1}{\ct}) \mid \ct \in N(\bt) }.
\end{equation}
Since $\col{i-1}{\at} = \col{i-1}{\bt}$,~\eqref{slice-claim:stmt:f} holds if, and only if, $\col{i}{\at} = \col{i}{\bt}$.

Hence, we showed that for all $\slice \in \slices(\at)$ and $\tup{t} \isdef \pi_{\slices}(\slice)$ it holds that $\gcr{2i-1}{v_{\slice}} = \gcr{2i-1}{v_{\tup{t}}}$ if, and only if, $\gcr{2i-2}{v_{\slice}} = \gcr{2i-2}{v_{\tup{t}}}$ and $\col{i}{\at} = \col{i}{\bt}$. Since $\gcr{2i-2}{v_{\slice}} \neq \gcr{2i-2}{v_{\tup{t}}}$ would contradict $\gcr{2i-1}{w_{\at}} = \gcr{2i-1}{w_{\bt}}$ according to \cref{lem:slices-determine-cr}, we know that $\gcr{2i-2}{v_{\slice}} = \gcr{2i-2}{v_{\tup{t}}}$ holds for all $\slice \in \slices(\at)$ and $\tup{t} \isdef \pi_{\slices}(\slice)$.

In total, we showed that for all $\slice \in \slices(\at)$ and $\tup{t} \isdef \pi_{\slices}(\slice)$ we have\; $\gcr{2i-1}{v_{\slice}} = \gcr{2i-1}{v_{\tup{t}}} \iff \col{i}{\at} = \col{i}{\bt}$. \end{proof} 

 \end{document}